%% file: noms.tex
\documentclass[10pt,conference]{IEEEtran}

\usepackage[utf8]{inputenc}
\usepackage{todonotes}
\usepackage{caption}
\usepackage{graphicx}
\usepackage{mathtools}

\usepackage{epstopdf}
\usepackage{amsthm}
\usepackage{amsmath}
\usepackage{cleveref}
\usepackage{flushend}
\usepackage{subcaption}
\usepackage{color}

\newtheorem{proposition}{Proposition}
\newtheorem{claim}{Claim}
\usepackage[final]{pdfpages}
\usepackage{algorithm,algorithmicx}
\usepackage[noend]{algpseudocode}

\usepackage[mode=buildnew]{standalone}
\usepackage{pgfplots}
\usepackage{tikz}
\usepgfplotslibrary{groupplots,dateplot}
\usetikzlibrary{patterns,shapes.arrows}

\title{Virtual Multi-Topology Routing for QoS Constraints}

\author{Nicolas Huin\IEEEauthorrefmark{1},  Sébastien Martin\IEEEauthorrefmark{2}, Jérémie Leguay\IEEEauthorrefmark{2}\\
    \IEEEauthorblockA{\IEEEauthorrefmark{1}\textit{IMT Atlantique, IRISA UMR CNRS 6074, F-35700 Rennes, France}}
    \IEEEauthorblockA{\IEEEauthorrefmark{2}\textit{Huawei Technologies Ltd., Paris Research Center, France.}}
}

\usepackage{notations}

\begin{document}

\maketitle

\begin{abstract}
Multi-topology routing (MTR) provides an attractive alternative to segment routing for traffic engineering when network devices cannot be upgraded. However, due to a high overhead in terms of link state messages exchanged by topologies and the need to frequently update link weights to follow evolving network conditions, MTR is often limited to a small number of topologies and the satisfaction of loose QoS constraints. To overcome these limitations we propose vMTR, an MTR extension where demands are routed over virtual topologies that are silent, i.e., they do not exchange LSA messages, and that are continuously derived from a very limited set of real topologies, optimizing each a QoS parameter. In this context, we present a polynomial and exact algorithm for vMTR and, as a benchmark, a local search algorithm for MTR. We show that vMTR helps reducing drastically the number of real topologies and that it is more robust to QoS changes.
\end{abstract}


\section{Introduction}

Traffic Engineering (TE) policies~\cite{Wang:TE} are key to controlling how bandwidth is shared in IP (Internet Protocol) networks.
These policies can be tuned to satisfy application requirements in terms of end-to-end delay, packet loss, and jitter, for instance.
They can be optimized by a centralized controller and applied by routers to steer traffic with data plane mechanisms such as Segment Routing (SR)~\cite{ventre2020segment, wu2022comprehensive} or MPLS (Multi-Protocol Label Switching).
However, most IP networks today still rely on shortest path and hop-by-hop routing from Interior Gateway Protocols (IGP).

Extensions of IGP protocols such as Multi-Topology Routing (MTR)~\cite{psenak2007multi, przygienda2008m} or FlexAlgo~\cite{Flexalgo} provide simple solutions to create and operate multiple routing planes, or "soft" slices~\cite{huin2021network}, to support various QoS and operational requirements.
Each routing plane minimizes a different additive link metric (e.g., IGP cost, TE metric, QoS) and can exclude some links when needed (i.e., \textit{constraints} in FlexAlgo or \textit{exclusions} in MTR).
In theory, MTR and FlexAlgo only require a software upgrade of devices to update the control plane, and they can still rely on standard IP forwarding in the data plane.
However, in practice, this is only the case for MTR, as FlexAlgo is always shipped along with SRv6 or SR-MPLS\@.
Indeed, FlexAlgo's implementations rely on SR labels to identify routing planes and forward traffic consistently on top of them.
When SR is available, TE policies can be configured at ingress nodes to provide source-based routing on top of existing IGP routing planes (or ``algos'' when FlexAlgo is rolled out).
SR policies consist of a label stack that is injected in packet headers to encode the sequence of shortest-path segments that need to be traversed by packets.
In this case, a data plane upgrade, i.e., device replacement or microcode update whenever possible, is necessary to support SR forwarding. Indeed, many legacy devices cannot support SR, as mentioned in~\cite{tian2020traffic}.

In this paper, we consider deployment scenarios where most devices still operate with legacy data planes (i.e., without SR support).
In this context, IGP extensions like MTR can be used to create multiple routing planes, but most routers have not been replaced by SR-capable devices.
The most popular IGP protocol is Open Shortest Path First (OSPF)~\cite{moy1998ospf}.
In this protocol, routers periodically broadcast Link State Advertisement (LSA) messages that contain metrics related to their adjacent links.
Upon reception of LSA messages, routers store them in an LSDB (Link State DataBase) and maintain a Shortest Path Tree (SPT) towards all the other nodes, which decides the next hop to use for all possible destinations.
This information is stored in the data plane in a Forwarding Information Base (FIB) and used to forward packets.
When only one link metric is used (e.g., IGP cost), a single routing plane is available.
The set of link metrics is generally \textit{designed}~\cite{fortz2000internet} to satisfy some traffic scenarios.
However, with a single routing plane for all applications, it is, in general, not possible to satisfy a heterogeneous set of QoS requirements.

To increase routing options, MTR offers the possibility of executing multiple IGP instances in parallel, each working over a different metric.
The SPT maintained for each instance, or topology, generally aims at satisfying a set of QoS requirements (e.g., cost, hop count, delay) and administrative constraints (e.g., exclusion of links or nodes).
FlexAlgo~\cite{Flexalgo} can be seen as an extension of MTR and adds the possibility to build a routing plane over 1) a given metric, 2) considering some \textit{constraints} (i.e., exclusions of nodes/links), and 3) a particular shortest path algorithm.
Each \textit{algo} is defined by the tuple (metric, constraints, SPT algorithm).
Each MTR topology, or ``algo'' in FlexAlgo, has its own IP prefix so that traffic can be forwarded consistently in a hop-by-hop fashion. In both cases, link metrics also need to be periodically advertised using LSAs so that routers can update their TE database.
To meet a variety of QoS constraints from applications, a large number of MTR topologies may be required to satisfy all application requirements.
However, multiplying the sets of TE metrics, i.e., MTR topologies, yields a significant network overhead as LSA messages are sent for each metric.
On top of that, as network conditions fluctuate over time, topologies may require frequent updates from the management system to follow the evolution of traffic and network performance.
However, link weights can only be updated at a slow pace, and they cannot follow the evolution of network quality over time.

To overcome this issue, we propose a new scalable and versatile IGP extension, called virtual MTR (vMTR), to dynamically maintain a set of \textit{virtual routing topologies}.
This solution aims at increasing routing flexibility in multi-topology routing to follow evolving network conditions and support heterogeneous QoS requirements while reducing the overhead from LSAs and weight updates.
The main idea is to route flows over virtual topologies that are silent, i.e., they do not exchange LSA messages, and that are continuously derived from a very limited set of \textit{real} MTR topologies, i.e., exchanging LSAs.
Each real topologies computes shortest paths for a basic QoS parameter, e.g. delay or packet loss, while virtual ones are designed to generate paths meeting multiple QoS constraints.

Searching for a shortest path satisfying one or multiple end-to-end constraints such as latency or hop count usually involves algorithms relying on Lagrangian relaxation.
These algorithms aim at finding feasible Lagrangian multipliers to modify link costs so that a shortest path provides a feasible solution. Well-known heuristic algorithms are LARAC~\cite{juttner2001lagrange}, for one constraint, or GEN-LARAC~\cite{xiao2005gen} for multiple constraints.
Inspired by this, we propose to derive virtual topologies from real ones applying a linear combination of their link metrics with designed multipliers. Compared to vanilla MTR, our vMTR solution reduces the overhead as each virtual topology is only defined by a set of multipliers, one per real IGP instance, and virtual topologies do not exchange LSAs. Furthermore, routing is automatically adjusted when basic QoS metrics vary and specific link metrics do not need to be updated when traffic conditions evolve.

In this paper, we first formulate the optimization problem to find optimal multipliers so as to minimize the number of virtual topologies in vMTR\@. We propose a polynomial and exact algorithm that finds the minimum set of Lagrangian multipliers to meet two QoS constraints for a set of demands. Then, as a benchmark, we revisit the IGP weight design problem~\cite{fortz2000internet} to minimize the number of topologies in MTR and derive a heuristic algorithm based on local search to solve the problem. Finally, we compare MTR and vMTR with numerical results and show that the number of real MTR topologies can be drastically reduced with vMTR when accepting all traffic demands. We also observe that vMTR topologies are more robust to QoS changes than MTR ones.

The rest of this paper is structured as follows. Related work is discussed in Sec.~\ref{sec:related}. The system architecture and the problem formulation are introduced in Sec.~\ref{sec:system} and~\ref{sec:problem}.
Sec.~\ref{sec:algos} presents algorithmic solutions. Sec.~\ref{sec:results} describes our numerical results and Sec.~\ref{sec:conclusion} concludes this paper.

\section{Related work}\label{sec:related}

The design of link weights so that traffic requirements can be satisfied is known as the \textit{IGP weight design} problem. This problem is a well-studied NP-hard problem~\cite{bley2008routing}.

\textbf{Single topology routing.} Fortz and Thorup~\cite{fortz2000internet,fortz2002optimizing} have initially proposed a local search algorithm where the main idea is to change a few link weights at each time step. They minimize a piece-wise linear cost function related to MLU (Maximum Link Utilization). Right after, a genetic algorithm has been proposed by Ericsson et al.\cite{ericsson2002genetic} and later combined~\cite{buriol2005hybrid} with the original local search. Bley et al.~\cite{bley2008routing} proposed an improved version of the local search algorithm. A math-heuristic based on column generation has also been proposed recently~\cite{codit2023}.

Other variants related to energy-minimization have been considered~\cite{capone2013ospf,amaldi2013energy}. Robust optimization variants have been studied to consider (weighted) link failures~\cite{fortz2003robust} or worst case variations of demands~\cite{vallet2014online} (i.e., minimum / maximum values for each demand) and hoses~\cite{ranaweera2012preventive} at ingress routers to bound aggregate traffic. Oblivious routing extensions where traffic can be load balanced over multiple paths have also been considered~\cite{altin2012oblivious, Duffield} using the hose model. Other traffic polytopes have been used to model traffic uncertainty~\cite{moulierac2015optimizing}.

Beyond traditional optimization techniques,  machine learning~\cite{kodialam2022network} have recently been used to solve the deterministic version of the problem. It is based on a smoothed formulation of the problem
that is suitable for gradient descent.

\textbf{Multi-topology routing.} In the context of network slicing where end-to-end latency constraints must be met, a global optimization model along with results obtained with a genetic algorithm has been presented~\cite{huin2021network} to design multiple topologies (iteratively) and assign demands to them.
Some papers have also proposed solution where topologies are designed offline to maximize edge-disjointness~\cite{wang2008adaptive, wang2012ample} or path diversity~\cite{mirzamany2014efficient}.
These topologies are then used to route traffic based on actual QoS conditions.
Other variants for MTR have been studied to consider robustness to traffic variations~\cite{wang2009robust} and failure scenarios~\cite{1563878}.

\textbf{Our contributions.} In this paper, we extend MTR with virtual topologies that are silent, i.e., they do not exchange LSA messages, and that are continuously derived from a very limited set of real IGP topologies, i.e., exchanging LSAs.

\section{MTR and vMTR for QoS Constraints}\label{sec:system}

This section introduces MTR, and the vMTR extension we propose, from a system perspective. As a preliminary step, we also first formulate the QoS routing problem they address.

\begin{table}[t]
    \centering
    \begin{tabular}{|p{0.1\columnwidth}|p{0.80\columnwidth}|}
        \hline
        {\it Symbol} & {\it Meaning}\\
        \hline
        $K$ & Set of demands\\
        $s_k$, $d_k$ & Source and destination for demand $k \in K$\\
        $r_a^t$  & Resource on arc $a \in A$ for topology $t$\\
        $T_Q$ & Set of \textit{basic} topologies (each related to a different QoS metric)\\
        $\DesignedTopology$ & Set of \textit{designed} topologies\\
        $r_k^t$ &  Max resource consumption for demand $k$ over topology $t$\\
        $SPT^t$ & Shortest path tree on topology $t$ \\
        $\lambda_t$ & Vector of $|T_Q|$ multipliers to design vMTR topology $t$\\

        \hline
    \end{tabular}\vspace{3mm}
    \caption{List of key notations.}\label{notation}
\end{table}

\subsection{QoS satisfaction problem}

Let's represent a network as a directed graph $G=(\NodeSet,\ArcSet)$ where $\NodeSet$ is the set of nodes (i.e., routers) and $\ArcSet$ is the set of arcs (i.e., links).
A set of basic IGP topologies $\BasicTopoSet$ is associated with different QoS metrics such as delay, jitter or losses.
On top of them, additional topologies, real for MTR or virtual for vMTR, can be designed to route demands that cannot be accepted on basic topologies.
In this case, these topologies defined by the set of their link weights, also called resources, denoted $\ResourceArcTopo{\Topo}$ for arc $\Arc \in \ArcSet$ and topology $\Topo \in \BasicTopoSet$, need to be designed along with their mapping to demands.
The set of traffic demands $\DmdSet$ is given as input.
Each demand $\Dmd$ goes from source node $\SrcDmd$ to destination node $\DestDmd$, and it has a set of QoS requirements that are defined by a maximum resource consumption $r_k^t$ over basic (additive) QoS metrics from topology $t \in \BasicTopoSet$.
Note that packet loss can be made additive using the log function.
The problem is to find, for every demand $\Dmd$, a routing path $\PathDmd$ satisfying a set of QoS constraints, as expressed as follows:

\begin{align}
    \sum_{\Arc \in \PathDmd} \ResourceArcTopo{\Topo} \leq \ResourceDmdTopo{\Topo} &\quad&  \forall \Topo \in \BasicTopoSet
\end{align}

We consider a practical use case scenario where the total amount of  traffic with QoS constraints is much smaller than best effort traffic. QoS constrained traffic is prioritized as high priority traffic in the data plane and the sum of high priority traffic (i.e., QoS constrained traffic) is not enough to saturate the bottleneck capacity. For this reason, we omit capacity constraints in the design of topologies and we focus only on QoS constraints satisfaction.

\subsection{MTR\@: Multi-Topology Routing}

When relying on multi-topology routing, the goal is to design the minimum number of topologies so that all demands have their QoS requirements satisfied. As mentioned in Sec.~\ref{sec:related}, this is typically realized by solving a series of inverse shortest path trees problems. When link weights are designed and deployed, traffic is routed over a set of $\DesignedTopology$  topologies, each one corresponding to a new IGP instance. Routers exchange periodically LSA messages to propagate weights and they maintain a minimum spanning tree  $SPT^\Topo$ for every $t \in \DesignedTopology$. Whenever network conditions evolve and the characteristic of links in terms of QoS changes, a significant update of the designed topologies, i.e. weights, may be required.

\subsection{vMTR\@: virtual MTR}

In the case of vMTR, additional virtual topologies need to be designed from the set of existing real topologies.
The central idea behind vMTR is rooted in Lagrangian relaxation, similarly to the LARAC~\cite{juttner2001lagrange} or GEN-LARAC~\cite{xiao2005gen} algorithms that find a path satisfying one or multiple QoS constraints, respectively.
Indeed, virtual topologies are derived from real ones using a set of Lagrangian multipliers.
A set of multipliers, each one associated with a real topology, is used to calculate and maintain link weights in a virtual topology. In this way, virtual topologies are automatically updated when QoS metrics on real instances change. In addition, virtual topologies are completely silent and do not exchange LSAs.

For instance, a virtual topology can be created from two real topologies on loss and delay so that (virtual) link weights are defined by $\lambda^l$.
Loss + $\lambda^d$. Delay with a vector $\lambda$ of multipliers associated to the topology. Virtual link weights are used by the shortest path tree algorithm inside routers to find constrained paths with end-to-end loss and delay constraints. Recall that, in practice, packet loss is made additive using the log function. In the general case, metrics from real topologies can be combined with any algebraic operation. However, without loss of generality, we will consider in the rest of the paper a linear combination. Considering other operations may help to further increase routing diversity and constraint satisfaction, but at the cost of some extra complexity. Indeed, Lagrangian algorithms are actually heuristic to solve the constrained shortest path problem and may not find feasible solutions when they exist (see Sec.~\ref{sec:results}). However, by design, vMTR based on a linear combination is guaranteed to route all demands for which paths can be found with LARAC.

\begin{figure}[t]
    \centering
    \includegraphics[width=\columnwidth]{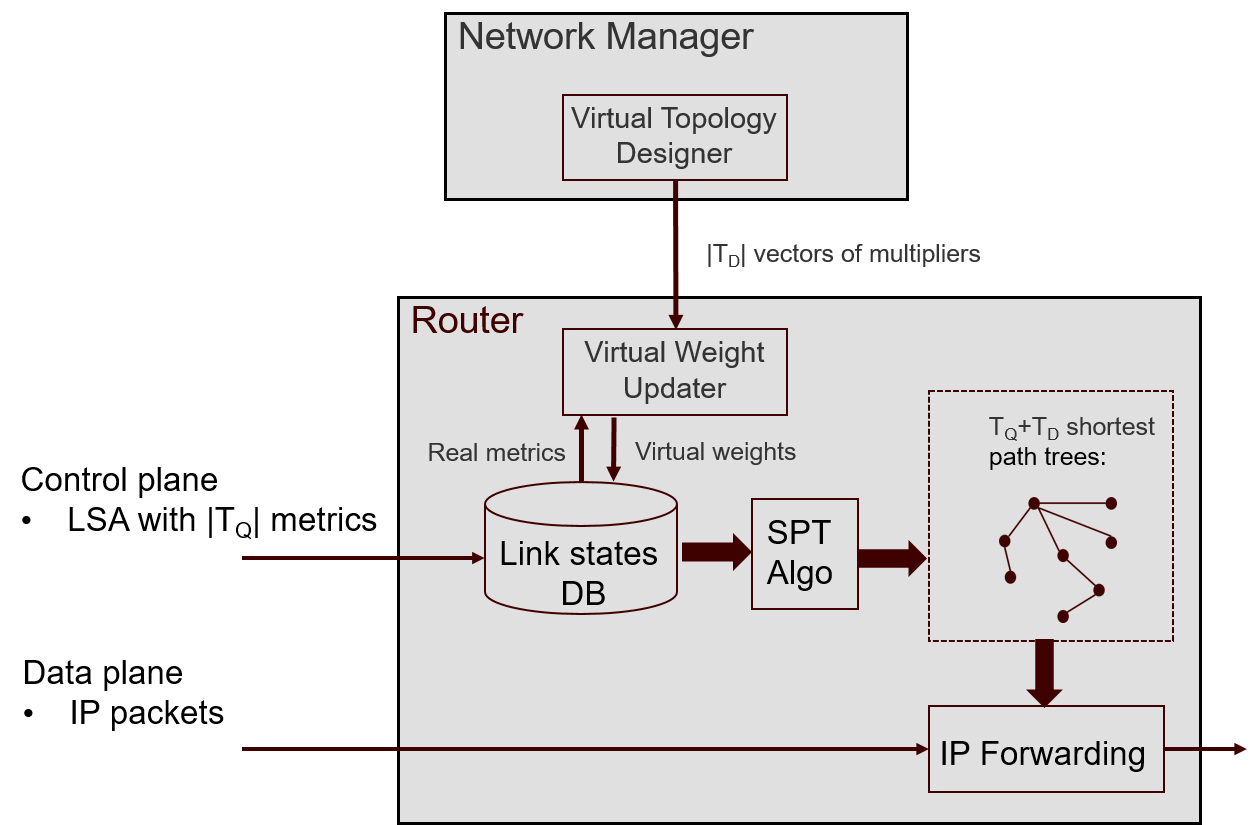}
    \caption{System overview for vMTR.}\label{archi}
    \vspace{-0.4cm}
\end{figure}

To create and manage virtual topologies, a virtual weight designer module sits in the network controller, as depicted in Fig.~\ref{archi}. It takes as inputs the network topology, link metrics related to “real” IGP instances (e.g., cost, capacity, delay), a traffic matrix with a set of demands and their attributes (e.g., source and, destination, QoS requirements). It produces a set of $|\DesignedTopology|$ vectors of multipliers, one for every virtual topologies, along with the mapping of demands to each topology. In practice it sends to all routers $|\DesignedTopology|$ vectors to create $|\DesignedTopology|$ virtual topologies, each vector being of size $|T_Q|$, i.e., the number of existing / real IGP topologies.
Every router continuously maintains a LSDB  with the set of virtual link weights derived from real metrics.
As virtual topologies are based on real ones, no specific LSAs need to be advertised. These topologies are almost immediately available to route traffic as routers do not need to wait for LSA exchanges. When metrics in real topologies are updated, all the virtual weights and their associated shortest path trees are updated locally by devices. Traditional Dynamic SPT~\cite{garg2021dynamizing} mechanisms implemented in IGPs can be used to stabilize the system in case of frequent updates. No routing loops can occur as multipliers and virtual metrics are positive. Unlike link weights in real topologies that need to be exchanged with LSAs, multipliers and the resulting virtual weights do not need to be integer.
As for MTR, packets can be routed over a specific topology  using QoS header fields like DSCP or ToS.

In terms of benefits and impact, one advantage of vMRT is the capability to explain a virtual topology as it is a combination of basic topologies related to metrics that users understand. This is impossible with MTR when weights are optimized by an algorithm. Actually, vMTR offers the possibility of manually or automatically setting multipliers to define a good tradeoff among basic topologies. An expert or customer can select the best multipliers according to his requirements. With regards to FlexAlgo, vMTR can be seen as a way to create new topologies (i.e., new metrics) over which "algos" are applied. We could even imagine an extension to the standard where multipliers are part of the algorithm definition and can be use to tune "algos" to meet QoS constraints. However, in this case, SR has to be rolled out.

Our proposition is in this paper two-fold: vMTR provides a way to create virtual topologies using multipliers and, on top of that, we propose in Sec.~\ref{sec:MTRalg} an algorithm to automatically generate the said multipliers.

\section{Topology design problems}
\label{sec:problem}

The optimization problem associated with the design of topologies for MTR, defined by link weights, consists in creating the minimum number of topologies such that all demands are accepted. This problem is similar in vMTR but differs in the way weights are computed. Indeed, weights of virtual topologies are derived from real topologies. This section presents the ILP (Integer Linear Programming) models for both cases.

In the rest of the paper, we will assume that MTR and vMTR are used to route demands that cannot be accepted over basic topologies $|T_Q|$ (each one maintaining a shortest path tree over a single QoS metric). Also, to ease the presentation, we will omit demands for which a constrained path cannot be found.

\begin{figure*}[t]
    \begin{align}
        \min \quad &\sum_{t\in \DesignedTopology} z_t   \label{objMRT}\\
                   & \sum_{k\in K} y_{k,t} \leq  z_t & & \forall t\in \DesignedTopology\label{ct:topo}\\
                   & \sum_{t\in \DesignedTopology} y_{k,t} \geq  1 & & \forall k\in K\label{ct:selectTopoMTR}\\
                   & \sum_{a\in \omega^+(v)} x_{a,k,t} -\sum_{a\in \omega^-(v)} x_{a,k,t} =  \begin{cases}
                       y_{k,t} \text{ if } v = s_k, \\
                       -y_{k,t}  \text{ if } v = d_k, \\
                       0 \text{ otherwise;}
                   \end{cases}&\quad& \forall v\in V, \forall k\in K, \forall t\in \DesignedTopology \label{ct:routingMTR}\\
                   &  x_{a,k,t} \leq u_{a,d_k,t} &\quad& \forall a\in A, \forall k\in K, \forall t\in \DesignedTopology              \label{ct:xtouMTR}\\
                   &  \sum_{a\in \omega^+(v)} u_{a,d,t}\leq 1 &\quad& \forall v\in V, \forall d\in d_K, \forall t\in \DesignedTopology                       \label{ct:treeIncidentMTR}\\
                   &  u_{a,d,t} \leq \sum_{k \in K: d=d_k} x_{a,d_k,t} &\quad& \forall a\in A, \forall d\in d_K, \forall t\in \DesignedTopology
                   \label{ct:utoxMTR}\\
                   &  w_{(u,v),t} - \pi_{u,d,t} +\pi_{v,d,t} \geq 1 - u_{(u,v),d,t} &\quad& \forall (u,v)\in A, \forall d\in d_K, \forall t\in \DesignedTopology \label{ct:wgeqMTR}\\
                   &  w_{(u,v),t} - \pi_{u,d,t} +\pi_{v,d,t} \leq M(1 - u_{(u,v),d,t}) &\quad& \forall (u,v)\in A, \forall d\in d_K, \forall t\in \DesignedTopology\label{ct:wleqMTR}\\
                   & \sum_{a\in A}\sum_{t\in \DesignedTopology} r_a^t  x_{a,k,t} \leq r_k^t &\quad& \forall k\in K, \forall t\in T_Q \label{ct:ResourceMTR}
    \end{align}
           \vspace{-0.2cm}
    \caption{Integer linear problem for the design of topologies in MTR.}\label{MTRmodel}
        \vspace{-0.3cm}
\end{figure*}

\subsection{MTR}

Let us introduce the variables to solve the topology design problem for MTR\@:

\begin{itemize}
    \item $x_{a,k,t}\in \{0,1\}$: equals 1 if the demand $k$ is routed along a path using link $a$ in the topology $t$, and 0 otherwise.
    \item $\pi_{u,v,t}\in \mathbb{R}^+$: is the potential of the node $u$, i.e., the distance between node $u$ and $v$ on the topology $t$.
    \item $u_{a,d,t}\in \{0,1\}$: equals 1 if the link $a$ belongs to a shortest path towards destination $d$ on topology $t$, and 0 otherwise.
    \item $w_{a,t}\in \mathbb{R}^+$: the weight assigned to link $a\in A$ on the topology $t$.
    \item $y_{k,t}\in \{0,1\}$: equals 1 if the demand $k$ is assigned to the topology $t$ (thus feasible for this topology), and 0 otherwise.
    \item $z_t\in \{0,1\}$: equals 1 if the topology $t$ is deployed.
\end{itemize}

Let us consider $d_K$ the set of nodes that are destinations of at least one demand. We denote by $\omega^-(\NodeIdx)$ (resp. $\omega^+(\NodeIdx)$) the set of ingoing (resp. outgoing) links of node $\NodeIdx$.

As we do not know the needed number of topologies, we can consider an upper bound to create the right number of variables. This upper bound can be set by the number of designed topologies returned by the local search algorithm presented in the next section.

The ILP model is given by Fig.~\ref{MTRmodel}.
This model is an extension of the well-known model from the literature (e.g.,~\cite{wang2009robust, huin2021network}) to jointly create multiple topologies.
The main goal is to minimize the number of designed topologies.
Constraints \eqref{ct:topo} link the activation of topologies and demands routed on top of them.
Inequalities \eqref{ct:selectTopoMTR} ensure that each demand is assigned to one designed topology.
Equalities \eqref{ct:routingMTR} are the well-known flow conservation constraints.
Constraints \eqref{ct:xtouMTR} and \eqref{ct:utoxMTR} link routing variables and shortest path variables.
Inequalities \eqref{ct:treeIncidentMTR} guarantee that the induced shortest path tree is a tree, i.e., there is at most one outgoing link at each node.
Constraints \eqref{ct:wgeqMTR} and \eqref{ct:wleqMTR} ensure that the shortest path tree follows the weight $w$.
Inequalities \eqref{ct:ResourceMTR} are resource constraints.

\subsection{vMTR}

For vMTR we need to add the following variables:
\begin{itemize}
    \item $\lambda_{t,\bar{t}}\in \mathbb{R}$: the multiplier associated with the basic topology $t\in T_Q$ and the virtual topology $\bar{t} \in \DesignedTopology$.
\end{itemize}

and the following constraints:

\begin{align}
   & w_{a,\bar{t}}=\sum_{t\in T_Q} \lambda_{t,\bar{t}} w_{a,t} &\quad& \forall a\in A, \forall \bar{t}\in \DesignedTopology \label{ct:ResourceMTR}
\end{align}

Remark that in the case of vMTR, we complete the designed virtual topologies with real MTR topologies when all demands cannot be accepted with vMTR. To manage the two kinds of designed topologies, we penalize in the objective function the activation of MTR topologies.

\section{Algorithmic solutions}\label{sec:algos}

This section presents algorithms for MTR and vMTR\@. While multiple algorithms exist for MTR (see Sec.~\ref{sec:related}), we build upon one of them~\cite{bley2008routing} to focus on QoS constraint satisfaction and minimize the number of designed topologies.

\subsection{MTR Algorithm}\label{sec:MTRalg}

\input{local_search.tex}

We developed a greedy algorithm to design the sets of weights, i.e., topologies, as shown in Algo.~\ref{alg:greedy}.
It iteratively generates random weight vectors, fine-tuned to maximize the number of accepted demands using a local search algorithm. MTR topologies are created one by one, each time trying to accept a maximum number of demands.
If the local search fails to generate a weight vector accepting any demands, it falls back to computing a valid resource-constrained shortest path for one of the remaining demands (using the TAMCRA algorithm~\cite{de2000tamcra}), then uses that path for building a valid weight vector from a new random one.

The greedy algorithm uses a run-of-the-mill local search that encodes a solution by its weight vector, and the neighborhood of a solution is the set of weight vectors that differ from the solution by a single element.
The fitness function used to evaluate a solution computes the shortest path tree according to the weight vector and counts the number of demands with a path respecting their required QoS.

Bley et al.~\cite{bley2008routing} already proposed a local search algorithm for the IGP weight design problem, but we consider different assumptions.
Firstly, while they consider an extended neighborhood (compared to ours) that takes into account load balancing, we only focus on the weight vector variation.
Secondly, while they consider weights in the range $[0, 20]$, we consider the whole range of possible OSPF weights ($[0, 2^{16}]$)
The last point is particularly important since the definition of their neighborhood considers atomic changes in the weight ($\pm 1$), which is quite effective considering the range of the weights.
However, since we allow a much larger range of weights, we need to define a neighborhood that guarantees changes in the shortest path trees induced by the weights.

We thus generate meaningful neighbors (see line~\ref{alg:generate-neighborhood}) computing the minimum increase or decrease in weight needed to change the shortest path trees.
We call these values the up and the down delta-weight of a link, respectively. We denote by $\pi(u,v)$ the distance of the shortest path from $u$ to $v$ with the current link weights.

First, we define the down delta-weight $\downdelta_{\NodeIdx}(\LinkIdx)$ of a link $\LinkIdx=(s_\LinkIdx, d_\LinkIdx)$ as the weight decrease needed to insert the link $\LinkIdx$ in the current shortest path tree $T_{\NodeIdx}$, rooted at ${\NodeIdx}$, in the place of the link $\LinkIdx^\prime = (s_\LinkIdx^\prime, d_\LinkIdx)$.
Formally, we are looking for a value $\downdelta_\NodeIdx(\LinkIdx)$ such that
\[
    \pi(\NodeIdx, s_\LinkIdx) + w_{(\LinkIdx)} - \downdelta_\NodeIdx(\LinkIdx) < \pi(\NodeIdx, t_\LinkIdx),
\]
\[
    \downdelta_\NodeIdx(\LinkIdx) = \pi(\NodeIdx, s_\LinkIdx) + w_{(\LinkIdx)} - \pi(\NodeIdx, t_\LinkIdx)
\]
Second, we define the up delta-weight $\updelta_\NodeIdx(\LinkIdx)$ of a link $\LinkIdx=(s_\LinkIdx, d_\LinkIdx)$ as the weight increase needed to remove the link from the current shortest path tree $SPT_\NodeIdx$, rooted at $\NodeIdx$.
Formally, we are looking for a value $\updelta_\NodeIdx(\LinkIdx)$ such that
\def\OtherLinkIdx{\LinkIdx^\prime}
\[
    \pi(\NodeIdx, s_\LinkIdx) + w_{(\LinkIdx)} + \updelta_\NodeIdx(\LinkIdx) > \min_{\OtherLinkIdx \in \omega^-(d_\LinkIdx) \setminus \{\LinkIdx\}} (\pi(\NodeIdx, s_{\OtherLinkIdx}) + w_{(\OtherLinkIdx)}).
\]
Thus, the up delta weight of a link $\LinkIdx$ is given by
\begin{align}
    \updelta_\NodeIdx(\LinkIdx) & = \min_{\OtherLinkIdx \in \omega^-(d_\LinkIdx) \setminus \{\LinkIdx\}} (\pi(\NodeIdx, s_{\OtherLinkIdx}) + w_{(\OtherLinkIdx)}) - \pi(\NodeIdx, s_\LinkIdx) - w_{(\LinkIdx)}\\
                                & = \min_{\OtherLinkIdx \in \omega^-(d_\LinkIdx) \setminus \{\LinkIdx\}} \downdelta_\NodeIdx(\OtherLinkIdx)
\end{align}

Once we computed the delta-weight of each link, we can propagate them (see line~\ref{alg:compute-delta-weights:propagate}) up the tree.
Indeed, any weight change occurring in the higher depth of the tree will impact the distance of the descendants.
Given a link $\LinkIdx_1 = (s_1, d_{\LinkIdx_1})$ and its swapped link $\LinkIdx_2 = (s_2, d_{\LinkIdx_1})$, we can propagate their up and down delta-weight up to their lowest common ancestor in the tree.
Formally, if $u$ is the lowest common ancestor of $s_{\LinkIdx_1}$ and $s_{\LinkIdx_2}$ in $SPT$, and $p_T(\NodeIdx_1, \NodeIdx_2)$ is the path from $\NodeIdx_1$ to $\NodeIdx_2$ in $SPT$, we have the following properties:
\[
    \forall  \LinkIdx \in p_T(u, s_1), \updelta(\LinkIdx) = \updelta(\LinkIdx_1),
\]
\[
    \forall\LinkIdx \in p_T(u, s_2), \downdelta(\LinkIdx) = \downdelta(\LinkIdx_2)
\]

Finally, we need to consider all possible shortest path trees in the graph.
We can combine their up and down delta-weight by taking the minimum over all trees.

\subsection{vMTR Algorithm}

The weight design problem for vMTR consists in maximizing the number of accepted demands (first objective) while minimizing the number of virtual topologies (second objective).
To create a virtual topology with vMTR over 2 basic topologies (e.g., delay and loss), we consider a single multiplier $\lambda$ that provides the following virtual weight $r_a^t = r_a^1+\lambda r_a^2$ for each link $a\in A$. A multiplier $\lambda$ is said to be feasible for demand $k$ if the shortest path $p$ over link metrics $r_a^t$ satisfies the maximum constraint on the end-to-end consumption for both resources, i.e., on $r_k^1$ and $r_k^2$.

In the following, we first present some properties for this particular case with 2 resources and then, we propose an exact algorithm that solves the problem in polynomial time.

\textbf{Lagrangian properties.} As one multiplier can cover several demands, a natural question is: does there exist one or several intervals of feasible multipliers for a given demand? In this section, we provide an answer to this question that allows solving the vMTR problem in polynomial time afterward.

Let us introduce the Constraint Shortest Path problem (CSP) where the goal is to find a path minimizing a resource under an additive end-to-end resource constraint.
In \cite{doi:10.1080/09728600.2005.12088797}, the authors propose and analyze the LARAC algorithm, a polynomial time heuristic solution for the CSP problem based on  Lagrangian relaxation.
This algorithm computes a multiplier $\lambda^*$ and the solution it returns is the shortest path $p_{\lambda^*}$ in the graph where the cost of each link $a\in A$ is given by $r_a^1+\lambda^* r_a^2$.

The most interesting claim proved in \cite{doi:10.1080/09728600.2005.12088797} for us is given below.

\begin{claim}\label{claim1}
    \cite{doi:10.1080/09728600.2005.12088797} If $\lambda<\lambda^*$, then $\sum_{a\in p_\lambda} r^2_a \geq r^2_k$ and if $\lambda>\lambda^*$, then $\sum_{a\in p_\lambda} r^2_a \leq r^2$
\end{claim}
This claim is the core of our vMTR algorithm.

We denote by $\lambda^{1*}_k$ (resp.\ $\lambda^{2*}_k$) the multiplier provided by LARAC algorithm when resource 2 (resp. resource 1) is minimized and resource 1 (resp. resource 2) is constrained. As proved in \cite{doi:10.1080/09728600.2005.12088797} this multiplier is optimal.

The following proposition ensures that, for a given demand, only one range of multiplier is feasible for the two resources.

\begin{proposition} \label{prop1}
    For a given demand $k\in K$, a multiplier $\lambda_k$ is feasible for $k$, if  \[\lambda^{2*}_{k}\leq \lambda_k \leq \frac{1}{\lambda^{1*}_{k}}.\]
\end{proposition}
\begin{proof}
Let us consider the shortest path $p$ where link costs $r_a^2 +\lambda^{1*}_{k} r_a^1, \forall a \in A$ are considered.
The path $p$ is also the shortest path in the graph where link costs are $r_a^1+\frac{1}{\lambda^{1*}_{k}}r_a^2, \forall a \in A$. Indeed, for each $a \in A$, dividing $r_a^2 +\lambda^{1*}_{k} r_a^1$ by $\lambda^{1*}_{k}$ induced $r_a^1+\frac{1}{\lambda^{1*}_{k}}r_a^2$.
Thus, thanks to the claim \ref{claim1}, all $\lambda^{2}_k\geq\lambda^{2*}_{k}$ induced a shortest path where resource 2 is respected and all $\frac{1}{\lambda^{1}_k}\leq \frac{1}{\lambda^{1*}_{k}}$ induced a shortest path where resource 1 is respected. Thus all $\lambda_k$ such that $\lambda^{2*}_{k}\leq \lambda_k \leq \frac{1}{\lambda^{1*}_{k}}$ are feasible for the two resources. \end{proof}

For a given multiplier $\lambda$, we may have several paths that are the shortest, but one of them can be unfeasible for resource 1 or 2.
To address this issue, we modified the Dijkstra algorithm to find all shortest paths between two nodes in a graph and we slightly changed the $\lambda$ value by adding an $\epsilon$ until all shortest paths are valid or only one path is the shortest.

\begin{figure}[h]
    \includegraphics[page = 2, width=\linewidth]{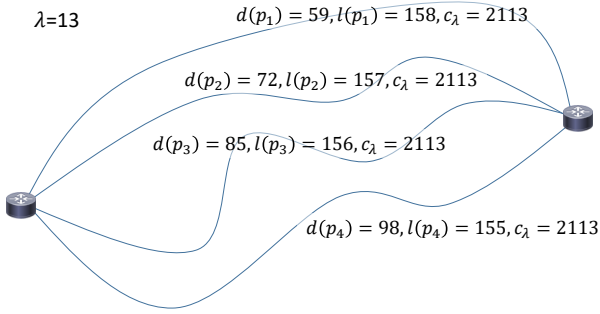}
    \caption{Example of non shortest paths uniqueness with $\lambda=13$.}\label{non-uniquement}
\end{figure}

In the following, for a given demand $k$, we denote $\lambda_k^{\min}=\lambda^{2*}_{k}$ and $\lambda_k^{\max}=\frac{1}{\lambda^{1*}_{k}}$.
Thanks to the convergence of LARAC algorithm \cite{doi:10.1080/09728600.2005.12088797}, $\lambda_k^{\min}$ is minimum for each demand $k\in K$ and $\lambda_k^{\max}$ is maximum for each demand $k\in K$.

From Proposition~\ref{prop1}, we can derive an efficient polynomial algorithm to solve the weight design problem for vMTR.
As a pre-processing step, we compute the multipliers $\lambda_k^{\min}$ and $\lambda_k^{\max}$ for each demand $k\in K$. Each demand where $\lambda_k^{\min}>\lambda_k^{\max}$ is discarded as it cannot be routed with vMTR. At the end of the procedure, the pool of discarded demand is processed with the MTR algorithm (Sec.~\ref{sec:MTRalg}). The overall goal is to minimize the number of MTR topologies, as they generate LSAs and do not adapt to evolving network conditions.

\textbf{Exact algorithm.} The goal of this algorithm is to find a set of multipliers $\Lambda$ such that all demands are covered. A multiplier $\lambda$ covers a demand $k$ if $\lambda$ is feasible for the demand (see Proposition 1). At each step, we find a new multiplier to add in $\Lambda$ such that some uncovered demands become covered with the new multiplier. More formally. Let $\bar{K}$ be the set of demands covered at the current step, and respectively, $K\setminus \bar{K}$  the set of uncovered demands. 

Algo.~\ref{alg:three} describes the whole algorithm and an example is given in Fig.~\ref{lambdaComputation}. First, the algorithm orders the $\lambda_k^{\max}$ in increasing order. Then, at each step, it selects a $\lambda_{k^*}^{\max}$ to add in $\Lambda$ such that all uncovered commodities have an associated $\lambda_k^{\max}$ greater or equal to the selected $\lambda_{k^*}^{\max}$. It continues until all demands are covered.

\begin{figure}[t]
    \includegraphics[page = 1, width=0.9\linewidth]{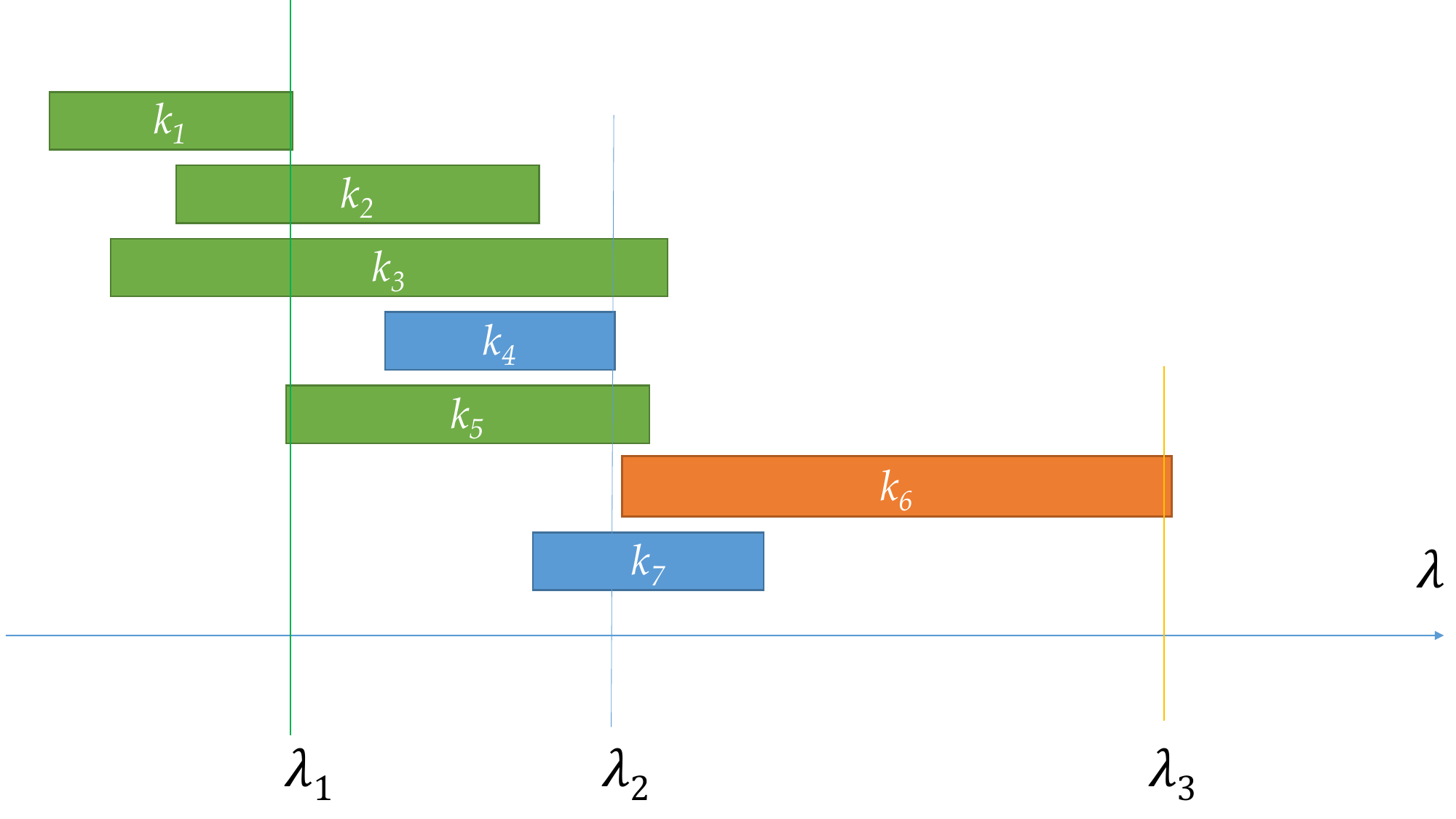}
    \caption{Example with $6$ demands.  Intervals between $\lambda_{k}^{\min}$ and $\lambda_{k}^{\max}$ are represented for each demand $k$. The first multiplier $\lambda_1$ covers the green demands, the second multiplier $\lambda_2$ covers the blue demands and $\lambda_3$ covers the orange demand.\label{lambdaComputation}}
        \vspace{-0.4cm}
\end{figure}

\begin{algorithm}[t]
\small
    \caption{Weight design algorithm for vMTR}\label{alg:three}
    \begin{algorithmic}[1]
        \Function{}{$G=(V,A), K$}
        \State{$\bar{K}=\emptyset$, $\Lambda=\emptyset$}
        \State{$\sigma$ ordering of demands $k\in K$ in increasing order of $\lambda_k^{\max}$}
        \While{$K\setminus \bar{K}\neq \emptyset$}
        \State{$\lambda^*$ the smallest $\lambda_k^{\max}$ such that $k \in K\setminus \bar{K}$}
        \State{add in $\bar{K}$ all demands covered by $\lambda^*$ not in $\bar{K}$}
        \State{add $\lambda^*$ in $\Lambda$ }
        \EndWhile
        \State{\Return{$\Lambda$}}
        \EndFunction
    \end{algorithmic}
\end{algorithm}

\textbf{Complexity and generalization.} To summarize the approach above, the problem is solved in two steps. First, we define the space of valid multipliers for each demand. Second, we search the minimum number of multiplier vectors such that all demands are covered.

The dimension of the space of valid multipliers is equal to $|T_Q|-1$. For the case with $2$ basic topologies, we have only 1 dimension, which can be represented by an interval graph and the problem is polynomial as presented above. If we consider $3$ basic topologies, we have 2 dimensions that can be represented by a plane. In this case, we can consider that the area of valid multipliers, vector of size 2, represents a rectangle, which is a particular case of our problem. This kind of graph is called boxicity 2-graph. For the interval graph or a boxicity 2-graph, the selection of the minimum number of multiplier vectors consists in solving a clique cover problem in the graph $H=(K,I)$ where $K$ is the set of demands and $I$ is the set of links. Two demands are linked if and only if they share a common valid multiplier vector. This problem is NP-hard for a boxicity 2-graph which is a particular case of the problem when 3 basic topologies are considered and polynomial for an interval graph \cite{zbMATH03874636}.

\section{Performance evaluation}
\label{sec:results}

We now compare MTR and vMTR over a set of instances generated from SNDlib~\cite{orlowski2010sndlib}.

\textbf{Description of instances}.
These instances contain the geographical coordinates for each node and the link capacities between them.
We selected 15 large topologies.
For each network, we have considered two resources, the loss and the delay.
For each link, the packet loss is set to be inversely proportional to the link capacity and the delay is considered to be proportional to the distance.
To avoid trivial instances, we generated all demands between each pair of routers.
We set the maximum end-to-end loss (resp. delay) to one of the shortest path on the loss (resp. delay) minus $\epsilon$.
This ensures that the demands cannot be routed on basic topologies, i.e., loss or delay.
For each demand, we solved optimally, i.e., not using LARAC, the constraint shortest path problem to ensure that a solution exists.
If no solution exists for a demand then we did not consider it.
All instances are available at~\cite{huin_2024}.

\begin{figure*}[t]
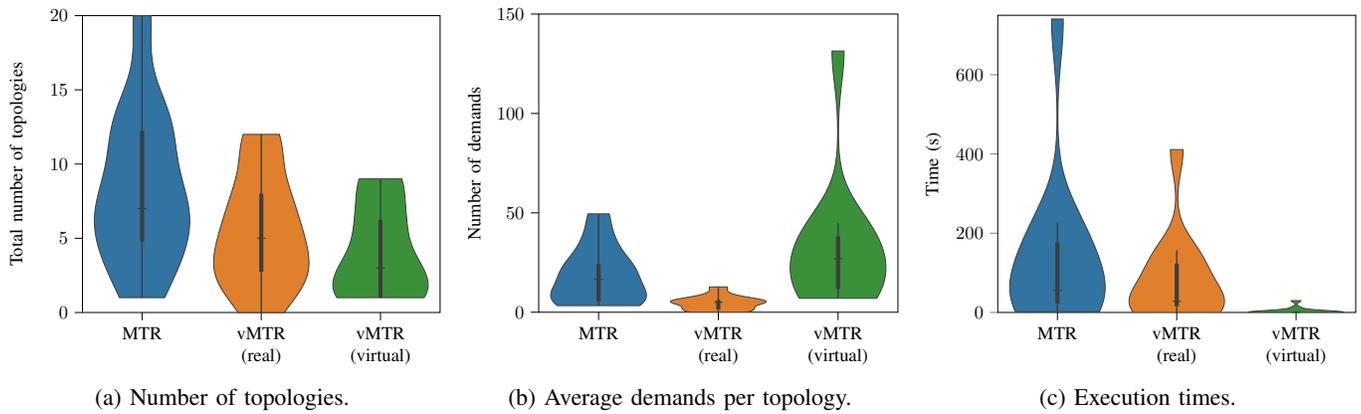

    \begin{subfigure}[t]{.32\textwidth}
        \centering
        \includestandalone[width=\linewidth]{figures/nb_topologies}
        \caption{Number of topologies.}\label{fig:nb_topologies}
    \end{subfigure}
    \hfil
    \begin{subfigure}[t]{.32\textwidth}
        \centering
        \includestandalone[width=\linewidth]{figures/average_demands_per_topo}
        \caption{Average demands per topology.}\label{fig:average_demands}
    \end{subfigure}
    \hfil
    \begin{subfigure}[t]{.32\textwidth}
        \centering
        \includestandalone[width=\linewidth]{figures/execution_time}
        \caption{Execution times.}\label{fig:execution-time}
    \end{subfigure}
    \caption{Comparison of MTR and vMTR using violin plots.}\label{fig:results}
\end{figure*}

\textbf{Numerical results}. Fig.~\ref{fig:results} presents violin plots to summarize results over all 15 instances. We can see, in Fig.~\ref{fig:nb_topologies}, that vMTR reduces the number of real topologies compared to MTR, from an average of 8.41 to 5.41.
It also reduces the maximum number of needed topologies from 20 to 12, and it also avoid the need for any real topology in some cases.
vMTR requires an average of 4.08 virtual topologies.
The decrease in the number of real topologies can be easily explained by looking at the number of average demands per topologies in Fig.~\ref{fig:average_demands}.
vMTR virtual topologies can accept more demands than MTR real ones: in average, a vMTR topology holds 32.7 demands while, in average, MTR topologies hold 18.7 demands.
In some cases, we can even have virtual topologies with 131.4 demands.
However, real topologies in vMTR hold less demands because the remaining demands not covered by a virtual topology are harder to cover with the greedy MTR algorithm.
Recall that as our implementation of vMTR with a linear combination of real topologies is based on the same principle as LARAC, which is a heuristic solution to the CSP problem, some demands cannot be satisfied with virtual topologies and real ones need to be added to ensure 100\% traffic acceptance.

Fig.~\ref{fig:execution-time} shows the execution times for both solutions. For vMTR, we present the time required to compute lambda multipliers for virtual topologies as well as the time needed by the greedy algorithm to provision the remaining demands over real topologies. vMTR, once again, shows an improvement over MTR with an average time of 86.2s over 141.3s, and it also decreases the maximum time from 740.5s to 411.1s. The computational time to compute only the virtual topologies is really small. It is due to the polynomiality of the algorithm.

\begin{figure}[b]
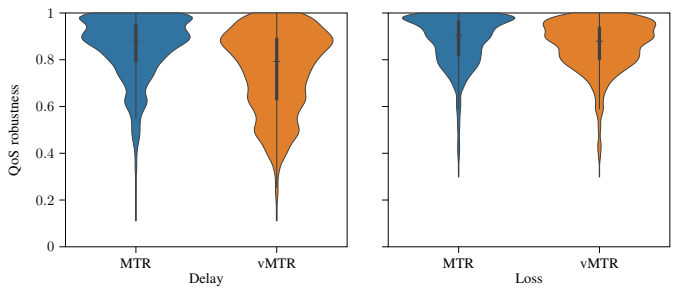

    \centering
    \includestandalone[width=\linewidth]{figures/qos-robustness}
    \caption{Violin plots of the QoS robustness which is, for a demand and a metric, the ratio between the end-to-end QoS of the provided paths and the requested QoS. A value of 1 for a given metric means that the path of a demand is at its maximum resource consumption.}\label{fig:qos_robustness}
\end{figure}

Finally, we study the robustness of the two approaches against changes of QoS metrics.
To this end, we define the \textit{QoS robustness} of a demand, for a given QoS metric, as the ratio of its provisioned path QoS over its end-to-end QoS requirement. A QoS robustness of 1 means that the provided path QoS equals the required QoS, and any increase of this QoS metric in the path's link would invalidate the path. While we can see in Fig.~\ref{fig:qos_robustness} that both solutions provision demands on strict paths (ratio of 1), on average, vMTR provides paths that are more robust to QoS changes than MTR (0.75 and 0.79 for vMTR vs 0.85 and 0.88 for MTR, for delay and losses, respectively). Remark that the robustness is slightly better from the delay point of view if we compare with the Loss. This is due to the distribution of metric values in our instances. For the two solutions, we have a better maximum robustness for the delay than for the loss. We then have more flexibility to improve the robustness of the delay.

\section{Conclusion}\label{sec:conclusion}

We introduced a new IGP extension, called virtual MTR (vMTR), to dynamically maintain a set of virtual routing topologies. This solution aims at increasing routing flexibility in multi-topology routing to follow evolving network conditions and support heterogeneous QoS requirements while reducing the overhead from LSAs and weight updates. We presented a polynomial and exact algorithm for vMTR when virtual topologies are derived from two real topologies and a local search algorithm for MTR as benchmark. We showed that vMTR can help reducing the number of real topologies and is more robust to QoS changes.

Future work along these lines include several directions. First, we may propose an efficient algorithm for the design of MTR topologies jointly with the creation of virtual ones in vMTR. We can also design lambda multipliers considering traffic scenarios or uncertainty sets on traffic. Also of interest, we may consider non additive and non linear ways to combine real topologies. Finally, we may develop efficient algorithms for the hard case where more than two real topologies are used to derive virtual ones.

\bibliographystyle{IEEEtran}
\bibliography{biblio}

\end{document}

%% file: local_search.tex
\begin{algorithm}[t]
\small
    \caption{Greedy algorithm for MTR}\label{alg:greedy}
    \begin{algorithmic}[1]
        \Function{greedy}{$G = (\NodeSet, \LinkSet), d, l, K$}
            \State{$\mathcal{T} \gets \emptyset$} \Comment{Set of weight-demands pairs}
            \While{$\DmdSet \neq \emptyset$}
                \State{$w \gets \{\mathcal{U}(0, 65535) | \LinkIdx\in\LinkSet\}$}\Comment{random values}
                \State{$w \gets \Call{localSearch}{G, w, d, l, K} $}
                \State{$K^\prime \gets \Call{getAcceptedDemands}{G, w, d, l, K}$}
                \If{$K^\prime = \emptyset$}
                    \State{$\PathIdx \gets \Call{TAMCRA}{G, d, l, K}$}
                    \State{$w \gets \{\mathcal{U}(0, 65535) | \LinkIdx\in\LinkSet\}$} \Comment{random values}
                    \ForAll{$\LinkIdx\in\PathIdx$}
                        \State{$ w_\LinkIdx \gets 1$}
                        \EndFor{}
                \State{$K^\prime \gets \Call{getAcceptedDemands}{G, w, d, l, K}$}
                \EndIf{}
                \State{$\mathcal{T} \in \mathcal{T}\cup{(w, K^\prime)}$}
                \State{$K \gets K \setminus K^\prime$}
            \EndWhile{}
        \EndFunction{}
        \Function{localSearch}{$G = (\NodeSet, \LinkSet), w, d, l, K$}\label{alg:local-search}
            \State{$\tilde{K}\gets \Call{getAcceptedDemands}{G, w, d, l, K}$}
            \State{$(w^\ast, K^\ast) \gets (w, \Call{getAcceptedDemands}{G, w, d, l, K})$}
            \For{$i \gets 0$ to maxIte}
                \State{$\mathcal{W} \gets \Call{generateNeighborhood}{w, G, d, l}$}
                \State{$w \gets \arg\max_{w\in\mathcal{W}} \Call{getAcceptedDemands}{G, w, d, l}$}
                \State{$\tilde{K} \gets $\Call{getAcceptedDemands}{$G, w, d, l, K$}}
                \If{$\vert\tilde{K}\vert > \vert K^\ast\vert$}
                    \State{$(w^\ast, K^\ast) \gets (w, \tilde{K})$}
                \EndIf
            \EndFor
            \State{\Return{$w^\ast$}}
        \EndFunction
        \Function{generateNeighborhood}{$w, G = (\NodeSet, \LinkSet)$}\label{alg:generate-neighborhood}
            \State{$(\updelta, \downdelta) \gets $\Call{computeDeltaWeights}{G, w}}
            \State{$\mathcal{W} \gets \emptyset$}
            \ForAll{$\LinkIdx\in\LinkSet$}
                \State{$w^-\gets w$}
                \State{$w_\LinkIdx^-\gets w_\LinkIdx + \downdelta(\LinkIdx$)}
                \State{$w^+\gets w$}
                \State{$w_\LinkIdx^+\gets w_\LinkIdx + \updelta(\LinkIdx)$}
                \State{$\mathcal{W}\gets\mathcal{W}\cup\{w^+, w^-\}$}
            \EndFor
            \State{\Return{$\mathcal{W}$}}
        \EndFunction
        \Function{computeDeltaWeights}{$G = (\NodeSet, \LinkSet), w$}\label{alg:compute-delta-weights}
            \State{$\updelta \gets \{\infty | \LinkIdx\in\LinkSet \}$}
            \State{$\downdelta \gets \{\infty | \LinkIdx\in\LinkSet \}$}
            \ForAll{$\NodeIdx\in\NodeSet$}
                \State{$T, d \gets\Call{dijkstraShortestPathTree}{G, \NodeIdx, w}$} \Comment{Get shortest path tree with distances }
                \ForAll{$(s_\LinkIdx, t_\LinkIdx)\in\LinkSet\setminus T$}\label{alg:compute-delta-weights:propagate}
                    \State{$\NodeIdx_A \gets \Call{LowestCommonAncestor}{T, s_\LinkIdx, t_\LinkIdx}$}
                    \ForAll{$\LinkIdx^\prime\in\PathIdx_{T(\NodeIdx_A, s_\LinkIdx)}$}
                        \State{$\downdelta(\LinkIdx^\prime)\gets \min(\downdelta(\LinkIdx^\prime), d(s_\LinkIdx) + w(\LinkIdx) - d(t_\LinkIdx))$}
                    \EndFor{}
                    \ForAll{$\LinkIdx^\prime\in\PathIdx_{T(\NodeIdx_A, t_\LinkIdx)}$}
                        \State{$\updelta(\LinkIdx^\prime)\gets \min(\updelta(\LinkIdx^\prime), d(s_\LinkIdx) + w(\LinkIdx) - d(t_\LinkIdx))$}
                    \EndFor{}
                \EndFor{}
            \EndFor
        \EndFunction
    \end{algorithmic}
\end{algorithm}

%% file: figures/nb_topologies.tex
\begin{tikzpicture}

\definecolor{darkgray176}{RGB}{176,176,176}
\definecolor{darkslategray61}{RGB}{61,61,61}
\definecolor{peru22412844}{RGB}{224,128,44}
\definecolor{seagreen5814558}{RGB}{58,145,58}
\definecolor{steelblue49115161}{RGB}{49,115,161}

\begin{axis}[
tick align=outside,
tick pos=left,
x grid style={darkgray176},
xmin=-0.5, xmax=2.5,
ymin=0, ymax=20,
xtick style={color=black},
xtick={0,1,2},
xticklabels={MTR,vMTR\\(real),vMTR\\(virtual)},
xticklabel style = {align=center},
y grid style={darkgray176},
ylabel={Total number of topologies},
ytick style={color=black}
]
\path [draw=darkslategray61, fill=steelblue49115161, line width=0.5pt]
(axis cs:0.190160418822847,1)
--(axis cs:-0.190160418822847,1)
--(axis cs:-0.201277965928942,1.19191919191919)
--(axis cs:-0.212459078254162,1.38383838383838)
--(axis cs:-0.223664781070795,1.57575757575758)
--(axis cs:-0.234856245058647,1.76767676767677)
--(axis cs:-0.245994915083779,1.95959595959596)
--(axis cs:-0.257042568816978,2.15151515151515)
--(axis cs:-0.26796130697276,2.34343434343434)
--(axis cs:-0.278713480810858,2.53535353535354)
--(axis cs:-0.289261566517662,2.72727272727273)
--(axis cs:-0.299568000028695,2.91919191919192)
--(axis cs:-0.309594989588885,3.11111111111111)
--(axis cs:-0.319304326670639,3.3030303030303)
--(axis cs:-0.328657218553599,3.49494949494949)
--(axis cs:-0.337614167674887,3.68686868686869)
--(axis cs:-0.346134923547337,3.87878787878788)
--(axis cs:-0.354178532398979,4.07070707070707)
--(axis cs:-0.361703507535267,4.26262626262626)
--(axis cs:-0.368668139655505,4.45454545454545)
--(axis cs:-0.375030960940922,4.64646464646465)
--(axis cs:-0.38075136975047,4.83838383838384)
--(axis cs:-0.385790414402312,5.03030303030303)
--(axis cs:-0.390111725092538,5.22222222222222)
--(axis cs:-0.393682572929148,5.41414141414141)
--(axis cs:-0.396475024857801,5.60606060606061)
--(axis cs:-0.398467153518402,5.7979797979798)
--(axis cs:-0.399644252431104,5.98989898989899)
--(axis cs:-0.4,6.18181818181818)
--(axis cs:-0.399537511234078,6.37373737373737)
--(axis cs:-0.398270214322586,6.56565656565656)
--(axis cs:-0.396222490643394,6.75757575757576)
--(axis cs:-0.393430021642986,6.94949494949495)
--(axis cs:-0.389939794331762,7.14141414141414)
--(axis cs:-0.385809728711003,7.33333333333333)
--(axis cs:-0.38110790490598,7.52525252525253)
--(axis cs:-0.375911384546948,7.71717171717172)
--(axis cs:-0.370304639271329,7.90909090909091)
--(axis cs:-0.36437761823837,8.1010101010101)
--(axis cs:-0.358223505287726,8.29292929292929)
--(axis cs:-0.351936233837377,8.48484848484848)
--(axis cs:-0.345607842828127,8.67676767676768)
--(axis cs:-0.339325769084362,8.86868686868687)
--(axis cs:-0.333170179608677,9.06060606060606)
--(axis cs:-0.327211450974817,9.25252525252525)
--(axis cs:-0.321507901759195,9.44444444444444)
--(axis cs:-0.316103877728837,9.63636363636363)
--(axis cs:-0.311028278411749,9.82828282828283)
--(axis cs:-0.306293598099233,10.020202020202)
--(axis cs:-0.301895534894556,10.2121212121212)
--(axis cs:-0.297813198968403,10.4040404040404)
--(axis cs:-0.294009926722071,10.5959595959596)
--(axis cs:-0.290434682231573,10.7878787878788)
--(axis cs:-0.287024002352298,10.979797979798)
--(axis cs:-0.283704418410589,11.1717171717172)
--(axis cs:-0.280395266640155,11.3636363636364)
--(axis cs:-0.27701178245954,11.5555555555556)
--(axis cs:-0.273468361174368,11.7474747474747)
--(axis cs:-0.269681860340997,11.9393939393939)
--(axis cs:-0.265574817201428,12.1313131313131)
--(axis cs:-0.261078458365787,12.3232323232323)
--(axis cs:-0.256135388065387,12.5151515151515)
--(axis cs:-0.250701855339215,12.7070707070707)
--(axis cs:-0.244749518717262,12.8989898989899)
--(axis cs:-0.23826664839067,13.0909090909091)
--(axis cs:-0.231258729429094,13.2828282828283)
--(axis cs:-0.223748454153861,13.4747474747475)
--(axis cs:-0.215775116118992,13.6666666666667)
--(axis cs:-0.207393441159435,13.8585858585859)
--(axis cs:-0.198671911618716,14.0505050505051)
--(axis cs:-0.189690657315389,14.2424242424242)
--(axis cs:-0.180539000405552,14.4343434343434)
--(axis cs:-0.171312750637746,14.6262626262626)
--(axis cs:-0.162111352410649,14.8181818181818)
--(axis cs:-0.153034985604452,15.010101010101)
--(axis cs:-0.144181718650568,15.2020202020202)
--(axis cs:-0.135644805200025,15.3939393939394)
--(axis cs:-0.127510205655333,15.5858585858586)
--(axis cs:-0.119854402439902,15.7777777777778)
--(axis cs:-0.112742563928049,15.969696969697)
--(axis cs:-0.106227097173434,16.1616161616162)
--(axis cs:-0.100346614629596,16.3535353535354)
--(axis cs:-0.0951253255449266,16.5454545454545)
--(axis cs:-0.0905728491213195,16.7373737373737)
--(axis cs:-0.0866844342189135,16.9292929292929)
--(axis cs:-0.0834415596186902,17.1212121212121)
--(axis cs:-0.0808128797609634,17.3131313131313)
--(axis cs:-0.0787554735087335,17.5050505050505)
--(axis cs:-0.0772163478161467,17.6969696969697)
--(axis cs:-0.0761341441402591,17.8888888888889)
--(axis cs:-0.0754409929180662,18.0808080808081)
--(axis cs:-0.0750644603309261,18.2727272727273)
--(axis cs:-0.0749295317912246,18.4646464646465)
--(axis cs:-0.074960578020993,18.6565656565657)
--(axis cs:-0.0750832521741686,18.8484848484848)
--(axis cs:-0.075226270119828,19.040404040404)
--(axis cs:-0.0753230306929817,19.2323232323232)
--(axis cs:-0.0753130383648085,19.4242424242424)
--(axis cs:-0.0751430972986015,19.6161616161616)
--(axis cs:-0.0747682530243038,19.8080808080808)
--(axis cs:-0.0741524658286663,20)
--(axis cs:0.0741524658286663,20)
--(axis cs:0.0741524658286663,20)
--(axis cs:0.0747682530243038,19.8080808080808)
--(axis cs:0.0751430972986015,19.6161616161616)
--(axis cs:0.0753130383648085,19.4242424242424)
--(axis cs:0.0753230306929817,19.2323232323232)
--(axis cs:0.075226270119828,19.040404040404)
--(axis cs:0.0750832521741686,18.8484848484848)
--(axis cs:0.074960578020993,18.6565656565657)
--(axis cs:0.0749295317912246,18.4646464646465)
--(axis cs:0.0750644603309261,18.2727272727273)
--(axis cs:0.0754409929180662,18.0808080808081)
--(axis cs:0.0761341441402591,17.8888888888889)
--(axis cs:0.0772163478161467,17.6969696969697)
--(axis cs:0.0787554735087335,17.5050505050505)
--(axis cs:0.0808128797609634,17.3131313131313)
--(axis cs:0.0834415596186902,17.1212121212121)
--(axis cs:0.0866844342189135,16.9292929292929)
--(axis cs:0.0905728491213195,16.7373737373737)
--(axis cs:0.0951253255449266,16.5454545454545)
--(axis cs:0.100346614629596,16.3535353535354)
--(axis cs:0.106227097173434,16.1616161616162)
--(axis cs:0.112742563928049,15.969696969697)
--(axis cs:0.119854402439902,15.7777777777778)
--(axis cs:0.127510205655333,15.5858585858586)
--(axis cs:0.135644805200025,15.3939393939394)
--(axis cs:0.144181718650568,15.2020202020202)
--(axis cs:0.153034985604452,15.010101010101)
--(axis cs:0.162111352410649,14.8181818181818)
--(axis cs:0.171312750637746,14.6262626262626)
--(axis cs:0.180539000405552,14.4343434343434)
--(axis cs:0.189690657315389,14.2424242424242)
--(axis cs:0.198671911618716,14.0505050505051)
--(axis cs:0.207393441159435,13.8585858585859)
--(axis cs:0.215775116118992,13.6666666666667)
--(axis cs:0.223748454153861,13.4747474747475)
--(axis cs:0.231258729429094,13.2828282828283)
--(axis cs:0.23826664839067,13.0909090909091)
--(axis cs:0.244749518717262,12.8989898989899)
--(axis cs:0.250701855339215,12.7070707070707)
--(axis cs:0.256135388065387,12.5151515151515)
--(axis cs:0.261078458365787,12.3232323232323)
--(axis cs:0.265574817201428,12.1313131313131)
--(axis cs:0.269681860340997,11.9393939393939)
--(axis cs:0.273468361174368,11.7474747474747)
--(axis cs:0.27701178245954,11.5555555555556)
--(axis cs:0.280395266640155,11.3636363636364)
--(axis cs:0.283704418410589,11.1717171717172)
--(axis cs:0.287024002352298,10.979797979798)
--(axis cs:0.290434682231573,10.7878787878788)
--(axis cs:0.294009926722071,10.5959595959596)
--(axis cs:0.297813198968403,10.4040404040404)
--(axis cs:0.301895534894556,10.2121212121212)
--(axis cs:0.306293598099233,10.020202020202)
--(axis cs:0.311028278411749,9.82828282828283)
--(axis cs:0.316103877728837,9.63636363636363)
--(axis cs:0.321507901759195,9.44444444444444)
--(axis cs:0.327211450974817,9.25252525252525)
--(axis cs:0.333170179608677,9.06060606060606)
--(axis cs:0.339325769084362,8.86868686868687)
--(axis cs:0.345607842828127,8.67676767676768)
--(axis cs:0.351936233837377,8.48484848484848)
--(axis cs:0.358223505287726,8.29292929292929)
--(axis cs:0.36437761823837,8.1010101010101)
--(axis cs:0.370304639271329,7.90909090909091)
--(axis cs:0.375911384546948,7.71717171717172)
--(axis cs:0.38110790490598,7.52525252525253)
--(axis cs:0.385809728711003,7.33333333333333)
--(axis cs:0.389939794331762,7.14141414141414)
--(axis cs:0.393430021642986,6.94949494949495)
--(axis cs:0.396222490643394,6.75757575757576)
--(axis cs:0.398270214322586,6.56565656565656)
--(axis cs:0.399537511234078,6.37373737373737)
--(axis cs:0.4,6.18181818181818)
--(axis cs:0.399644252431104,5.98989898989899)
--(axis cs:0.398467153518402,5.7979797979798)
--(axis cs:0.396475024857801,5.60606060606061)
--(axis cs:0.393682572929148,5.41414141414141)
--(axis cs:0.390111725092538,5.22222222222222)
--(axis cs:0.385790414402312,5.03030303030303)
--(axis cs:0.38075136975047,4.83838383838384)
--(axis cs:0.375030960940922,4.64646464646465)
--(axis cs:0.368668139655505,4.45454545454545)
--(axis cs:0.361703507535267,4.26262626262626)
--(axis cs:0.354178532398979,4.07070707070707)
--(axis cs:0.346134923547337,3.87878787878788)
--(axis cs:0.337614167674887,3.68686868686869)
--(axis cs:0.328657218553599,3.49494949494949)
--(axis cs:0.319304326670639,3.3030303030303)
--(axis cs:0.309594989588885,3.11111111111111)
--(axis cs:0.299568000028695,2.91919191919192)
--(axis cs:0.289261566517662,2.72727272727273)
--(axis cs:0.278713480810858,2.53535353535354)
--(axis cs:0.26796130697276,2.34343434343434)
--(axis cs:0.257042568816978,2.15151515151515)
--(axis cs:0.245994915083779,1.95959595959596)
--(axis cs:0.234856245058647,1.76767676767677)
--(axis cs:0.223664781070795,1.57575757575758)
--(axis cs:0.212459078254162,1.38383838383838)
--(axis cs:0.201277965928942,1.19191919191919)
--(axis cs:0.190160418822847,1)
--cycle;

\path [draw=darkslategray61, fill=peru22412844, line width=0.5pt]
(axis cs:1.19373231176186,0)
--(axis cs:0.806267688238138,0)
--(axis cs:0.795267769079237,0.121212121212121)
--(axis cs:0.784138106581265,0.242424242424242)
--(axis cs:0.77292116008391,0.363636363636364)
--(axis cs:0.761662627472654,0.484848484848485)
--(axis cs:0.750411170040259,0.606060606060606)
--(axis cs:0.739218036717573,0.727272727272727)
--(axis cs:0.728136585377723,0.848484848484849)
--(axis cs:0.717221703318368,0.96969696969697)
--(axis cs:0.706529133810066,1.09090909090909)
--(axis cs:0.696114720522918,1.21212121212121)
--(axis cs:0.686033586438152,1.33333333333333)
--(axis cs:0.67633926823336,1.45454545454545)
--(axis cs:0.667082830821867,1.57575757575758)
--(axis cs:0.658311989473593,1.6969696969697)
--(axis cs:0.650070268532713,1.81818181818182)
--(axis cs:0.642396226019793,1.93939393939394)
--(axis cs:0.635322772275315,2.06060606060606)
--(axis cs:0.628876608259644,2.18181818181818)
--(axis cs:0.623077805245846,2.3030303030303)
--(axis cs:0.617939542583038,2.42424242424242)
--(axis cs:0.613468014200528,2.54545454545455)
--(axis cs:0.609662507862026,2.66666666666667)
--(axis cs:0.60651565420629,2.78787878787879)
--(axis cs:0.604013835694523,2.90909090909091)
--(axis cs:0.602137739098495,3.03030303030303)
--(axis cs:0.600863029460925,3.15151515151515)
--(axis cs:0.600161118853629,3.27272727272727)
--(axis cs:0.6,3.39393939393939)
--(axis cs:0.60034511308858,3.51515151515152)
--(axis cs:0.601160213965804,3.63636363636364)
--(axis cs:0.602408213343861,3.75757575757576)
--(axis cs:0.604051959584117,3.87878787878788)
--(axis cs:0.606054941818409,4)
--(axis cs:0.608381895374533,4.12121212121212)
--(axis cs:0.610999297344842,4.24242424242424)
--(axis cs:0.613875746308447,4.36363636363636)
--(axis cs:0.616982226306454,4.48484848484849)
--(axis cs:0.620292260808342,4.60606060606061)
--(axis cs:0.623781967265942,4.72727272727273)
--(axis cs:0.627430026659485,4.84848484848485)
--(axis cs:0.63121758500507,4.96969696969697)
--(axis cs:0.635128105011312,5.09090909090909)
--(axis cs:0.639147185937228,5.21212121212121)
--(axis cs:0.643262368300396,5.33333333333333)
--(axis cs:0.647462937588415,5.45454545454546)
--(axis cs:0.651739737785818,5.57575757575758)
--(axis cs:0.656085001646066,5.6969696969697)
--(axis cs:0.660492200550245,5.81818181818182)
--(axis cs:0.664955912844594,5.93939393939394)
--(axis cs:0.669471706065311,6.06060606060606)
--(axis cs:0.674036025727329,6.18181818181818)
--(axis cs:0.678646081598794,6.3030303030303)
--(axis cs:0.683299721751696,6.42424242424242)
--(axis cs:0.687995285229917,6.54545454545455)
--(axis cs:0.692731425874644,6.66666666666667)
--(axis cs:0.697506902569194,6.78787878787879)
--(axis cs:0.702320334704703,6.90909090909091)
--(axis cs:0.707169925752274,7.03030303030303)
--(axis cs:0.712053162136028,7.15151515151515)
--(axis cs:0.716966498790497,7.27272727272727)
--(axis cs:0.721905046510503,7.39393939393939)
--(axis cs:0.726862279142221,7.51515151515152)
--(axis cs:0.731829780548076,7.63636363636364)
--(axis cs:0.736797051900241,7.75757575757576)
--(axis cs:0.74175139909466,7.87878787878788)
--(axis cs:0.746677917898103,8)
--(axis cs:0.751559590908205,8.12121212121212)
--(axis cs:0.756377505675625,8.24242424242424)
--(axis cs:0.761111197646335,8.36363636363636)
--(axis cs:0.765739115235964,8.48484848484848)
--(axis cs:0.770239197702182,8.60606060606061)
--(axis cs:0.774589549916885,8.72727272727273)
--(axis cs:0.778769192042841,8.84848484848485)
--(axis cs:0.782758856854528,8.96969696969697)
--(axis cs:0.786541803332732,9.09090909090909)
--(axis cs:0.790104612468144,9.21212121212121)
--(axis cs:0.79343793011436,9.33333333333333)
--(axis cs:0.79653712233183,9.45454545454546)
--(axis cs:0.799402810964875,9.57575757575758)
--(axis cs:0.802041261105445,9.6969696969697)
--(axis cs:0.80446459744487,9.81818181818182)
--(axis cs:0.806690833047806,9.93939393939394)
--(axis cs:0.808743701489437,10.0606060606061)
--(axis cs:0.810652291223142,10.1818181818182)
--(axis cs:0.812450489113215,10.3030303030303)
--(axis cs:0.814176247895023,10.4242424242424)
--(axis cs:0.815870699549359,10.5454545454545)
--(axis cs:0.81757714287025,10.6666666666667)
--(axis cs:0.819339938589148,10.7878787878788)
--(axis cs:0.821203349079944,10.9090909090909)
--(axis cs:0.823210361767987,11.030303030303)
--(axis cs:0.825401535838697,11.1515151515152)
--(axis cs:0.827813910701524,11.2727272727273)
--(axis cs:0.830480012000957,11.3939393939394)
--(axis cs:0.833426986932953,11.5151515151515)
--(axis cs:0.836675895434277,11.6363636363636)
--(axis cs:0.840241177719582,11.7575757575758)
--(axis cs:0.844130311932888,11.8787878787879)
--(axis cs:0.848343668657659,12)
--(axis cs:1.15165633134234,12)
--(axis cs:1.15165633134234,12)
--(axis cs:1.15586968806711,11.8787878787879)
--(axis cs:1.15975882228042,11.7575757575758)
--(axis cs:1.16332410456572,11.6363636363636)
--(axis cs:1.16657301306705,11.5151515151515)
--(axis cs:1.16951998799904,11.3939393939394)
--(axis cs:1.17218608929848,11.2727272727273)
--(axis cs:1.1745984641613,11.1515151515152)
--(axis cs:1.17678963823201,11.030303030303)
--(axis cs:1.17879665092006,10.9090909090909)
--(axis cs:1.18066006141085,10.7878787878788)
--(axis cs:1.18242285712975,10.6666666666667)
--(axis cs:1.18412930045064,10.5454545454545)
--(axis cs:1.18582375210498,10.4242424242424)
--(axis cs:1.18754951088679,10.3030303030303)
--(axis cs:1.18934770877686,10.1818181818182)
--(axis cs:1.19125629851056,10.0606060606061)
--(axis cs:1.19330916695219,9.93939393939394)
--(axis cs:1.19553540255513,9.81818181818182)
--(axis cs:1.19795873889456,9.6969696969697)
--(axis cs:1.20059718903513,9.57575757575758)
--(axis cs:1.20346287766817,9.45454545454546)
--(axis cs:1.20656206988564,9.33333333333333)
--(axis cs:1.20989538753186,9.21212121212121)
--(axis cs:1.21345819666727,9.09090909090909)
--(axis cs:1.21724114314547,8.96969696969697)
--(axis cs:1.22123080795716,8.84848484848485)
--(axis cs:1.22541045008312,8.72727272727273)
--(axis cs:1.22976080229782,8.60606060606061)
--(axis cs:1.23426088476404,8.48484848484848)
--(axis cs:1.23888880235366,8.36363636363636)
--(axis cs:1.24362249432437,8.24242424242424)
--(axis cs:1.2484404090918,8.12121212121212)
--(axis cs:1.2533220821019,8)
--(axis cs:1.25824860090534,7.87878787878788)
--(axis cs:1.26320294809976,7.75757575757576)
--(axis cs:1.26817021945192,7.63636363636364)
--(axis cs:1.27313772085778,7.51515151515152)
--(axis cs:1.2780949534895,7.39393939393939)
--(axis cs:1.2830335012095,7.27272727272727)
--(axis cs:1.28794683786397,7.15151515151515)
--(axis cs:1.29283007424773,7.03030303030303)
--(axis cs:1.2976796652953,6.90909090909091)
--(axis cs:1.30249309743081,6.78787878787879)
--(axis cs:1.30726857412536,6.66666666666667)
--(axis cs:1.31200471477008,6.54545454545455)
--(axis cs:1.3167002782483,6.42424242424242)
--(axis cs:1.32135391840121,6.3030303030303)
--(axis cs:1.32596397427267,6.18181818181818)
--(axis cs:1.33052829393469,6.06060606060606)
--(axis cs:1.33504408715541,5.93939393939394)
--(axis cs:1.33950779944976,5.81818181818182)
--(axis cs:1.34391499835393,5.6969696969697)
--(axis cs:1.34826026221418,5.57575757575758)
--(axis cs:1.35253706241158,5.45454545454546)
--(axis cs:1.3567376316996,5.33333333333333)
--(axis cs:1.36085281406277,5.21212121212121)
--(axis cs:1.36487189498869,5.09090909090909)
--(axis cs:1.36878241499493,4.96969696969697)
--(axis cs:1.37256997334051,4.84848484848485)
--(axis cs:1.37621803273406,4.72727272727273)
--(axis cs:1.37970773919166,4.60606060606061)
--(axis cs:1.38301777369355,4.48484848484849)
--(axis cs:1.38612425369155,4.36363636363636)
--(axis cs:1.38900070265516,4.24242424242424)
--(axis cs:1.39161810462547,4.12121212121212)
--(axis cs:1.39394505818159,4)
--(axis cs:1.39594804041588,3.87878787878788)
--(axis cs:1.39759178665614,3.75757575757576)
--(axis cs:1.3988397860342,3.63636363636364)
--(axis cs:1.39965488691142,3.51515151515152)
--(axis cs:1.4,3.39393939393939)
--(axis cs:1.39983888114637,3.27272727272727)
--(axis cs:1.39913697053908,3.15151515151515)
--(axis cs:1.39786226090151,3.03030303030303)
--(axis cs:1.39598616430548,2.90909090909091)
--(axis cs:1.39348434579371,2.78787878787879)
--(axis cs:1.39033749213797,2.66666666666667)
--(axis cs:1.38653198579947,2.54545454545455)
--(axis cs:1.38206045741696,2.42424242424242)
--(axis cs:1.37692219475415,2.3030303030303)
--(axis cs:1.37112339174036,2.18181818181818)
--(axis cs:1.36467722772469,2.06060606060606)
--(axis cs:1.35760377398021,1.93939393939394)
--(axis cs:1.34992973146729,1.81818181818182)
--(axis cs:1.34168801052641,1.6969696969697)
--(axis cs:1.33291716917813,1.57575757575758)
--(axis cs:1.32366073176664,1.45454545454545)
--(axis cs:1.31396641356185,1.33333333333333)
--(axis cs:1.30388527947708,1.21212121212121)
--(axis cs:1.29347086618993,1.09090909090909)
--(axis cs:1.28277829668163,0.96969696969697)
--(axis cs:1.27186341462228,0.848484848484849)
--(axis cs:1.26078196328243,0.727272727272727)
--(axis cs:1.24958882995974,0.606060606060606)
--(axis cs:1.23833737252735,0.484848484848485)
--(axis cs:1.22707883991609,0.363636363636364)
--(axis cs:1.21586189341874,0.242424242424242)
--(axis cs:1.20473223092076,0.121212121212121)
--(axis cs:1.19373231176186,0)
--cycle;

\path [draw=darkslategray61, fill=seagreen5814558, line width=0.5pt]
(axis cs:2.36558965935376,1)
--(axis cs:1.63441034064624,1)
--(axis cs:1.62815550408336,1.08080808080808)
--(axis cs:1.62252227694445,1.16161616161616)
--(axis cs:1.61752043108051,1.24242424242424)
--(axis cs:1.61315515526662,1.32323232323232)
--(axis cs:1.60942723440034,1.4040404040404)
--(axis cs:1.60633328216751,1.48484848484848)
--(axis cs:1.60386601812318,1.56565656565657)
--(axis cs:1.60201457940038,1.64646464646465)
--(axis cs:1.60076485688366,1.72727272727273)
--(axis cs:1.60009984567196,1.80808080808081)
--(axis cs:1.6,1.88888888888889)
--(axis cs:1.6004435834646,1.96969696969697)
--(axis cs:1.60140700637792,2.05050505050505)
--(axis cs:1.60286514329563,2.13131313131313)
--(axis cs:1.60479162518827,2.21212121212121)
--(axis cs:1.60715910227514,2.29292929292929)
--(axis cs:1.6099394751545,2.37373737373737)
--(axis cs:1.61310409347263,2.45454545454545)
--(axis cs:1.61662392291202,2.53535353535354)
--(axis cs:1.62046968268143,2.61616161616162)
--(axis cs:1.6246119569054,2.6969696969697)
--(axis cs:1.62902128429108,2.77777777777778)
--(axis cs:1.63366823116301,2.85858585858586)
--(axis cs:1.63852345337986,2.93939393939394)
--(axis cs:1.64355775277358,3.02020202020202)
--(axis cs:1.64874213358492,3.1010101010101)
--(axis cs:1.65404786392836,3.18181818181818)
--(axis cs:1.65944654663265,3.26262626262626)
--(axis cs:1.66491020290939,3.34343434343434)
--(axis cs:1.67041137124872,3.42424242424242)
--(axis cs:1.6759232227798,3.50505050505051)
--(axis cs:1.68141969312094,3.58585858585859)
--(axis cs:1.68687562953531,3.66666666666667)
--(axis cs:1.69226695105837,3.74747474747475)
--(axis cs:1.69757081822174,3.82828282828283)
--(axis cs:1.70276580811049,3.90909090909091)
--(axis cs:1.70783208979221,3.98989898989899)
--(axis cs:1.71275159467464,4.07070707070707)
--(axis cs:1.71750817610119,4.15151515151515)
--(axis cs:1.7220877524882,4.23232323232323)
--(axis cs:1.7264784285411,4.31313131313131)
--(axis cs:1.73067058954697,4.39393939393939)
--(axis cs:1.73465696440704,4.47474747474747)
--(axis cs:1.73843265391683,4.55555555555556)
--(axis cs:1.74199512178915,4.63636363636364)
--(axis cs:1.74534414700796,4.71717171717172)
--(axis cs:1.74848173725798,4.7979797979798)
--(axis cs:1.75141200435349,4.87878787878788)
--(axis cs:1.75414100374907,4.95959595959596)
--(axis cs:1.75667654131525,5.04040404040404)
--(axis cs:1.75902795156744,5.12121212121212)
--(axis cs:1.76120585241535,5.2020202020202)
--(axis cs:1.76322188222645,5.28282828282828)
--(axis cs:1.76508842554985,5.36363636363636)
--(axis cs:1.76681833421257,5.44444444444444)
--(axis cs:1.76842465066934,5.52525252525253)
--(axis cs:1.76992034045923,5.60606060606061)
--(axis cs:1.77131804039926,5.68686868686869)
--(axis cs:1.77262982873731,5.76767676767677)
--(axis cs:1.77386702290692,5.84848484848485)
--(axis cs:1.77504000979267,5.92929292929293)
--(axis cs:1.77615811254862,6.01010101010101)
--(axis cs:1.77722949703698,6.09090909090909)
--(axis cs:1.7782611198968,6.17171717171717)
--(axis cs:1.77925871913975,6.25252525252525)
--(axis cs:1.78022684703105,6.33333333333333)
--(axis cs:1.78116894387665,6.41414141414141)
--(axis cs:1.78208745023164,6.4949494949495)
--(axis cs:1.78298395399662,6.57575757575758)
--(axis cs:1.78385936790557,6.65656565656566)
--(axis cs:1.7847141320539,6.73737373737374)
--(axis cs:1.7855484353927,6.81818181818182)
--(axis cs:1.78636244954236,6.8989898989899)
--(axis cs:1.78715656787408,6.97979797979798)
--(axis cs:1.78793164258201,7.06060606060606)
--(axis cs:1.78868921243225,7.14141414141414)
--(axis cs:1.78943171403095,7.22222222222222)
--(axis cs:1.79016266980218,7.3030303030303)
--(axis cs:1.79088684640196,7.38383838383838)
--(axis cs:1.79161037800748,7.46464646464647)
--(axis cs:1.79234084979453,7.54545454545455)
--(axis cs:1.79308733793201,7.62626262626263)
--(axis cs:1.7938604035549,7.70707070707071)
--(axis cs:1.79467203939832,7.78787878787879)
--(axis cs:1.79553556905236,7.86868686868687)
--(axis cs:1.79646550009722,7.94949494949495)
--(axis cs:1.79747733366306,8.03030303030303)
--(axis cs:1.79858733419364,8.11111111111111)
--(axis cs:1.79981226434051,8.19191919191919)
--(axis cs:1.80116909094082,8.27272727272727)
--(axis cs:1.80267466890607,8.35353535353535)
--(axis cs:1.80434541054319,8.43434343434344)
--(axis cs:1.8061969483214,8.51515151515152)
--(axis cs:1.80824379937204,8.5959595959596)
--(axis cs:1.81049904005363,8.67676767676768)
--(axis cs:1.81297399872871,8.75757575757576)
--(axis cs:1.81567797448675,8.83838383838384)
--(axis cs:1.81861798891976,8.91919191919192)
--(axis cs:1.82179857723372,9)
--(axis cs:2.17820142276628,9)
--(axis cs:2.17820142276628,9)
--(axis cs:2.18138201108024,8.91919191919192)
--(axis cs:2.18432202551325,8.83838383838384)
--(axis cs:2.18702600127129,8.75757575757576)
--(axis cs:2.18950095994637,8.67676767676768)
--(axis cs:2.19175620062796,8.5959595959596)
--(axis cs:2.1938030516786,8.51515151515152)
--(axis cs:2.19565458945681,8.43434343434344)
--(axis cs:2.19732533109393,8.35353535353535)
--(axis cs:2.19883090905918,8.27272727272727)
--(axis cs:2.20018773565949,8.19191919191919)
--(axis cs:2.20141266580636,8.11111111111111)
--(axis cs:2.20252266633694,8.03030303030303)
--(axis cs:2.20353449990278,7.94949494949495)
--(axis cs:2.20446443094764,7.86868686868687)
--(axis cs:2.20532796060168,7.78787878787879)
--(axis cs:2.2061395964451,7.70707070707071)
--(axis cs:2.20691266206799,7.62626262626263)
--(axis cs:2.20765915020547,7.54545454545455)
--(axis cs:2.20838962199252,7.46464646464647)
--(axis cs:2.20911315359804,7.38383838383838)
--(axis cs:2.20983733019782,7.3030303030303)
--(axis cs:2.21056828596905,7.22222222222222)
--(axis cs:2.21131078756775,7.14141414141414)
--(axis cs:2.21206835741799,7.06060606060606)
--(axis cs:2.21284343212592,6.97979797979798)
--(axis cs:2.21363755045764,6.8989898989899)
--(axis cs:2.2144515646073,6.81818181818182)
--(axis cs:2.2152858679461,6.73737373737374)
--(axis cs:2.21614063209443,6.65656565656566)
--(axis cs:2.21701604600338,6.57575757575758)
--(axis cs:2.21791254976836,6.4949494949495)
--(axis cs:2.21883105612335,6.41414141414141)
--(axis cs:2.21977315296895,6.33333333333333)
--(axis cs:2.22074128086025,6.25252525252525)
--(axis cs:2.2217388801032,6.17171717171717)
--(axis cs:2.22277050296302,6.09090909090909)
--(axis cs:2.22384188745138,6.01010101010101)
--(axis cs:2.22495999020733,5.92929292929293)
--(axis cs:2.22613297709308,5.84848484848485)
--(axis cs:2.22737017126269,5.76767676767677)
--(axis cs:2.22868195960074,5.68686868686869)
--(axis cs:2.23007965954077,5.60606060606061)
--(axis cs:2.23157534933067,5.52525252525253)
--(axis cs:2.23318166578743,5.44444444444444)
--(axis cs:2.23491157445015,5.36363636363636)
--(axis cs:2.23677811777355,5.28282828282828)
--(axis cs:2.23879414758465,5.2020202020202)
--(axis cs:2.24097204843256,5.12121212121212)
--(axis cs:2.24332345868475,5.04040404040404)
--(axis cs:2.24585899625093,4.95959595959596)
--(axis cs:2.24858799564651,4.87878787878788)
--(axis cs:2.25151826274202,4.7979797979798)
--(axis cs:2.25465585299204,4.71717171717172)
--(axis cs:2.25800487821085,4.63636363636364)
--(axis cs:2.26156734608317,4.55555555555556)
--(axis cs:2.26534303559296,4.47474747474747)
--(axis cs:2.26932941045303,4.39393939393939)
--(axis cs:2.2735215714589,4.31313131313131)
--(axis cs:2.2779122475118,4.23232323232323)
--(axis cs:2.28249182389881,4.15151515151515)
--(axis cs:2.28724840532536,4.07070707070707)
--(axis cs:2.29216791020779,3.98989898989899)
--(axis cs:2.29723419188951,3.90909090909091)
--(axis cs:2.30242918177826,3.82828282828283)
--(axis cs:2.30773304894163,3.74747474747475)
--(axis cs:2.31312437046469,3.66666666666667)
--(axis cs:2.31858030687906,3.58585858585859)
--(axis cs:2.3240767772202,3.50505050505051)
--(axis cs:2.32958862875128,3.42424242424242)
--(axis cs:2.33508979709061,3.34343434343434)
--(axis cs:2.34055345336735,3.26262626262626)
--(axis cs:2.34595213607164,3.18181818181818)
--(axis cs:2.35125786641508,3.1010101010101)
--(axis cs:2.35644224722642,3.02020202020202)
--(axis cs:2.36147654662014,2.93939393939394)
--(axis cs:2.36633176883699,2.85858585858586)
--(axis cs:2.37097871570892,2.77777777777778)
--(axis cs:2.3753880430946,2.6969696969697)
--(axis cs:2.37953031731857,2.61616161616162)
--(axis cs:2.38337607708798,2.53535353535354)
--(axis cs:2.38689590652737,2.45454545454545)
--(axis cs:2.3900605248455,2.37373737373737)
--(axis cs:2.39284089772486,2.29292929292929)
--(axis cs:2.39520837481173,2.21212121212121)
--(axis cs:2.39713485670437,2.13131313131313)
--(axis cs:2.39859299362208,2.05050505050505)
--(axis cs:2.3995564165354,1.96969696969697)
--(axis cs:2.4,1.88888888888889)
--(axis cs:2.39990015432804,1.80808080808081)
--(axis cs:2.39923514311634,1.72727272727273)
--(axis cs:2.39798542059962,1.64646464646465)
--(axis cs:2.39613398187682,1.56565656565657)
--(axis cs:2.39366671783249,1.48484848484848)
--(axis cs:2.39057276559966,1.4040404040404)
--(axis cs:2.38684484473338,1.32323232323232)
--(axis cs:2.38247956891949,1.24242424242424)
--(axis cs:2.37747772305555,1.16161616161616)
--(axis cs:2.37184449591664,1.08080808080808)
--(axis cs:2.36558965935376,1)
--cycle;

\addplot [line width=0.75pt, darkslategray61]
table {%
0 1
0 20
};
\addplot [line width=2.25pt, darkslategray61]
table {%
0 4.75
0 12.25
};
\addplot [semithick, darkslategray61, mark=-, mark size=2.34375, mark options={solid,fill=white}]
table {%
0 7
};
\addplot [line width=0.75pt, darkslategray61]
table {%
1 0
1 12
};
\addplot [line width=2.25pt, darkslategray61]
table {%
1 2.75
1 8
};
\addplot [semithick, darkslategray61, mark=-, mark size=2.34375, mark options={solid,fill=white}]
table {%
1 5
};
\addplot [line width=0.75pt, darkslategray61]
table {%
2 1
2 9
};
\addplot [line width=2.25pt, darkslategray61]
table {%
2 1
2 6.25
};
\addplot [semithick, darkslategray61, mark=-, mark size=2.34375, mark options={solid,fill=white}]
table {%
2 3
};
\end{axis}

\end{tikzpicture}

%% file: figures/average_demands_per_topo.tex
\begin{tikzpicture}

\definecolor{darkgray176}{RGB}{176,176,176}
\definecolor{darkslategray61}{RGB}{61,61,61}
\definecolor{peru22412844}{RGB}{224,128,44}
\definecolor{seagreen5814558}{RGB}{58,145,58}
\definecolor{steelblue49115161}{RGB}{49,115,161}

\begin{axis}[
tick align=outside,
tick pos=left,
x grid style={darkgray176},
xmin=-0.5, xmax=2.5,
ymin=0, ymax=150,
xtick style={color=black},
xtick={0,1,2},
xticklabels={MTR,vMTR\\(real),vMTR\\(virtual)},
xticklabel style = {align=center},
y grid style={darkgray176},
ylabel={Number of demands},
ytick style={color=black}
]
\path [draw=darkslategray61, fill=steelblue49115161, line width=0.5pt]
(axis cs:0.344911263181135,3.33333333333333)
--(axis cs:-0.344911263181135,3.33333333333333)
--(axis cs:-0.354562167455376,3.7996632996633)
--(axis cs:-0.363331562930168,4.26599326599327)
--(axis cs:-0.371177907135602,4.73232323232323)
--(axis cs:-0.378072467934707,5.1986531986532)
--(axis cs:-0.383999704150968,5.66498316498317)
--(axis cs:-0.388957348565389,6.13131313131313)
--(axis cs:-0.392956190665741,6.5976430976431)
--(axis cs:-0.396019565051227,7.06397306397306)
--(axis cs:-0.398182559689181,7.53030303030303)
--(axis cs:-0.399490966022534,7.996632996633)
--(axis cs:-0.4,8.46296296296296)
--(axis cs:-0.399772829241495,8.92929292929293)
--(axis cs:-0.398878946593794,9.3956228956229)
--(axis cs:-0.397392434151638,9.86195286195286)
--(axis cs:-0.395390164334963,10.3282828282828)
--(axis cs:-0.392949985781664,10.7946127946128)
--(axis cs:-0.390148941633869,11.2609427609428)
--(axis cs:-0.387061566295483,11.7272727272727)
--(axis cs:-0.383758303983025,12.1936026936027)
--(axis cs:-0.380304088471928,12.6599326599327)
--(axis cs:-0.376757118472774,13.1262626262626)
--(axis cs:-0.373167857195077,13.5925925925926)
--(axis cs:-0.369578278028686,14.0589225589226)
--(axis cs:-0.366021371070962,14.5252525252525)
--(axis cs:-0.362520917643641,14.9915824915825)
--(axis cs:-0.359091532182038,15.4579124579125)
--(axis cs:-0.355738963156863,15.9242424242424)
--(axis cs:-0.352460637227356,16.3905723905724)
--(axis cs:-0.349246423847412,16.8569023569024)
--(axis cs:-0.34607959127246,17.3232323232323)
--(axis cs:-0.342937919550602,17.7895622895623)
--(axis cs:-0.339794931813124,18.2558922558923)
--(axis cs:-0.336621202164182,18.7222222222222)
--(axis cs:-0.333385696827098,19.1885521885522)
--(axis cs:-0.330057105009973,19.6548821548822)
--(axis cs:-0.326605117229497,20.1212121212121)
--(axis cs:-0.323001611546388,20.5875420875421)
--(axis cs:-0.319221712228849,21.0538720538721)
--(axis cs:-0.315244690625504,21.520202020202)
--(axis cs:-0.311054684298149,21.986531986532)
--(axis cs:-0.306641217494457,22.452861952862)
--(axis cs:-0.301999513553871,22.9191919191919)
--(axis cs:-0.297130597536835,23.3855218855219)
--(axis cs:-0.292041194940081,23.8518518518519)
--(axis cs:-0.286743439507551,24.3181818181818)
--(axis cs:-0.281254409587007,24.7845117845118)
--(axis cs:-0.27559551797025,25.2508417508417)
--(axis cs:-0.269791784489919,25.7171717171717)
--(axis cs:-0.263871023682683,26.1835016835017)
--(axis cs:-0.25786298148297,26.6498316498316)
--(axis cs:-0.251798455162082,27.1161616161616)
--(axis cs:-0.245708429616088,27.5824915824916)
--(axis cs:-0.239623260732197,28.0488215488215)
--(axis cs:-0.233571933078257,28.5151515151515)
--(axis cs:-0.227581414756828,28.9814814814815)
--(axis cs:-0.2216761271676,29.4478114478114)
--(axis cs:-0.215877541872806,29.9141414141414)
--(axis cs:-0.210203911009242,30.3804713804714)
--(axis cs:-0.204670131982684,30.8468013468013)
--(axis cs:-0.199287741745495,31.3131313131313)
--(axis cs:-0.194065031001787,31.7794612794613)
--(axis cs:-0.189007264381036,32.2457912457912)
--(axis cs:-0.18411698910882,32.7121212121212)
--(axis cs:-0.17939441208149,33.1784511784512)
--(axis cs:-0.17483782357893,33.6447811447811)
--(axis cs:-0.170444045145443,34.1111111111111)
--(axis cs:-0.166208879414961,34.5774410774411)
--(axis cs:-0.162127540800676,35.043771043771)
--(axis cs:-0.158195047927983,35.510101010101)
--(axis cs:-0.154406561354683,35.976430976431)
--(axis cs:-0.150757653364417,36.4427609427609)
--(axis cs:-0.147244500293001,36.9090909090909)
--(axis cs:-0.143863991797143,37.3754208754209)
--(axis cs:-0.140613755539604,37.8417508417508)
--(axis cs:-0.137492099782734,38.3080808080808)
--(axis cs:-0.134497880196423,38.7744107744108)
--(axis cs:-0.131630300649571,39.2407407407407)
--(axis cs:-0.128888660732917,39.7070707070707)
--(axis cs:-0.126272065141422,40.1734006734007)
--(axis cs:-0.123779111734923,40.6397306397306)
--(axis cs:-0.121407576031996,41.1060606060606)
--(axis cs:-0.119154110038695,41.5723905723906)
--(axis cs:-0.11701397266779,42.0387205387205)
--(axis cs:-0.114980807594079,42.5050505050505)
--(axis cs:-0.113046482278604,42.9713804713805)
--(axis cs:-0.111200999170419,43.4377104377104)
--(axis cs:-0.109432486877398,43.9040404040404)
--(axis cs:-0.107727275528413,44.3703703703704)
--(axis cs:-0.106070056785513,44.8367003367003)
--(axis cs:-0.104444125173543,45.3030303030303)
--(axis cs:-0.10283169374529,45.7693602693603)
--(axis cs:-0.101214273755736,46.2356902356902)
--(axis cs:-0.0995731051289275,46.7020202020202)
--(axis cs:-0.0978896221934358,47.1683501683502)
--(axis cs:-0.0961459375382251,47.6346801346801)
--(axis cs:-0.0943253259681665,48.1010101010101)
--(axis cs:-0.0924126904498307,48.5673400673401)
--(axis cs:-0.0903949926286661,49.03367003367)
--(axis cs:-0.0882616319265345,49.5)
--(axis cs:0.0882616319265345,49.5)
--(axis cs:0.0882616319265345,49.5)
--(axis cs:0.0903949926286661,49.03367003367)
--(axis cs:0.0924126904498307,48.5673400673401)
--(axis cs:0.0943253259681665,48.1010101010101)
--(axis cs:0.0961459375382251,47.6346801346801)
--(axis cs:0.0978896221934358,47.1683501683502)
--(axis cs:0.0995731051289275,46.7020202020202)
--(axis cs:0.101214273755736,46.2356902356902)
--(axis cs:0.10283169374529,45.7693602693603)
--(axis cs:0.104444125173543,45.3030303030303)
--(axis cs:0.106070056785513,44.8367003367003)
--(axis cs:0.107727275528413,44.3703703703704)
--(axis cs:0.109432486877398,43.9040404040404)
--(axis cs:0.111200999170419,43.4377104377104)
--(axis cs:0.113046482278604,42.9713804713805)
--(axis cs:0.114980807594079,42.5050505050505)
--(axis cs:0.11701397266779,42.0387205387205)
--(axis cs:0.119154110038695,41.5723905723906)
--(axis cs:0.121407576031996,41.1060606060606)
--(axis cs:0.123779111734923,40.6397306397306)
--(axis cs:0.126272065141422,40.1734006734007)
--(axis cs:0.128888660732917,39.7070707070707)
--(axis cs:0.131630300649571,39.2407407407407)
--(axis cs:0.134497880196423,38.7744107744108)
--(axis cs:0.137492099782734,38.3080808080808)
--(axis cs:0.140613755539604,37.8417508417508)
--(axis cs:0.143863991797143,37.3754208754209)
--(axis cs:0.147244500293001,36.9090909090909)
--(axis cs:0.150757653364417,36.4427609427609)
--(axis cs:0.154406561354683,35.976430976431)
--(axis cs:0.158195047927983,35.510101010101)
--(axis cs:0.162127540800676,35.043771043771)
--(axis cs:0.166208879414961,34.5774410774411)
--(axis cs:0.170444045145443,34.1111111111111)
--(axis cs:0.17483782357893,33.6447811447811)
--(axis cs:0.17939441208149,33.1784511784512)
--(axis cs:0.18411698910882,32.7121212121212)
--(axis cs:0.189007264381036,32.2457912457912)
--(axis cs:0.194065031001787,31.7794612794613)
--(axis cs:0.199287741745495,31.3131313131313)
--(axis cs:0.204670131982684,30.8468013468013)
--(axis cs:0.210203911009242,30.3804713804714)
--(axis cs:0.215877541872806,29.9141414141414)
--(axis cs:0.2216761271676,29.4478114478114)
--(axis cs:0.227581414756828,28.9814814814815)
--(axis cs:0.233571933078257,28.5151515151515)
--(axis cs:0.239623260732197,28.0488215488215)
--(axis cs:0.245708429616088,27.5824915824916)
--(axis cs:0.251798455162082,27.1161616161616)
--(axis cs:0.25786298148297,26.6498316498316)
--(axis cs:0.263871023682683,26.1835016835017)
--(axis cs:0.269791784489919,25.7171717171717)
--(axis cs:0.27559551797025,25.2508417508417)
--(axis cs:0.281254409587007,24.7845117845118)
--(axis cs:0.286743439507551,24.3181818181818)
--(axis cs:0.292041194940081,23.8518518518519)
--(axis cs:0.297130597536835,23.3855218855219)
--(axis cs:0.301999513553871,22.9191919191919)
--(axis cs:0.306641217494457,22.452861952862)
--(axis cs:0.311054684298149,21.986531986532)
--(axis cs:0.315244690625504,21.520202020202)
--(axis cs:0.319221712228849,21.0538720538721)
--(axis cs:0.323001611546388,20.5875420875421)
--(axis cs:0.326605117229497,20.1212121212121)
--(axis cs:0.330057105009973,19.6548821548822)
--(axis cs:0.333385696827098,19.1885521885522)
--(axis cs:0.336621202164182,18.7222222222222)
--(axis cs:0.339794931813124,18.2558922558923)
--(axis cs:0.342937919550602,17.7895622895623)
--(axis cs:0.34607959127246,17.3232323232323)
--(axis cs:0.349246423847412,16.8569023569024)
--(axis cs:0.352460637227356,16.3905723905724)
--(axis cs:0.355738963156863,15.9242424242424)
--(axis cs:0.359091532182038,15.4579124579125)
--(axis cs:0.362520917643641,14.9915824915825)
--(axis cs:0.366021371070962,14.5252525252525)
--(axis cs:0.369578278028686,14.0589225589226)
--(axis cs:0.373167857195077,13.5925925925926)
--(axis cs:0.376757118472774,13.1262626262626)
--(axis cs:0.380304088471928,12.6599326599327)
--(axis cs:0.383758303983025,12.1936026936027)
--(axis cs:0.387061566295483,11.7272727272727)
--(axis cs:0.390148941633869,11.2609427609428)
--(axis cs:0.392949985781664,10.7946127946128)
--(axis cs:0.395390164334963,10.3282828282828)
--(axis cs:0.397392434151638,9.86195286195286)
--(axis cs:0.398878946593794,9.3956228956229)
--(axis cs:0.399772829241495,8.92929292929293)
--(axis cs:0.4,8.46296296296296)
--(axis cs:0.399490966022534,7.996632996633)
--(axis cs:0.398182559689181,7.53030303030303)
--(axis cs:0.396019565051227,7.06397306397306)
--(axis cs:0.392956190665741,6.5976430976431)
--(axis cs:0.388957348565389,6.13131313131313)
--(axis cs:0.383999704150968,5.66498316498317)
--(axis cs:0.378072467934707,5.1986531986532)
--(axis cs:0.371177907135602,4.73232323232323)
--(axis cs:0.363331562930168,4.26599326599327)
--(axis cs:0.354562167455376,3.7996632996633)
--(axis cs:0.344911263181135,3.33333333333333)
--cycle;

\path [draw=darkslategray61, fill=peru22412844, line width=0.5pt]
(axis cs:1.23356656600741,0)
--(axis cs:0.766433433992589,0)
--(axis cs:0.756786286544108,0.128787878787879)
--(axis cs:0.747778058354191,0.257575757575758)
--(axis cs:0.739478045538586,0.386363636363636)
--(axis cs:0.731941661606473,0.515151515151515)
--(axis cs:0.725208498883977,0.643939393939394)
--(axis cs:0.719300758092229,0.772727272727273)
--(axis cs:0.714222119095976,0.901515151515151)
--(axis cs:0.709957114046175,1.03030303030303)
--(axis cs:0.706471049263253,1.15909090909091)
--(axis cs:0.703710504734094,1.28787878787879)
--(axis cs:0.701604420657758,1.41666666666667)
--(axis cs:0.700065759823198,1.54545454545455)
--(axis cs:0.698993713570561,1.67424242424242)
--(axis cs:0.69827639855266,1.8030303030303)
--(axis cs:0.697793972350801,1.93181818181818)
--(axis cs:0.697422079040176,2.06060606060606)
--(axis cs:0.697035521788635,2.18939393939394)
--(axis cs:0.696512049129578,2.31818181818182)
--(axis cs:0.695736135141757,2.4469696969697)
--(axis cs:0.694602631688892,2.57575757575758)
--(axis cs:0.693020173226623,2.70454545454545)
--(axis cs:0.690914221393464,2.83333333333333)
--(axis cs:0.68822964740706,2.96212121212121)
--(axis cs:0.68493276476674,3.09090909090909)
--(axis cs:0.681012742359769,3.21969696969697)
--(axis cs:0.676482348114772,3.34848484848485)
--(axis cs:0.671377995097521,3.47727272727273)
--(axis cs:0.665759084613484,3.60606060606061)
--(axis cs:0.659706663667738,3.73484848484848)
--(axis cs:0.653321436252589,3.86363636363636)
--(axis cs:0.646721188646026,3.99242424242424)
--(axis cs:0.640037707534132,4.12121212121212)
--(axis cs:0.633413285724368,4.25)
--(axis cs:0.626996922996848,4.37878787878788)
--(axis cs:0.62094033885589,4.50757575757576)
--(axis cs:0.615393919316467,4.63636363636364)
--(axis cs:0.610502721230027,4.76515151515151)
--(axis cs:0.606402654983517,4.89393939393939)
--(axis cs:0.603216959779143,5.02272727272727)
--(axis cs:0.601053075328362,5.15151515151515)
--(axis cs:0.6,5.28030303030303)
--(axis cs:0.600126208693735,5.40909090909091)
--(axis cs:0.601478184518399,5.53787878787879)
--(axis cs:0.604079597387717,5.66666666666667)
--(axis cs:0.607931140631377,5.79545454545454)
--(axis cs:0.613011014441867,5.92424242424242)
--(axis cs:0.619276023254049,6.0530303030303)
--(axis cs:0.626663233803021,6.18181818181818)
--(axis cs:0.635092122412104,6.31060606060606)
--(axis cs:0.64446712474466,6.43939393939394)
--(axis cs:0.654680489426254,6.56818181818182)
--(axis cs:0.665615329086862,6.6969696969697)
--(axis cs:0.677148758802886,6.82575757575758)
--(axis cs:0.689155012767013,6.95454545454545)
--(axis cs:0.701508435214956,7.08333333333333)
--(axis cs:0.714086250926807,7.21212121212121)
--(axis cs:0.726771033541262,7.34090909090909)
--(axis cs:0.7394528058461,7.46969696969697)
--(axis cs:0.752030724369216,7.59848484848485)
--(axis cs:0.764414320119232,7.72727272727273)
--(axis cs:0.776524287282859,7.85606060606061)
--(axis cs:0.788292831138193,7.98484848484848)
--(axis cs:0.799663604489437,8.11363636363636)
--(axis cs:0.810591277757368,8.24242424242424)
--(axis cs:0.821040800789262,8.37121212121212)
--(axis cs:0.830986423964124,8.5)
--(axis cs:0.840410551934862,8.62878787878788)
--(axis cs:0.849302505241017,8.75757575757576)
--(axis cs:0.857657263119722,8.88636363636364)
--(axis cs:0.865474255406916,9.01515151515152)
--(axis cs:0.872756262895926,9.14393939393939)
--(axis cs:0.879508474488023,9.27272727272727)
--(axis cs:0.885737736616454,9.40151515151515)
--(axis cs:0.89145201650218,9.53030303030303)
--(axis cs:0.896660086577787,9.65909090909091)
--(axis cs:0.901371423646475,9.78787878787879)
--(axis cs:0.905596303717038,9.91666666666667)
--(axis cs:0.909346062572265,10.0454545454545)
--(axis cs:0.912633483467125,10.1742424242424)
--(axis cs:0.915473267257055,10.3030303030303)
--(axis cs:0.917882536920664,10.4318181818182)
--(axis cs:0.91988132791036,10.5606060606061)
--(axis cs:0.921493017940079,10.6893939393939)
--(axis cs:0.922744654471096,10.8181818181818)
--(axis cs:0.923667144940985,10.9469696969697)
--(axis cs:0.924295283262493,11.0757575757576)
--(axis cs:0.924667595798237,11.2045454545455)
--(axis cs:0.924826000355453,11.3333333333333)
--(axis cs:0.924815282193438,11.4621212121212)
--(axis cs:0.924682401061801,11.5909090909091)
--(axis cs:0.924475652397818,11.719696969697)
--(axis cs:0.92424371357632,11.8484848484848)
--(axis cs:0.924034612176595,11.9772727272727)
--(axis cs:0.92389465735369,12.1060606060606)
--(axis cs:0.923867377425857,12.2348484848485)
--(axis cs:0.923992506673563,12.3636363636364)
--(axis cs:0.924305062153339,12.4924242424242)
--(axis cs:0.924834547227209,12.6212121212121)
--(axis cs:0.925604312751207,12.75)
--(axis cs:1.07439568724879,12.75)
--(axis cs:1.07439568724879,12.75)
--(axis cs:1.07516545277279,12.6212121212121)
--(axis cs:1.07569493784666,12.4924242424242)
--(axis cs:1.07600749332644,12.3636363636364)
--(axis cs:1.07613262257414,12.2348484848485)
--(axis cs:1.07610534264631,12.1060606060606)
--(axis cs:1.07596538782341,11.9772727272727)
--(axis cs:1.07575628642368,11.8484848484848)
--(axis cs:1.07552434760218,11.719696969697)
--(axis cs:1.0753175989382,11.5909090909091)
--(axis cs:1.07518471780656,11.4621212121212)
--(axis cs:1.07517399964455,11.3333333333333)
--(axis cs:1.07533240420176,11.2045454545455)
--(axis cs:1.07570471673751,11.0757575757576)
--(axis cs:1.07633285505901,10.9469696969697)
--(axis cs:1.0772553455289,10.8181818181818)
--(axis cs:1.07850698205992,10.6893939393939)
--(axis cs:1.08011867208964,10.5606060606061)
--(axis cs:1.08211746307934,10.4318181818182)
--(axis cs:1.08452673274295,10.3030303030303)
--(axis cs:1.08736651653287,10.1742424242424)
--(axis cs:1.09065393742773,10.0454545454545)
--(axis cs:1.09440369628296,9.91666666666667)
--(axis cs:1.09862857635352,9.78787878787879)
--(axis cs:1.10333991342221,9.65909090909091)
--(axis cs:1.10854798349782,9.53030303030303)
--(axis cs:1.11426226338355,9.40151515151515)
--(axis cs:1.12049152551198,9.27272727272727)
--(axis cs:1.12724373710407,9.14393939393939)
--(axis cs:1.13452574459308,9.01515151515152)
--(axis cs:1.14234273688028,8.88636363636364)
--(axis cs:1.15069749475898,8.75757575757576)
--(axis cs:1.15958944806514,8.62878787878788)
--(axis cs:1.16901357603588,8.5)
--(axis cs:1.17895919921074,8.37121212121212)
--(axis cs:1.18940872224263,8.24242424242424)
--(axis cs:1.20033639551056,8.11363636363636)
--(axis cs:1.21170716886181,7.98484848484848)
--(axis cs:1.22347571271714,7.85606060606061)
--(axis cs:1.23558567988077,7.72727272727273)
--(axis cs:1.24796927563078,7.59848484848485)
--(axis cs:1.2605471941539,7.46969696969697)
--(axis cs:1.27322896645874,7.34090909090909)
--(axis cs:1.28591374907319,7.21212121212121)
--(axis cs:1.29849156478504,7.08333333333333)
--(axis cs:1.31084498723299,6.95454545454545)
--(axis cs:1.32285124119711,6.82575757575758)
--(axis cs:1.33438467091314,6.6969696969697)
--(axis cs:1.34531951057375,6.56818181818182)
--(axis cs:1.35553287525534,6.43939393939394)
--(axis cs:1.3649078775879,6.31060606060606)
--(axis cs:1.37333676619698,6.18181818181818)
--(axis cs:1.38072397674595,6.0530303030303)
--(axis cs:1.38698898555813,5.92424242424242)
--(axis cs:1.39206885936862,5.79545454545454)
--(axis cs:1.39592040261228,5.66666666666667)
--(axis cs:1.3985218154816,5.53787878787879)
--(axis cs:1.39987379130627,5.40909090909091)
--(axis cs:1.4,5.28030303030303)
--(axis cs:1.39894692467164,5.15151515151515)
--(axis cs:1.39678304022086,5.02272727272727)
--(axis cs:1.39359734501648,4.89393939393939)
--(axis cs:1.38949727876997,4.76515151515151)
--(axis cs:1.38460608068353,4.63636363636364)
--(axis cs:1.37905966114411,4.50757575757576)
--(axis cs:1.37300307700315,4.37878787878788)
--(axis cs:1.36658671427563,4.25)
--(axis cs:1.35996229246587,4.12121212121212)
--(axis cs:1.35327881135397,3.99242424242424)
--(axis cs:1.34667856374741,3.86363636363636)
--(axis cs:1.34029333633226,3.73484848484848)
--(axis cs:1.33424091538652,3.60606060606061)
--(axis cs:1.32862200490248,3.47727272727273)
--(axis cs:1.32351765188523,3.34848484848485)
--(axis cs:1.31898725764023,3.21969696969697)
--(axis cs:1.31506723523326,3.09090909090909)
--(axis cs:1.31177035259294,2.96212121212121)
--(axis cs:1.30908577860654,2.83333333333333)
--(axis cs:1.30697982677338,2.70454545454545)
--(axis cs:1.30539736831111,2.57575757575758)
--(axis cs:1.30426386485824,2.4469696969697)
--(axis cs:1.30348795087042,2.31818181818182)
--(axis cs:1.30296447821136,2.18939393939394)
--(axis cs:1.30257792095982,2.06060606060606)
--(axis cs:1.3022060276492,1.93181818181818)
--(axis cs:1.30172360144734,1.8030303030303)
--(axis cs:1.30100628642944,1.67424242424242)
--(axis cs:1.2999342401768,1.54545454545455)
--(axis cs:1.29839557934224,1.41666666666667)
--(axis cs:1.29628949526591,1.28787878787879)
--(axis cs:1.29352895073675,1.15909090909091)
--(axis cs:1.29004288595383,1.03030303030303)
--(axis cs:1.28577788090402,0.901515151515151)
--(axis cs:1.28069924190777,0.772727272727273)
--(axis cs:1.27479150111602,0.643939393939394)
--(axis cs:1.26805833839353,0.515151515151515)
--(axis cs:1.26052195446141,0.386363636363636)
--(axis cs:1.25222194164581,0.257575757575758)
--(axis cs:1.24321371345589,0.128787878787879)
--(axis cs:1.23356656600741,0)
--cycle;

\path [draw=darkslategray61, fill=seagreen5814558, line width=0.5pt]
(axis cs:2.32415857264009,7)
--(axis cs:1.67584142735991,7)
--(axis cs:1.66401044285255,8.25685425685426)
--(axis cs:1.65307716077868,9.51370851370851)
--(axis cs:1.64310185774856,10.7705627705628)
--(axis cs:1.63413049320198,12.027417027417)
--(axis cs:1.62619510446942,13.2842712842713)
--(axis cs:1.61931469725721,14.5411255411255)
--(axis cs:1.61349656507272,15.7979797979798)
--(axis cs:1.60873794878517,17.0548340548341)
--(axis cs:1.60502793041395,18.3116883116883)
--(axis cs:1.60234944444333,19.5685425685426)
--(axis cs:1.60068128616718,20.8253968253968)
--(axis cs:1.6,22.0822510822511)
--(axis cs:1.60028154110538,23.3391053391053)
--(axis cs:1.60150262038865,24.5959595959596)
--(axis cs:1.60364166476684,25.8528138528138)
--(axis cs:1.60667935023108,27.1096681096681)
--(axis cs:1.61059869289287,28.3665223665224)
--(axis cs:1.61538471118791,29.6233766233766)
--(axis cs:1.62102369894512,30.8802308802309)
--(axis cs:1.62750217249335,32.1370851370851)
--(axis cs:1.63480557399161,33.3939393939394)
--(axis cs:1.64291682667416,34.6507936507936)
--(axis cs:1.6518148450283,35.9076479076479)
--(axis cs:1.66147310381568,37.1645021645022)
--(axis cs:1.67185836447247,38.4213564213564)
--(axis cs:1.6829296463338,39.6782106782107)
--(axis cs:1.69463751421536,40.9350649350649)
--(axis cs:1.70692373430504,42.1919191919192)
--(axis cs:1.71972132839964,43.4487734487734)
--(axis cs:1.73295503367496,44.7056277056277)
--(axis cs:1.74654215279675,45.962481962482)
--(axis cs:1.76039375855268,47.2193362193362)
--(axis cs:1.77441619942222,48.4761904761905)
--(axis cs:1.78851283846833,49.7330447330447)
--(axis cs:1.80258594822365,50.989898989899)
--(axis cs:1.81653867914361,52.2467532467532)
--(axis cs:1.83027701870813,53.5036075036075)
--(axis cs:1.84371166209373,54.7604617604618)
--(axis cs:1.85675972299999,56.017316017316)
--(axis cs:1.8693462239993,57.2741702741703)
--(axis cs:1.88140531886672,58.5310245310245)
--(axis cs:1.89288121384848,59.7878787878788)
--(axis cs:1.90372876985037,61.044733044733)
--(axis cs:1.9139137822296,62.3015873015873)
--(axis cs:1.92341294850859,63.5584415584416)
--(axis cs:1.9322135462841,64.8152958152958)
--(axis cs:1.94031285342298,66.0721500721501)
--(axis cs:1.94771735002742,67.3290043290043)
--(axis cs:1.95444174649694,68.5858585858586)
--(axis cs:1.96050788434556,69.8427128427128)
--(axis cs:1.96594355642219,71.0995670995671)
--(axis cs:1.97078129111343,72.3564213564214)
--(axis cs:1.97505714134172,73.6132756132756)
--(axis cs:1.97880951412232,74.8701298701299)
--(axis cs:1.98207807053935,76.1269841269841)
--(axis cs:1.98490271966604,77.3838383838384)
--(axis cs:1.98732272357618,78.6406926406926)
--(axis cs:1.98937592450697,79.8975468975469)
--(axis cs:1.99109809971046,81.1544011544011)
--(axis cs:1.9925224447677,82.4112554112554)
--(axis cs:1.99367918225905,83.6681096681097)
--(axis cs:1.99459528973528,84.9249639249639)
--(axis cs:1.99529433890002,86.1818181818182)
--(axis cs:1.9957964367194,87.4386724386724)
--(axis cs:1.99611825870162,88.6955266955267)
--(axis cs:1.996273164683,89.9523809523809)
--(axis cs:1.99627138794874,91.2092352092352)
--(axis cs:1.99612028922652,92.4660894660895)
--(axis cs:1.99582466784878,93.7229437229437)
--(axis cs:1.99538712303147,94.979797979798)
--(axis cs:1.99480845863462,96.2366522366522)
--(axis cs:1.99408812486244,97.4935064935065)
--(axis cs:1.99322469007417,98.7503607503607)
--(axis cs:1.99221633520276,100.007215007215)
--(axis cs:1.99106136225194,101.264069264069)
--(axis cs:1.98975870703914,102.520923520924)
--(axis cs:1.98830844488813,103.777777777778)
--(axis cs:1.98671227649503,105.034632034632)
--(axis cs:1.98497397986127,106.291486291486)
--(axis cs:1.98309981318311,107.548340548341)
--(axis cs:1.98109885308163,108.805194805195)
--(axis cs:1.97898325270593,110.062049062049)
--(axis cs:1.97676840516947,111.318903318903)
--(axis cs:1.97447299956866,112.575757575758)
--(axis cs:1.97211895950952,113.832611832612)
--(axis cs:1.96973125760204,115.089466089466)
--(axis cs:1.9673376036725,116.34632034632)
--(axis cs:1.964968009332,117.603174603175)
--(axis cs:1.96265423680403,118.860028860029)
--(axis cs:1.96042914528959,120.116883116883)
--(axis cs:1.95832595333738,121.373737373737)
--(axis cs:1.9563774403764,122.630591630592)
--(axis cs:1.95461511445755,123.887445887446)
--(axis cs:1.95306837606415,125.1443001443)
--(axis cs:1.95176370936689,126.401154401154)
--(axis cs:1.95072393236168,127.658008658009)
--(axis cs:1.94996753586829,128.914862914863)
--(axis cs:1.9495081384036,130.171717171717)
--(axis cs:1.94935407958878,131.428571428571)
--(axis cs:2.05064592041122,131.428571428571)
--(axis cs:2.05064592041122,131.428571428571)
--(axis cs:2.0504918615964,130.171717171717)
--(axis cs:2.05003246413171,128.914862914863)
--(axis cs:2.04927606763832,127.658008658009)
--(axis cs:2.04823629063311,126.401154401154)
--(axis cs:2.04693162393585,125.1443001443)
--(axis cs:2.04538488554245,123.887445887446)
--(axis cs:2.0436225596236,122.630591630592)
--(axis cs:2.04167404666262,121.373737373737)
--(axis cs:2.03957085471041,120.116883116883)
--(axis cs:2.03734576319597,118.860028860029)
--(axis cs:2.035031990668,117.603174603175)
--(axis cs:2.0326623963275,116.34632034632)
--(axis cs:2.03026874239796,115.089466089466)
--(axis cs:2.02788104049048,113.832611832612)
--(axis cs:2.02552700043134,112.575757575758)
--(axis cs:2.02323159483053,111.318903318903)
--(axis cs:2.02101674729407,110.062049062049)
--(axis cs:2.01890114691837,108.805194805195)
--(axis cs:2.01690018681689,107.548340548341)
--(axis cs:2.01502602013873,106.291486291486)
--(axis cs:2.01328772350497,105.034632034632)
--(axis cs:2.01169155511187,103.777777777778)
--(axis cs:2.01024129296086,102.520923520924)
--(axis cs:2.00893863774806,101.264069264069)
--(axis cs:2.00778366479724,100.007215007215)
--(axis cs:2.00677530992583,98.7503607503607)
--(axis cs:2.00591187513756,97.4935064935065)
--(axis cs:2.00519154136538,96.2366522366522)
--(axis cs:2.00461287696853,94.979797979798)
--(axis cs:2.00417533215122,93.7229437229437)
--(axis cs:2.00387971077348,92.4660894660895)
--(axis cs:2.00372861205126,91.2092352092352)
--(axis cs:2.003726835317,89.9523809523809)
--(axis cs:2.00388174129838,88.6955266955267)
--(axis cs:2.0042035632806,87.4386724386724)
--(axis cs:2.00470566109998,86.1818181818182)
--(axis cs:2.00540471026472,84.9249639249639)
--(axis cs:2.00632081774095,83.6681096681097)
--(axis cs:2.0074775552323,82.4112554112554)
--(axis cs:2.00890190028954,81.1544011544011)
--(axis cs:2.01062407549303,79.8975468975469)
--(axis cs:2.01267727642382,78.6406926406926)
--(axis cs:2.01509728033396,77.3838383838384)
--(axis cs:2.01792192946065,76.1269841269841)
--(axis cs:2.02119048587768,74.8701298701299)
--(axis cs:2.02494285865828,73.6132756132756)
--(axis cs:2.02921870888657,72.3564213564214)
--(axis cs:2.03405644357781,71.0995670995671)
--(axis cs:2.03949211565444,69.8427128427128)
--(axis cs:2.04555825350306,68.5858585858586)
--(axis cs:2.05228264997258,67.3290043290043)
--(axis cs:2.05968714657702,66.0721500721501)
--(axis cs:2.0677864537159,64.8152958152958)
--(axis cs:2.07658705149141,63.5584415584416)
--(axis cs:2.0860862177704,62.3015873015873)
--(axis cs:2.09627123014963,61.044733044733)
--(axis cs:2.10711878615152,59.7878787878788)
--(axis cs:2.11859468113328,58.5310245310245)
--(axis cs:2.1306537760007,57.2741702741703)
--(axis cs:2.14324027700001,56.017316017316)
--(axis cs:2.15628833790627,54.7604617604618)
--(axis cs:2.16972298129187,53.5036075036075)
--(axis cs:2.18346132085639,52.2467532467532)
--(axis cs:2.19741405177635,50.989898989899)
--(axis cs:2.21148716153167,49.7330447330447)
--(axis cs:2.22558380057778,48.4761904761905)
--(axis cs:2.23960624144732,47.2193362193362)
--(axis cs:2.25345784720325,45.962481962482)
--(axis cs:2.26704496632504,44.7056277056277)
--(axis cs:2.28027867160036,43.4487734487734)
--(axis cs:2.29307626569496,42.1919191919192)
--(axis cs:2.30536248578464,40.9350649350649)
--(axis cs:2.3170703536662,39.6782106782107)
--(axis cs:2.32814163552753,38.4213564213564)
--(axis cs:2.33852689618432,37.1645021645022)
--(axis cs:2.3481851549717,35.9076479076479)
--(axis cs:2.35708317332584,34.6507936507936)
--(axis cs:2.36519442600839,33.3939393939394)
--(axis cs:2.37249782750665,32.1370851370851)
--(axis cs:2.37897630105488,30.8802308802309)
--(axis cs:2.38461528881209,29.6233766233766)
--(axis cs:2.38940130710713,28.3665223665224)
--(axis cs:2.39332064976892,27.1096681096681)
--(axis cs:2.39635833523316,25.8528138528138)
--(axis cs:2.39849737961135,24.5959595959596)
--(axis cs:2.39971845889462,23.3391053391053)
--(axis cs:2.4,22.0822510822511)
--(axis cs:2.39931871383282,20.8253968253968)
--(axis cs:2.39765055555667,19.5685425685426)
--(axis cs:2.39497206958605,18.3116883116883)
--(axis cs:2.39126205121483,17.0548340548341)
--(axis cs:2.38650343492728,15.7979797979798)
--(axis cs:2.38068530274279,14.5411255411255)
--(axis cs:2.37380489553058,13.2842712842713)
--(axis cs:2.36586950679802,12.027417027417)
--(axis cs:2.35689814225144,10.7705627705628)
--(axis cs:2.34692283922132,9.51370851370851)
--(axis cs:2.33598955714745,8.25685425685426)
--(axis cs:2.32415857264009,7)
--cycle;

\addplot [line width=0.75pt, darkslategray61]
table {%
0 3.33333333333333
0 49.5
};
\addplot [line width=2.25pt, darkslategray61]
table {%
0 5.53571428571429
0 24.3253968253968
};
\addplot [semithick, darkslategray61, mark=-, mark size=2.34375, mark options={solid,fill=white}]
table {%
0 16.5384615384615
};
\addplot [line width=0.75pt, darkslategray61]
table {%
1 0
1 12.75
};
\addplot [line width=2.25pt, darkslategray61]
table {%
1 1.45833333333333
1 6.18465909090909
};
\addplot [semithick, darkslategray61, mark=-, mark size=2.34375, mark options={solid,fill=white}]
table {%
1 5.08333333333333
};
\addplot [line width=0.75pt, darkslategray61]
table {%
2 7
2 44.7777777777778
};
\addplot [line width=2.25pt, darkslategray61]
table {%
2 11.6666666666667
2 38.1666666666667
};
\addplot [semithick, darkslategray61, mark=-, mark size=2.34375, mark options={solid,fill=white}]
table {%
2 26.8888888888889
};
\end{axis}

\end{tikzpicture}

%% file: figures/execution_time.tex
\begin{tikzpicture}

\definecolor{darkgray176}{RGB}{176,176,176}
\definecolor{darkslategray61}{RGB}{61,61,61}
\definecolor{peru22412844}{RGB}{224,128,44}
\definecolor{seagreen5814558}{RGB}{58,145,58}
\definecolor{steelblue49115161}{RGB}{49,115,161}

\begin{axis}[
log basis y={10},
tick align=outside,
tick pos=left,
x grid style={darkgray176},
xtick style={color=black},
xtick={0,1,2},
xticklabels={MTR,vMTR\\(real),vMTR \\(virtual)},
y grid style={darkgray176},
ymax = 750, ymin=0,
xmin = -.5, xmax=2.5,
ylabel={Time (s)},
xticklabel style={align=center}
]
\path [draw=darkslategray61, fill=steelblue49115161, line width=0.5pt]
(axis cs:0.352095009610348,1.683971518)
--(axis cs:-0.352095009610348,1.683971518)
--(axis cs:-0.362950046414645,9.14745839658586)
--(axis cs:-0.372492361024892,16.6109452751717)
--(axis cs:-0.380654163196544,24.0744321537576)
--(axis cs:-0.387390970640688,31.5379190323434)
--(axis cs:-0.392681738833578,39.0014059109293)
--(axis cs:-0.396528204232154,46.4648927895151)
--(axis cs:-0.398953482645656,53.928379668101)
--(axis cs:-0.4,61.3918665466869)
--(axis cs:-0.399726864167851,68.8553534252727)
--(axis cs:-0.398206812051047,76.3188403038586)
--(axis cs:-0.395522884177128,83.7823271824444)
--(axis cs:-0.39176498865482,91.2458140610303)
--(axis cs:-0.387026516908452,98.7093009396162)
--(axis cs:-0.381401165235697,106.172787818202)
--(axis cs:-0.374980099532324,113.636274696788)
--(axis cs:-0.367849576649414,121.099761575374)
--(axis cs:-0.36008910637912,128.56324845396)
--(axis cs:-0.351770204915698,136.026735332545)
--(axis cs:-0.342955755909529,143.490222211131)
--(axis cs:-0.333699961066329,150.953709089717)
--(axis cs:-0.324048830681548,158.417195968303)
--(axis cs:-0.314041137342738,165.880682846889)
--(axis cs:-0.303709734730854,173.344169725475)
--(axis cs:-0.293083129019504,180.807656604061)
--(axis cs:-0.282187183337959,188.271143482646)
--(axis cs:-0.271046836159126,195.734630361232)
--(axis cs:-0.259687721851387,203.198117239818)
--(axis cs:-0.248137595126404,210.661604118404)
--(axis cs:-0.236427479518327,218.12509099699)
--(axis cs:-0.224592481900944,225.588577875576)
--(axis cs:-0.212672238820734,233.052064754162)
--(axis cs:-0.200710984515963,240.515551632747)
--(axis cs:-0.188757253419969,247.979038511333)
--(axis cs:-0.176863250411953,255.442525389919)
--(axis cs:-0.165083939040474,262.906012268505)
--(axis cs:-0.153475910669403,270.369499147091)
--(axis cs:-0.142096105578164,277.832986025677)
--(axis cs:-0.131000460409413,285.296472904263)
--(axis cs:-0.120242555221666,292.759959782848)
--(axis cs:-0.109872328254531,300.223446661434)
--(axis cs:-0.0999349180345526,307.68693354002)
--(axis cs:-0.0904696814614876,315.150420418606)
--(axis cs:-0.0815094239081942,322.613907297192)
--(axis cs:-0.0730798640359698,330.077394175778)
--(axis cs:-0.065199342810189,337.540881054364)
--(axis cs:-0.0578787738350143,345.004367932949)
--(axis cs:-0.0511218212093099,352.467854811535)
--(axis cs:-0.0449252820774817,359.931341690121)
--(axis cs:-0.0392796441795561,367.394828568707)
--(axis cs:-0.0341697841011026,374.858315447293)
--(axis cs:-0.0295757695428206,382.321802325879)
--(axis cs:-0.0254737286020637,389.785289204465)
--(axis cs:-0.0218367505157084,397.24877608305)
--(axis cs:-0.0186357852189302,404.712262961636)
--(axis cs:-0.0158405130539036,412.175749840222)
--(axis cs:-0.0134201606344995,419.639236718808)
--(axis cs:-0.0113442438736042,427.102723597394)
--(axis cs:-0.00958322418294901,434.56621047598)
--(axis cs:-0.00810906858915671,442.029697354566)
--(axis cs:-0.00689570876524654,449.493184233152)
--(axis cs:-0.00591939761274801,456.956671111737)
--(axis cs:-0.00515896497136806,464.420157990323)
--(axis cs:-0.00459597626831723,471.883644868909)
--(axis cs:-0.0042147994892386,479.347131747495)
--(axis cs:-0.00400258684181919,486.810618626081)
--(axis cs:-0.00394917800753105,494.274105504667)
--(axis cs:-0.00404693207116717,501.737592383252)
--(axis cs:-0.00429049522268665,509.201079261838)
--(axis cs:-0.00467651127652277,516.664566140424)
--(axis cs:-0.00520328206895784,524.12805301901)
--(axis cs:-0.00587038496849805,531.591539897596)
--(axis cs:-0.00667825512972984,539.055026776182)
--(axis cs:-0.00762774076404683,546.518513654768)
--(axis cs:-0.0087196405788595,553.982000533353)
--(axis cs:-0.0099542336008132,561.445487411939)
--(axis cs:-0.0113308127641491,568.908974290525)
--(axis cs:-0.0128472347999847,576.372461169111)
--(axis cs:-0.0144994999726237,583.835948047697)
--(axis cs:-0.0162813759309244,591.299434926283)
--(axis cs:-0.0181840802326693,598.762921804869)
--(axis cs:-0.0201960358268233,606.226408683455)
--(axis cs:-0.0223027128361815,613.68989556204)
--(axis cs:-0.0244865683006584,621.153382440626)
--(axis cs:-0.0267270930939908,628.616869319212)
--(axis cs:-0.029000972040531,636.080356197798)
--(axis cs:-0.031282359416062,643.543843076384)
--(axis cs:-0.0335432676540328,651.00732995497)
--(axis cs:-0.0357540623833094,658.470816833556)
--(axis cs:-0.0378840521237912,665.934303712141)
--(axis cs:-0.0399021563189971,673.397790590727)
--(axis cs:-0.0417776311597437,680.861277469313)
--(axis cs:-0.0434808291150685,688.324764347899)
--(axis cs:-0.0449839654766259,695.788251226485)
--(axis cs:-0.0462618637402239,703.251738105071)
--(axis cs:-0.0472926514345433,710.715224983657)
--(axis cs:-0.0480583791337156,718.178711862242)
--(axis cs:-0.0485455378504294,725.642198740828)
--(axis cs:-0.0487454537118712,733.105685619414)
--(axis cs:-0.0486545436067839,740.569172498)
--(axis cs:0.0486545436067839,740.569172498)
--(axis cs:0.0486545436067839,740.569172498)
--(axis cs:0.0487454537118712,733.105685619414)
--(axis cs:0.0485455378504294,725.642198740828)
--(axis cs:0.0480583791337156,718.178711862242)
--(axis cs:0.0472926514345433,710.715224983657)
--(axis cs:0.0462618637402239,703.251738105071)
--(axis cs:0.0449839654766259,695.788251226485)
--(axis cs:0.0434808291150685,688.324764347899)
--(axis cs:0.0417776311597437,680.861277469313)
--(axis cs:0.0399021563189971,673.397790590727)
--(axis cs:0.0378840521237912,665.934303712141)
--(axis cs:0.0357540623833094,658.470816833556)
--(axis cs:0.0335432676540328,651.00732995497)
--(axis cs:0.031282359416062,643.543843076384)
--(axis cs:0.029000972040531,636.080356197798)
--(axis cs:0.0267270930939908,628.616869319212)
--(axis cs:0.0244865683006584,621.153382440626)
--(axis cs:0.0223027128361815,613.68989556204)
--(axis cs:0.0201960358268233,606.226408683455)
--(axis cs:0.0181840802326693,598.762921804869)
--(axis cs:0.0162813759309244,591.299434926283)
--(axis cs:0.0144994999726237,583.835948047697)
--(axis cs:0.0128472347999847,576.372461169111)
--(axis cs:0.0113308127641491,568.908974290525)
--(axis cs:0.0099542336008132,561.445487411939)
--(axis cs:0.0087196405788595,553.982000533353)
--(axis cs:0.00762774076404683,546.518513654768)
--(axis cs:0.00667825512972984,539.055026776182)
--(axis cs:0.00587038496849805,531.591539897596)
--(axis cs:0.00520328206895784,524.12805301901)
--(axis cs:0.00467651127652277,516.664566140424)
--(axis cs:0.00429049522268665,509.201079261838)
--(axis cs:0.00404693207116717,501.737592383252)
--(axis cs:0.00394917800753105,494.274105504667)
--(axis cs:0.00400258684181919,486.810618626081)
--(axis cs:0.0042147994892386,479.347131747495)
--(axis cs:0.00459597626831723,471.883644868909)
--(axis cs:0.00515896497136806,464.420157990323)
--(axis cs:0.00591939761274801,456.956671111737)
--(axis cs:0.00689570876524654,449.493184233152)
--(axis cs:0.00810906858915671,442.029697354566)
--(axis cs:0.00958322418294901,434.56621047598)
--(axis cs:0.0113442438736042,427.102723597394)
--(axis cs:0.0134201606344995,419.639236718808)
--(axis cs:0.0158405130539036,412.175749840222)
--(axis cs:0.0186357852189302,404.712262961636)
--(axis cs:0.0218367505157084,397.24877608305)
--(axis cs:0.0254737286020637,389.785289204465)
--(axis cs:0.0295757695428206,382.321802325879)
--(axis cs:0.0341697841011026,374.858315447293)
--(axis cs:0.0392796441795561,367.394828568707)
--(axis cs:0.0449252820774817,359.931341690121)
--(axis cs:0.0511218212093099,352.467854811535)
--(axis cs:0.0578787738350143,345.004367932949)
--(axis cs:0.065199342810189,337.540881054364)
--(axis cs:0.0730798640359698,330.077394175778)
--(axis cs:0.0815094239081942,322.613907297192)
--(axis cs:0.0904696814614876,315.150420418606)
--(axis cs:0.0999349180345526,307.68693354002)
--(axis cs:0.109872328254531,300.223446661434)
--(axis cs:0.120242555221666,292.759959782848)
--(axis cs:0.131000460409413,285.296472904263)
--(axis cs:0.142096105578164,277.832986025677)
--(axis cs:0.153475910669403,270.369499147091)
--(axis cs:0.165083939040474,262.906012268505)
--(axis cs:0.176863250411953,255.442525389919)
--(axis cs:0.188757253419969,247.979038511333)
--(axis cs:0.200710984515963,240.515551632747)
--(axis cs:0.212672238820734,233.052064754162)
--(axis cs:0.224592481900944,225.588577875576)
--(axis cs:0.236427479518327,218.12509099699)
--(axis cs:0.248137595126404,210.661604118404)
--(axis cs:0.259687721851387,203.198117239818)
--(axis cs:0.271046836159126,195.734630361232)
--(axis cs:0.282187183337959,188.271143482646)
--(axis cs:0.293083129019504,180.807656604061)
--(axis cs:0.303709734730854,173.344169725475)
--(axis cs:0.314041137342738,165.880682846889)
--(axis cs:0.324048830681548,158.417195968303)
--(axis cs:0.333699961066329,150.953709089717)
--(axis cs:0.342955755909529,143.490222211131)
--(axis cs:0.351770204915698,136.026735332545)
--(axis cs:0.36008910637912,128.56324845396)
--(axis cs:0.367849576649414,121.099761575374)
--(axis cs:0.374980099532324,113.636274696788)
--(axis cs:0.381401165235697,106.172787818202)
--(axis cs:0.387026516908452,98.7093009396162)
--(axis cs:0.39176498865482,91.2458140610303)
--(axis cs:0.395522884177128,83.7823271824444)
--(axis cs:0.398206812051047,76.3188403038586)
--(axis cs:0.399726864167851,68.8553534252727)
--(axis cs:0.4,61.3918665466869)
--(axis cs:0.398953482645656,53.928379668101)
--(axis cs:0.396528204232154,46.4648927895151)
--(axis cs:0.392681738833578,39.0014059109293)
--(axis cs:0.387390970640688,31.5379190323434)
--(axis cs:0.380654163196544,24.0744321537576)
--(axis cs:0.372492361024892,16.6109452751717)
--(axis cs:0.362950046414645,9.14745839658586)
--(axis cs:0.352095009610348,1.683971518)
--cycle;

\path [draw=darkslategray61, fill=peru22412844, line width=0.5pt]
(axis cs:1.36395990702431,0.053029134)
--(axis cs:0.636040092975688,0.053029134)
--(axis cs:0.626251908389558,4.20512398706061)
--(axis cs:0.617958466379592,8.35721884012121)
--(axis cs:0.611219585959163,12.5093136931818)
--(axis cs:0.60606647215902,16.6614085462424)
--(axis cs:0.602501189512009,20.813503399303)
--(axis cs:0.60049706307511,24.9655982523636)
--(axis cs:0.6,29.1176931054242)
--(axis cs:0.600930685656212,33.2697879584848)
--(axis cs:0.603187570968918,37.4218828115455)
--(axis cs:0.606650533893833,41.5739776646061)
--(axis cs:0.611185069567402,45.7260725176667)
--(axis cs:0.616646842102591,49.8781673707273)
--(axis cs:0.622886417331662,54.0302622237879)
--(axis cs:0.629753990667545,58.1823570768485)
--(axis cs:0.637103927829704,62.3344519299091)
--(axis cs:0.644798948144378,66.4865467829697)
--(axis cs:0.652713799717618,70.6386416360303)
--(axis cs:0.660738301831058,74.7907364890909)
--(axis cs:0.668779660946416,78.9428313421515)
--(axis cs:0.676764001025219,83.0949261952121)
--(axis cs:0.684637084659757,87.2470210482727)
--(axis cs:0.692364236947805,91.3991159013333)
--(axis cs:0.69992951740218,95.5512107543939)
--(axis cs:0.707334214932114,99.7033056074545)
--(axis cs:0.714594765798137,103.855400460515)
--(axis cs:0.721740213479141,108.007495313576)
--(axis cs:0.728809342008023,112.159590166636)
--(axis cs:0.735847620304031,116.311685019697)
--(axis cs:0.742904094481832,120.463779872758)
--(axis cs:0.750028358496359,124.615874725818)
--(axis cs:0.757267721509437,128.767969578879)
--(axis cs:0.764664673973844,132.920064431939)
--(axis cs:0.772254734705823,137.072159285)
--(axis cs:0.780064739320612,141.224254138061)
--(axis cs:0.788111607512336,145.376348991121)
--(axis cs:0.796401603896373,149.528443844182)
--(axis cs:0.804930085524258,153.680538697242)
--(axis cs:0.813681709610378,157.832633550303)
--(axis cs:0.822631058183777,161.984728403364)
--(axis cs:0.831743622812906,166.136823256424)
--(axis cs:0.840977082562109,170.288918109485)
--(axis cs:0.850282802047357,174.441012962545)
--(axis cs:0.859607473804433,178.593107815606)
--(axis cs:0.86889482994412,182.745202668667)
--(axis cs:0.878087351890604,186.897297521727)
--(axis cs:0.887127913422067,191.049392374788)
--(axis cs:0.895961300725926,195.201487227849)
--(axis cs:0.904535563175643,199.353582080909)
--(axis cs:0.912803159453464,203.50567693397)
--(axis cs:0.920721874925357,207.65777178703)
--(axis cs:0.928255497305873,211.809866640091)
--(axis cs:0.935374248179783,215.961961493152)
--(axis cs:0.942054977499264,220.114056346212)
--(axis cs:0.948281136461696,224.266151199273)
--(axis cs:0.954042550994821,228.418246052333)
--(axis cs:0.959335023322773,232.570340905394)
--(axis cs:0.964159792730989,236.722435758455)
--(axis cs:0.968522888736116,240.874530611515)
--(axis cs:0.972434410506533,245.026625464576)
--(axis cs:0.975907765724914,249.178720317636)
--(axis cs:0.978958900324491,253.330815170697)
--(axis cs:0.981605547871286,257.482910023758)
--(axis cs:0.983866524016365,261.635004876818)
--(axis cs:0.98576108760846,265.787099729879)
--(axis cs:0.987308385924827,269.939194582939)
--(axis cs:0.988526997209987,274.091289436)
--(axis cs:0.989434579443982,278.243384289061)
--(axis cs:0.990047630101395,282.395479142121)
--(axis cs:0.990381357690536,286.547573995182)
--(axis cs:0.990449662136878,290.699668848242)
--(axis cs:0.990265217636435,294.851763701303)
--(axis cs:0.989839648482318,299.003858554364)
--(axis cs:0.98918378558516,303.155953407424)
--(axis cs:0.988307988989781,307.308048260485)
--(axis cs:0.987222519665851,311.460143113545)
--(axis cs:0.985937942256239,315.612237966606)
--(axis cs:0.984465539348182,319.764332819667)
--(axis cs:0.982817717240585,323.916427672727)
--(axis cs:0.981008383169238,328.068522525788)
--(axis cs:0.97905327457074,332.220617378848)
--(axis cs:0.9769702222551,336.372712231909)
--(axis cs:0.974779331337919,340.52480708497)
--(axis cs:0.972503066451616,344.67690193803)
--(axis cs:0.9701662310749,348.828996791091)
--(axis cs:0.967795834717442,352.981091644151)
--(axis cs:0.965420846060145,357.133186497212)
--(axis cs:0.963071834830116,361.285281350273)
--(axis cs:0.960780509999053,365.437376203333)
--(axis cs:0.958579166623556,369.589471056394)
--(axis cs:0.956500058069172,373.741565909455)
--(axis cs:0.954574714247165,377.893660762515)
--(axis cs:0.952833229625867,382.045755615576)
--(axis cs:0.951303546965095,386.197850468636)
--(axis cs:0.950010763810991,390.349945321697)
--(axis cs:0.948976488681014,394.502040174758)
--(axis cs:0.948218272527817,398.654135027818)
--(axis cs:0.947749138526946,402.806229880879)
--(axis cs:0.947577229585566,406.958324733939)
--(axis cs:0.947705588381116,411.110419587)
--(axis cs:1.05229441161888,411.110419587)
--(axis cs:1.05229441161888,411.110419587)
--(axis cs:1.05242277041443,406.958324733939)
--(axis cs:1.05225086147305,402.806229880879)
--(axis cs:1.05178172747218,398.654135027818)
--(axis cs:1.05102351131899,394.502040174758)
--(axis cs:1.04998923618901,390.349945321697)
--(axis cs:1.0486964530349,386.197850468636)
--(axis cs:1.04716677037413,382.045755615576)
--(axis cs:1.04542528575283,377.893660762515)
--(axis cs:1.04349994193083,373.741565909455)
--(axis cs:1.04142083337644,369.589471056394)
--(axis cs:1.03921949000095,365.437376203333)
--(axis cs:1.03692816516988,361.285281350273)
--(axis cs:1.03457915393986,357.133186497212)
--(axis cs:1.03220416528256,352.981091644151)
--(axis cs:1.0298337689251,348.828996791091)
--(axis cs:1.02749693354838,344.67690193803)
--(axis cs:1.02522066866208,340.52480708497)
--(axis cs:1.0230297777449,336.372712231909)
--(axis cs:1.02094672542926,332.220617378848)
--(axis cs:1.01899161683076,328.068522525788)
--(axis cs:1.01718228275942,323.916427672727)
--(axis cs:1.01553446065182,319.764332819667)
--(axis cs:1.01406205774376,315.612237966606)
--(axis cs:1.01277748033415,311.460143113545)
--(axis cs:1.01169201101022,307.308048260485)
--(axis cs:1.01081621441484,303.155953407424)
--(axis cs:1.01016035151768,299.003858554364)
--(axis cs:1.00973478236357,294.851763701303)
--(axis cs:1.00955033786312,290.699668848242)
--(axis cs:1.00961864230946,286.547573995182)
--(axis cs:1.00995236989861,282.395479142121)
--(axis cs:1.01056542055602,278.243384289061)
--(axis cs:1.01147300279001,274.091289436)
--(axis cs:1.01269161407517,269.939194582939)
--(axis cs:1.01423891239154,265.787099729879)
--(axis cs:1.01613347598363,261.635004876818)
--(axis cs:1.01839445212871,257.482910023758)
--(axis cs:1.02104109967551,253.330815170697)
--(axis cs:1.02409223427509,249.178720317636)
--(axis cs:1.02756558949347,245.026625464576)
--(axis cs:1.03147711126388,240.874530611515)
--(axis cs:1.03584020726901,236.722435758455)
--(axis cs:1.04066497667723,232.570340905394)
--(axis cs:1.04595744900518,228.418246052333)
--(axis cs:1.0517188635383,224.266151199273)
--(axis cs:1.05794502250074,220.114056346212)
--(axis cs:1.06462575182022,215.961961493152)
--(axis cs:1.07174450269413,211.809866640091)
--(axis cs:1.07927812507464,207.65777178703)
--(axis cs:1.08719684054654,203.50567693397)
--(axis cs:1.09546443682436,199.353582080909)
--(axis cs:1.10403869927407,195.201487227849)
--(axis cs:1.11287208657793,191.049392374788)
--(axis cs:1.1219126481094,186.897297521727)
--(axis cs:1.13110517005588,182.745202668667)
--(axis cs:1.14039252619557,178.593107815606)
--(axis cs:1.14971719795264,174.441012962545)
--(axis cs:1.15902291743789,170.288918109485)
--(axis cs:1.16825637718709,166.136823256424)
--(axis cs:1.17736894181622,161.984728403364)
--(axis cs:1.18631829038962,157.832633550303)
--(axis cs:1.19506991447574,153.680538697242)
--(axis cs:1.20359839610363,149.528443844182)
--(axis cs:1.21188839248766,145.376348991121)
--(axis cs:1.21993526067939,141.224254138061)
--(axis cs:1.22774526529418,137.072159285)
--(axis cs:1.23533532602616,132.920064431939)
--(axis cs:1.24273227849056,128.767969578879)
--(axis cs:1.24997164150364,124.615874725818)
--(axis cs:1.25709590551817,120.463779872758)
--(axis cs:1.26415237969597,116.311685019697)
--(axis cs:1.27119065799198,112.159590166636)
--(axis cs:1.27825978652086,108.007495313576)
--(axis cs:1.28540523420186,103.855400460515)
--(axis cs:1.29266578506789,99.7033056074545)
--(axis cs:1.30007048259782,95.5512107543939)
--(axis cs:1.30763576305219,91.3991159013333)
--(axis cs:1.31536291534024,87.2470210482727)
--(axis cs:1.32323599897478,83.0949261952121)
--(axis cs:1.33122033905358,78.9428313421515)
--(axis cs:1.33926169816894,74.7907364890909)
--(axis cs:1.34728620028238,70.6386416360303)
--(axis cs:1.35520105185562,66.4865467829697)
--(axis cs:1.3628960721703,62.3344519299091)
--(axis cs:1.37024600933245,58.1823570768485)
--(axis cs:1.37711358266834,54.0302622237879)
--(axis cs:1.38335315789741,49.8781673707273)
--(axis cs:1.3888149304326,45.7260725176667)
--(axis cs:1.39334946610617,41.5739776646061)
--(axis cs:1.39681242903108,37.4218828115455)
--(axis cs:1.39906931434379,33.2697879584848)
--(axis cs:1.4,29.1176931054242)
--(axis cs:1.39950293692489,24.9655982523636)
--(axis cs:1.39749881048799,20.813503399303)
--(axis cs:1.39393352784098,16.6614085462424)
--(axis cs:1.38878041404084,12.5093136931818)
--(axis cs:1.38204153362041,8.35721884012121)
--(axis cs:1.37374809161044,4.20512398706061)
--(axis cs:1.36395990702431,0.053029134)
--cycle;

\path [draw=darkslategray61, fill=seagreen5814558, line width=0.5pt]
(axis cs:2.39301082965206,0.0146852)
--(axis cs:1.60698917034794,0.0146852)
--(axis cs:1.60270000993394,0.315732824242424)
--(axis cs:1.60036181611351,0.616780448484848)
--(axis cs:1.6,0.917828072727273)
--(axis cs:1.60160912871309,1.2188756969697)
--(axis cs:1.60515303655289,1.51992332121212)
--(axis cs:1.61056571750381,1.82097094545455)
--(axis cs:1.61775296907183,2.12201856969697)
--(axis cs:1.6265947317682,2.42306619393939)
--(axis cs:1.63694804537149,2.72411381818182)
--(axis cs:1.64865052339306,3.02516144242424)
--(axis cs:1.66152423171127,3.32620906666667)
--(axis cs:1.67537984667271,3.62725669090909)
--(axis cs:1.69002096238032,3.92830431515152)
--(axis cs:1.70524841644112,4.22935193939394)
--(axis cs:1.72086450792814,4.53039956363636)
--(axis cs:1.73667699029825,4.83144718787879)
--(axis cs:1.75250273489131,5.13249481212121)
--(axis cs:1.76817097666233,5.43354243636364)
--(axis cs:1.78352607212147,5.73459006060606)
--(axis cs:1.79842971918387,6.03563768484848)
--(axis cs:1.81276260887647,6.33668530909091)
--(axis cs:1.82642549877206,6.63773293333333)
--(axis cs:1.83933971686921,6.93878055757576)
--(axis cs:1.85144712177013,7.23982818181818)
--(axis cs:1.8627095599221,7.54087580606061)
--(axis cs:1.87310787301705,7.84192343030303)
--(axis cs:1.88264051817118,8.14297105454546)
--(axis cs:1.89132187014818,8.44401867878788)
--(axis cs:1.89918027868833,8.7450663030303)
--(axis cs:1.90625595510733,9.04611392727273)
--(axis cs:1.9125987609663,9.34716155151515)
--(axis cs:1.91826596808295,9.64820917575758)
--(axis cs:1.92332005379076,9.9492568)
--(axis cs:1.9278265885206,10.2503044242424)
--(axis cs:1.93185226484222,10.5513520484848)
--(axis cs:1.935463108422,10.8523996727273)
--(axis cs:1.93872290226771,11.1534472969697)
--(axis cs:1.9416918464543,11.4544949212121)
--(axis cs:1.94442546653864,11.7555425454545)
--(axis cs:1.94697377532268,12.056590169697)
--(axis cs:1.94938068472263,12.3576377939394)
--(axis cs:1.95168365742022,12.6586854181818)
--(axis cs:1.95391358184601,12.9597330424242)
--(axis cs:1.95609484897276,13.2607806666667)
--(axis cs:1.95824560544372,13.5618282909091)
--(axis cs:1.96037815475282,13.8628759151515)
--(axis cs:1.96249947652784,14.1639235393939)
--(axis cs:1.96461183340495,14.4649711636364)
--(axis cs:1.9667134354568,14.7660187878788)
--(axis cs:1.9687991335516,15.0670664121212)
--(axis cs:1.97086111525977,15.3681140363636)
--(axis cs:1.97288957984935,15.6691616606061)
--(axis cs:1.97487337237006,15.9702092848485)
--(axis cs:1.97680056065878,16.2712569090909)
--(axis cs:1.97865894314199,16.5723045333333)
--(axis cs:1.98043647940524,16.8733521575758)
--(axis cs:1.98212163949347,17.1743997818182)
--(axis cs:1.98370367166423,17.4754474060606)
--(axis cs:1.98517279171814,17.776495030303)
--(axis cs:1.98652029998245,18.0775426545455)
--(axis cs:1.98773863444963,18.3785902787879)
--(axis cs:1.9888213704254,18.6796379030303)
--(axis cs:1.98976317829485,18.9806855272727)
--(axis cs:1.99055975166897,19.2817331515152)
--(axis cs:1.99120771824675,19.5827807757576)
--(axis cs:1.99170454525612,19.8838284)
--(axis cs:1.99204845037423,20.1848760242424)
--(axis cs:1.99223832763683,20.4859236484848)
--(axis cs:1.99227369610381,20.7869712727273)
--(axis cs:1.99215467703273,21.0880188969697)
--(axis cs:1.99188200311032,21.3890665212121)
--(axis cs:1.99145706099168,21.6901141454545)
--(axis cs:1.99088196608568,21.991161769697)
--(axis cs:1.99015966629161,22.2922093939394)
--(axis cs:1.98929406932082,22.5932570181818)
--(axis cs:1.98829018640865,22.8943046424242)
--(axis cs:1.98715428371215,23.1953522666667)
--(axis cs:1.98589403156308,23.4963998909091)
--(axis cs:1.98451864106038,23.7974475151515)
--(axis cs:1.98303897728329,24.0984951393939)
--(axis cs:1.9814676387123,24.3995427636364)
--(axis cs:1.97981899326726,24.7005903878788)
--(axis cs:1.97810916269712,25.0016380121212)
--(axis cs:1.97635594884839,25.3026856363636)
--(axis cs:1.9745786975398,25.6037332606061)
--(axis cs:1.97279809829937,25.9047808848485)
--(axis cs:1.97103592097404,26.2058285090909)
--(axis cs:1.96931469308416,26.5068761333333)
--(axis cs:1.96765732463401,26.8079237575758)
--(axis cs:1.96608668976916,27.1089713818182)
--(axis cs:1.96462517705549,27.4100190060606)
--(axis cs:1.96329422211503,27.711066630303)
--(axis cs:1.9621138377801,28.0121142545455)
--(axis cs:1.96110215772805,28.3131618787879)
--(axis cs:1.96027500967713,28.6142095030303)
--(axis cs:1.95964553363051,28.9152571272727)
--(axis cs:1.95922385936301,29.2163047515152)
--(axis cs:1.95901685539866,29.5173523757576)
--(axis cs:1.95902795921049,29.8184)
--(axis cs:2.04097204078951,29.8184)
--(axis cs:2.04097204078951,29.8184)
--(axis cs:2.04098314460134,29.5173523757576)
--(axis cs:2.04077614063699,29.2163047515152)
--(axis cs:2.04035446636949,28.9152571272727)
--(axis cs:2.03972499032287,28.6142095030303)
--(axis cs:2.03889784227195,28.3131618787879)
--(axis cs:2.0378861622199,28.0121142545455)
--(axis cs:2.03670577788497,27.711066630303)
--(axis cs:2.03537482294451,27.4100190060606)
--(axis cs:2.03391331023084,27.1089713818182)
--(axis cs:2.03234267536599,26.8079237575758)
--(axis cs:2.03068530691584,26.5068761333333)
--(axis cs:2.02896407902596,26.2058285090909)
--(axis cs:2.02720190170063,25.9047808848485)
--(axis cs:2.0254213024602,25.6037332606061)
--(axis cs:2.02364405115161,25.3026856363636)
--(axis cs:2.02189083730288,25.0016380121212)
--(axis cs:2.02018100673274,24.7005903878788)
--(axis cs:2.01853236128769,24.3995427636364)
--(axis cs:2.01696102271671,24.0984951393939)
--(axis cs:2.01548135893962,23.7974475151515)
--(axis cs:2.01410596843692,23.4963998909091)
--(axis cs:2.01284571628785,23.1953522666667)
--(axis cs:2.01170981359135,22.8943046424242)
--(axis cs:2.01070593067918,22.5932570181818)
--(axis cs:2.00984033370839,22.2922093939394)
--(axis cs:2.00911803391432,21.991161769697)
--(axis cs:2.00854293900832,21.6901141454545)
--(axis cs:2.00811799688968,21.3890665212121)
--(axis cs:2.00784532296727,21.0880188969697)
--(axis cs:2.00772630389619,20.7869712727273)
--(axis cs:2.00776167236317,20.4859236484848)
--(axis cs:2.00795154962577,20.1848760242424)
--(axis cs:2.00829545474388,19.8838284)
--(axis cs:2.00879228175325,19.5827807757576)
--(axis cs:2.00944024833103,19.2817331515152)
--(axis cs:2.01023682170515,18.9806855272727)
--(axis cs:2.0111786295746,18.6796379030303)
--(axis cs:2.01226136555037,18.3785902787879)
--(axis cs:2.01347970001755,18.0775426545455)
--(axis cs:2.01482720828186,17.776495030303)
--(axis cs:2.01629632833577,17.4754474060606)
--(axis cs:2.01787836050653,17.1743997818182)
--(axis cs:2.01956352059476,16.8733521575758)
--(axis cs:2.02134105685801,16.5723045333333)
--(axis cs:2.02319943934122,16.2712569090909)
--(axis cs:2.02512662762994,15.9702092848485)
--(axis cs:2.02711042015065,15.6691616606061)
--(axis cs:2.02913888474023,15.3681140363636)
--(axis cs:2.0312008664484,15.0670664121212)
--(axis cs:2.0332865645432,14.7660187878788)
--(axis cs:2.03538816659505,14.4649711636364)
--(axis cs:2.03750052347216,14.1639235393939)
--(axis cs:2.03962184524718,13.8628759151515)
--(axis cs:2.04175439455628,13.5618282909091)
--(axis cs:2.04390515102724,13.2607806666667)
--(axis cs:2.04608641815399,12.9597330424242)
--(axis cs:2.04831634257978,12.6586854181818)
--(axis cs:2.05061931527737,12.3576377939394)
--(axis cs:2.05302622467732,12.056590169697)
--(axis cs:2.05557453346136,11.7555425454545)
--(axis cs:2.0583081535457,11.4544949212121)
--(axis cs:2.06127709773229,11.1534472969697)
--(axis cs:2.064536891578,10.8523996727273)
--(axis cs:2.06814773515778,10.5513520484848)
--(axis cs:2.0721734114794,10.2503044242424)
--(axis cs:2.07667994620924,9.9492568)
--(axis cs:2.08173403191705,9.64820917575758)
--(axis cs:2.08740123903369,9.34716155151515)
--(axis cs:2.09374404489267,9.04611392727273)
--(axis cs:2.10081972131167,8.7450663030303)
--(axis cs:2.10867812985182,8.44401867878788)
--(axis cs:2.11735948182882,8.14297105454546)
--(axis cs:2.12689212698295,7.84192343030303)
--(axis cs:2.1372904400779,7.54087580606061)
--(axis cs:2.14855287822987,7.23982818181818)
--(axis cs:2.16066028313079,6.93878055757576)
--(axis cs:2.17357450122794,6.63773293333333)
--(axis cs:2.18723739112353,6.33668530909091)
--(axis cs:2.20157028081613,6.03563768484848)
--(axis cs:2.21647392787853,5.73459006060606)
--(axis cs:2.23182902333767,5.43354243636364)
--(axis cs:2.24749726510869,5.13249481212121)
--(axis cs:2.26332300970175,4.83144718787879)
--(axis cs:2.27913549207186,4.53039956363636)
--(axis cs:2.29475158355888,4.22935193939394)
--(axis cs:2.30997903761968,3.92830431515152)
--(axis cs:2.32462015332729,3.62725669090909)
--(axis cs:2.33847576828873,3.32620906666667)
--(axis cs:2.35134947660694,3.02516144242424)
--(axis cs:2.36305195462851,2.72411381818182)
--(axis cs:2.3734052682318,2.42306619393939)
--(axis cs:2.38224703092817,2.12201856969697)
--(axis cs:2.38943428249619,1.82097094545455)
--(axis cs:2.39484696344711,1.51992332121212)
--(axis cs:2.39839087128691,1.2188756969697)
--(axis cs:2.4,0.917828072727273)
--(axis cs:2.39963818388649,0.616780448484848)
--(axis cs:2.39729999006606,0.315732824242424)
--(axis cs:2.39301082965206,0.0146852)
--cycle;

\addplot [line width=0.75pt, darkslategray61]
table {%
0 1.683971518
0 227.583080252
};
\addplot [line width=2.25pt, darkslategray61]
table {%
0 23.034927962
0 176.976051352
};
\addplot [semithick, darkslategray61, mark=-, mark size=2.34375, mark options={solid,fill=white}]
table {%
0 55.923386335
};
\addplot [line width=0.75pt, darkslategray61]
table {%
1 0.053029134
1 157.36404766
};
\addplot [line width=2.25pt, darkslategray61]
table {%
1 15.730662998
1 123.2655110505
};
\addplot [semithick, darkslategray61, mark=-, mark size=2.34375, mark options={solid,fill=white}]
table {%
1 28.2057292185
};
\addplot [line width=0.75pt, darkslategray61]
table {%
2 0.0146852
2 3.11293
};
\addplot [line width=2.25pt, darkslategray61]
table {%
2 0.1289125
2 2.177665
};
\addplot [semithick, darkslategray61, mark=-, mark size=2.34375, mark options={solid,fill=white}]
table {%
2 0.565673
};
\end{axis}

\end{tikzpicture}

%% file: figures/qos-robustness.tex
\begin{tikzpicture}

\definecolor{darkgray176}{RGB}{176,176,176}
\definecolor{darkslategray63}{RGB}{63,63,63}
\definecolor{peru22412844}{RGB}{224,128,44}
\definecolor{steelblue49115161}{RGB}{49,115,161}

\begin{groupplot}[group style={group size=2 by 1}]
\nextgroupplot[
tick align=outside,
tick pos=left,
x grid style={darkgray176},
xlabel={Delay},
xmin=-0.5, xmax=1.5,
xtick style={color=black},
xtick={0,1},
xticklabels={MTR,vMTR},
y grid style={darkgray176},
ylabel={QoS robustness},
ymin=0, ymax=1,
ytick style={color=black}
]
\path [draw=darkslategray63, fill=steelblue49115161, line width=0.5pt]
(axis cs:0.00139150107804296,0.111588108907547)
--(axis cs:-0.00139150107804296,0.111588108907547)
--(axis cs:-0.00128639546372311,0.12056196639333)
--(axis cs:-0.00101636001303194,0.129535823879112)
--(axis cs:-0.00068628663384409,0.138509681364894)
--(axis cs:-0.000396078208293963,0.147483538850677)
--(axis cs:-0.00019554856089172,0.156457396336459)
--(axis cs:-8.34374065550129e-05,0.165431253822241)
--(axis cs:-3.43477344618287e-05,0.174405111308024)
--(axis cs:-2.614425417216e-05,0.183378968793806)
--(axis cs:-5.52995887229017e-05,0.192352826279588)
--(axis cs:-0.000141547808418516,0.201326683765371)
--(axis cs:-0.00032251343029231,0.210300541251153)
--(axis cs:-0.000632624989963815,0.219274398736936)
--(axis cs:-0.00106580806407208,0.228248256222718)
--(axis cs:-0.00154324471044448,0.2372221137085)
--(axis cs:-0.00192417686579181,0.246195971194283)
--(axis cs:-0.00207674456925985,0.255169828680065)
--(axis cs:-0.00197102429144828,0.264143686165847)
--(axis cs:-0.00172323938563346,0.27311754365163)
--(axis cs:-0.00155234505478375,0.282091401137412)
--(axis cs:-0.00167603007526013,0.291065258623194)
--(axis cs:-0.00220546456934755,0.300039116108977)
--(axis cs:-0.00308434441037527,0.309012973594759)
--(axis cs:-0.00409855465432257,0.317986831080541)
--(axis cs:-0.00496917376737004,0.326960688566324)
--(axis cs:-0.00549982311862202,0.335934546052106)
--(axis cs:-0.00568691159086225,0.344908403537888)
--(axis cs:-0.00570060609110531,0.353882261023671)
--(axis cs:-0.00574949546235486,0.362856118509453)
--(axis cs:-0.00595495435453108,0.371829975995235)
--(axis cs:-0.00634373022679368,0.380803833481018)
--(axis cs:-0.0069336887031621,0.3897776909668)
--(axis cs:-0.00779796858014094,0.398751548452583)
--(axis cs:-0.00904212887534592,0.407725405938365)
--(axis cs:-0.0107361180224093,0.416699263424147)
--(axis cs:-0.0128756754404305,0.42567312090993)
--(axis cs:-0.0153877726774629,0.434646978395712)
--(axis cs:-0.0181392732259973,0.443620835881494)
--(axis cs:-0.0209327474181293,0.452594693367277)
--(axis cs:-0.0235288977091892,0.461568550853059)
--(axis cs:-0.0257183128747899,0.470542408338841)
--(axis cs:-0.0273887714806968,0.479516265824624)
--(axis cs:-0.0285255419599149,0.488490123310406)
--(axis cs:-0.0291727095396266,0.497463980796188)
--(axis cs:-0.0294351571122542,0.506437838281971)
--(axis cs:-0.0295260902381494,0.515411695767753)
--(axis cs:-0.029778092979835,0.524385553253535)
--(axis cs:-0.0305704847742558,0.533359410739318)
--(axis cs:-0.0322289139866522,0.5423332682251)
--(axis cs:-0.0349790308502608,0.551307125710882)
--(axis cs:-0.0389603511162273,0.560280983196665)
--(axis cs:-0.0442347816735462,0.569254840682447)
--(axis cs:-0.0507345180521867,0.57822869816823)
--(axis cs:-0.0581556724488198,0.587202555654012)
--(axis cs:-0.0658504243978386,0.596176413139794)
--(axis cs:-0.0727965136670194,0.605150270625577)
--(axis cs:-0.0777557234714583,0.614124128111359)
--(axis cs:-0.0797052150933979,0.623097985597141)
--(axis cs:-0.0784216846304452,0.632071843082924)
--(axis cs:-0.0748321936548819,0.641045700568706)
--(axis cs:-0.0707722498838092,0.650019558054488)
--(axis cs:-0.068219145017674,0.658993415540271)
--(axis cs:-0.0684901359849722,0.667967273026053)
--(axis cs:-0.0718840983752834,0.676941130511835)
--(axis cs:-0.0778804589889605,0.685914987997618)
--(axis cs:-0.0856544162227027,0.6948888454834)
--(axis cs:-0.0945194221428165,0.703862702969182)
--(axis cs:-0.104008624671748,0.712836560454965)
--(axis cs:-0.113645904275719,0.721810417940747)
--(axis cs:-0.122789077548032,0.730784275426529)
--(axis cs:-0.130862607029894,0.739758132912312)
--(axis cs:-0.137839341637312,0.748731990398094)
--(axis cs:-0.144492257815951,0.757705847883877)
--(axis cs:-0.152098449792174,0.766679705369659)
--(axis cs:-0.16174081972897,0.775653562855441)
--(axis cs:-0.173679843952853,0.784627420341224)
--(axis cs:-0.187252621838654,0.793601277827006)
--(axis cs:-0.201427107198586,0.802575135312788)
--(axis cs:-0.215647414733061,0.811548992798571)
--(axis cs:-0.230308237048514,0.820522850284353)
--(axis cs:-0.246465360296895,0.829496707770135)
--(axis cs:-0.265067597810608,0.838470565255918)
--(axis cs:-0.286366631685794,0.8474444227417)
--(axis cs:-0.309788805494691,0.856418280227482)
--(axis cs:-0.333943731386718,0.865392137713265)
--(axis cs:-0.356483492559477,0.874365995199047)
--(axis cs:-0.374209249217763,0.883339852684829)
--(axis cs:-0.384053251159796,0.892313710170612)
--(axis cs:-0.384685809853443,0.901287567656394)
--(axis cs:-0.377515994846426,0.910261425142176)
--(axis cs:-0.366217080749893,0.919235282627959)
--(axis cs:-0.35533557426986,0.928209140113741)
--(axis cs:-0.349210004994604,0.937182997599524)
--(axis cs:-0.351394250798207,0.946156855085306)
--(axis cs:-0.363440450096107,0.955130712571088)
--(axis cs:-0.382440018049673,0.964104570056871)
--(axis cs:-0.399162790219156,0.973078427542653)
--(axis cs:-0.4,0.982052285028435)
--(axis cs:-0.373499729363517,0.991026142514218)
--(axis cs:-0.317652132558836,1)
--(axis cs:0.317652132558836,1)
--(axis cs:0.317652132558836,1)
--(axis cs:0.373499729363517,0.991026142514218)
--(axis cs:0.4,0.982052285028435)
--(axis cs:0.399162790219156,0.973078427542653)
--(axis cs:0.382440018049673,0.964104570056871)
--(axis cs:0.363440450096107,0.955130712571088)
--(axis cs:0.351394250798207,0.946156855085306)
--(axis cs:0.349210004994604,0.937182997599524)
--(axis cs:0.35533557426986,0.928209140113741)
--(axis cs:0.366217080749893,0.919235282627959)
--(axis cs:0.377515994846426,0.910261425142176)
--(axis cs:0.384685809853443,0.901287567656394)
--(axis cs:0.384053251159796,0.892313710170612)
--(axis cs:0.374209249217763,0.883339852684829)
--(axis cs:0.356483492559477,0.874365995199047)
--(axis cs:0.333943731386718,0.865392137713265)
--(axis cs:0.309788805494691,0.856418280227482)
--(axis cs:0.286366631685794,0.8474444227417)
--(axis cs:0.265067597810608,0.838470565255918)
--(axis cs:0.246465360296895,0.829496707770135)
--(axis cs:0.230308237048514,0.820522850284353)
--(axis cs:0.215647414733061,0.811548992798571)
--(axis cs:0.201427107198586,0.802575135312788)
--(axis cs:0.187252621838654,0.793601277827006)
--(axis cs:0.173679843952853,0.784627420341224)
--(axis cs:0.16174081972897,0.775653562855441)
--(axis cs:0.152098449792174,0.766679705369659)
--(axis cs:0.144492257815951,0.757705847883877)
--(axis cs:0.137839341637312,0.748731990398094)
--(axis cs:0.130862607029894,0.739758132912312)
--(axis cs:0.122789077548032,0.730784275426529)
--(axis cs:0.113645904275719,0.721810417940747)
--(axis cs:0.104008624671748,0.712836560454965)
--(axis cs:0.0945194221428165,0.703862702969182)
--(axis cs:0.0856544162227027,0.6948888454834)
--(axis cs:0.0778804589889605,0.685914987997618)
--(axis cs:0.0718840983752834,0.676941130511835)
--(axis cs:0.0684901359849722,0.667967273026053)
--(axis cs:0.068219145017674,0.658993415540271)
--(axis cs:0.0707722498838092,0.650019558054488)
--(axis cs:0.0748321936548819,0.641045700568706)
--(axis cs:0.0784216846304452,0.632071843082924)
--(axis cs:0.0797052150933979,0.623097985597141)
--(axis cs:0.0777557234714583,0.614124128111359)
--(axis cs:0.0727965136670194,0.605150270625577)
--(axis cs:0.0658504243978386,0.596176413139794)
--(axis cs:0.0581556724488198,0.587202555654012)
--(axis cs:0.0507345180521867,0.57822869816823)
--(axis cs:0.0442347816735462,0.569254840682447)
--(axis cs:0.0389603511162273,0.560280983196665)
--(axis cs:0.0349790308502608,0.551307125710882)
--(axis cs:0.0322289139866522,0.5423332682251)
--(axis cs:0.0305704847742558,0.533359410739318)
--(axis cs:0.029778092979835,0.524385553253535)
--(axis cs:0.0295260902381494,0.515411695767753)
--(axis cs:0.0294351571122542,0.506437838281971)
--(axis cs:0.0291727095396266,0.497463980796188)
--(axis cs:0.0285255419599149,0.488490123310406)
--(axis cs:0.0273887714806968,0.479516265824624)
--(axis cs:0.0257183128747899,0.470542408338841)
--(axis cs:0.0235288977091892,0.461568550853059)
--(axis cs:0.0209327474181293,0.452594693367277)
--(axis cs:0.0181392732259973,0.443620835881494)
--(axis cs:0.0153877726774629,0.434646978395712)
--(axis cs:0.0128756754404305,0.42567312090993)
--(axis cs:0.0107361180224093,0.416699263424147)
--(axis cs:0.00904212887534592,0.407725405938365)
--(axis cs:0.00779796858014094,0.398751548452583)
--(axis cs:0.0069336887031621,0.3897776909668)
--(axis cs:0.00634373022679368,0.380803833481018)
--(axis cs:0.00595495435453108,0.371829975995235)
--(axis cs:0.00574949546235486,0.362856118509453)
--(axis cs:0.00570060609110531,0.353882261023671)
--(axis cs:0.00568691159086225,0.344908403537888)
--(axis cs:0.00549982311862202,0.335934546052106)
--(axis cs:0.00496917376737004,0.326960688566324)
--(axis cs:0.00409855465432257,0.317986831080541)
--(axis cs:0.00308434441037527,0.309012973594759)
--(axis cs:0.00220546456934755,0.300039116108977)
--(axis cs:0.00167603007526013,0.291065258623194)
--(axis cs:0.00155234505478375,0.282091401137412)
--(axis cs:0.00172323938563346,0.27311754365163)
--(axis cs:0.00197102429144828,0.264143686165847)
--(axis cs:0.00207674456925985,0.255169828680065)
--(axis cs:0.00192417686579181,0.246195971194283)
--(axis cs:0.00154324471044448,0.2372221137085)
--(axis cs:0.00106580806407208,0.228248256222718)
--(axis cs:0.000632624989963815,0.219274398736936)
--(axis cs:0.00032251343029231,0.210300541251153)
--(axis cs:0.000141547808418516,0.201326683765371)
--(axis cs:5.52995887229017e-05,0.192352826279588)
--(axis cs:2.614425417216e-05,0.183378968793806)
--(axis cs:3.43477344618287e-05,0.174405111308024)
--(axis cs:8.34374065550129e-05,0.165431253822241)
--(axis cs:0.00019554856089172,0.156457396336459)
--(axis cs:0.000396078208293963,0.147483538850677)
--(axis cs:0.00068628663384409,0.138509681364894)
--(axis cs:0.00101636001303194,0.129535823879112)
--(axis cs:0.00128639546372311,0.12056196639333)
--(axis cs:0.00139150107804296,0.111588108907547)
--cycle;

\path [draw=darkslategray63, fill=peru22412844, line width=0.5pt]
(axis cs:1.00227790566296,0.111588108907547)
--(axis cs:0.99772209433704,0.111588108907547)
--(axis cs:0.997496952120608,0.12056196639333)
--(axis cs:0.997390165518425,0.129535823879112)
--(axis cs:0.997392549451976,0.138509681364894)
--(axis cs:0.997463370731101,0.147483538850677)
--(axis cs:0.997528899660181,0.156457396336459)
--(axis cs:0.997489561053057,0.165431253822241)
--(axis cs:0.997241626310488,0.174405111308024)
--(axis cs:0.996711336039519,0.183378968793806)
--(axis cs:0.995888038883918,0.192352826279588)
--(axis cs:0.994838429640583,0.201326683765371)
--(axis cs:0.993692171499484,0.210300541251153)
--(axis cs:0.992605208990416,0.219274398736936)
--(axis cs:0.991718165910949,0.228248256222718)
--(axis cs:0.991124319791112,0.2372221137085)
--(axis cs:0.990849292530907,0.246195971194283)
--(axis cs:0.990837427777873,0.255169828680065)
--(axis cs:0.990946507244527,0.264143686165847)
--(axis cs:0.990963966721845,0.27311754365163)
--(axis cs:0.990655532680666,0.282091401137412)
--(axis cs:0.989834797656001,0.291065258623194)
--(axis cs:0.988416693369217,0.300039116108977)
--(axis cs:0.986415736377388,0.309012973594759)
--(axis cs:0.983881197108324,0.317986831080541)
--(axis cs:0.980806622413835,0.326960688566324)
--(axis cs:0.977075965003914,0.335934546052106)
--(axis cs:0.97249148900601,0.344908403537888)
--(axis cs:0.966878686505589,0.353882261023671)
--(axis cs:0.960213157990866,0.362856118509453)
--(axis cs:0.952695588674619,0.371829975995235)
--(axis cs:0.944723882112188,0.380803833481018)
--(axis cs:0.936761977128954,0.3897776909668)
--(axis cs:0.929156516792188,0.398751548452583)
--(axis cs:0.921981634807027,0.407725405938365)
--(axis cs:0.914986175125403,0.416699263424147)
--(axis cs:0.907678475846506,0.42567312090993)
--(axis cs:0.899527358344083,0.434646978395712)
--(axis cs:0.890208807594612,0.443620835881494)
--(axis cs:0.879807613969955,0.452594693367277)
--(axis cs:0.86889847402479,0.461568550853059)
--(axis cs:0.858470623150678,0.470542408338841)
--(axis cs:0.849705795852549,0.479516265824624)
--(axis cs:0.84366009190486,0.488490123310406)
--(axis cs:0.840934736954972,0.497463980796188)
--(axis cs:0.841443864684882,0.506437838281971)
--(axis cs:0.844379789897862,0.515411695767753)
--(axis cs:0.848416575260423,0.524385553253535)
--(axis cs:0.852089818100679,0.533359410739318)
--(axis cs:0.854197494771619,0.5423332682251)
--(axis cs:0.854050122006661,0.551307125710882)
--(axis cs:0.851483245452848,0.560280983196665)
--(axis cs:0.846684524148771,0.569254840682447)
--(axis cs:0.839989931893497,0.57822869816823)
--(axis cs:0.831797449932648,0.587202555654012)
--(axis cs:0.822634354105633,0.596176413139794)
--(axis cs:0.813274731778905,0.605150270625577)
--(axis cs:0.804741128797145,0.614124128111359)
--(axis cs:0.79809282915947,0.623097985597141)
--(axis cs:0.794062309985077,0.632071843082924)
--(axis cs:0.79273581459525,0.641045700568706)
--(axis cs:0.79347774536888,0.650019558054488)
--(axis cs:0.795160300028317,0.658993415540271)
--(axis cs:0.796577189125068,0.667967273026053)
--(axis cs:0.796824442575576,0.676941130511835)
--(axis cs:0.795482080000365,0.685914987997618)
--(axis cs:0.7925732928429,0.6948888454834)
--(axis cs:0.788397623050483,0.703862702969182)
--(axis cs:0.783355011350038,0.712836560454965)
--(axis cs:0.77781684700863,0.721810417940747)
--(axis cs:0.77203783733433,0.730784275426529)
--(axis cs:0.766094939610321,0.739758132912312)
--(axis cs:0.759870229805911,0.748731990398094)
--(axis cs:0.75310573933133,0.757705847883877)
--(axis cs:0.745522267773739,0.766679705369659)
--(axis cs:0.736942726472317,0.775653562855441)
--(axis cs:0.727346126946161,0.784627420341224)
--(axis cs:0.716819096944225,0.793601277827006)
--(axis cs:0.705439840118202,0.802575135312788)
--(axis cs:0.693180406662605,0.811548992798571)
--(axis cs:0.679915933153705,0.820522850284353)
--(axis cs:0.665576559777853,0.829496707770135)
--(axis cs:0.650388906284694,0.838470565255918)
--(axis cs:0.635077772255393,0.8474444227417)
--(axis cs:0.62089118695622,0.856418280227482)
--(axis cs:0.609391937211381,0.865392137713265)
--(axis cs:0.60208069992383,0.874365995199047)
--(axis cs:0.6,0.883339852684829)
--(axis cs:0.603466260680572,0.892313710170612)
--(axis cs:0.612007693043672,0.901287567656394)
--(axis cs:0.624504510352418,0.910261425142176)
--(axis cs:0.639475439742479,0.919235282627959)
--(axis cs:0.655437037132667,0.928209140113741)
--(axis cs:0.671268287985579,0.937182997599524)
--(axis cs:0.686524933376774,0.946156855085306)
--(axis cs:0.701642402918648,0.955130712571088)
--(axis cs:0.717933491552746,0.964104570056871)
--(axis cs:0.737269009135201,0.973078427542653)
--(axis cs:0.761409323216482,0.982052285028435)
--(axis cs:0.79115906470019,0.991026142514218)
--(axis cs:0.825727599371757,1)
--(axis cs:1.17427240062824,1)
--(axis cs:1.17427240062824,1)
--(axis cs:1.20884093529981,0.991026142514218)
--(axis cs:1.23859067678352,0.982052285028435)
--(axis cs:1.2627309908648,0.973078427542653)
--(axis cs:1.28206650844725,0.964104570056871)
--(axis cs:1.29835759708135,0.955130712571088)
--(axis cs:1.31347506662323,0.946156855085306)
--(axis cs:1.32873171201442,0.937182997599524)
--(axis cs:1.34456296286733,0.928209140113741)
--(axis cs:1.36052456025752,0.919235282627959)
--(axis cs:1.37549548964758,0.910261425142176)
--(axis cs:1.38799230695633,0.901287567656394)
--(axis cs:1.39653373931943,0.892313710170612)
--(axis cs:1.4,0.883339852684829)
--(axis cs:1.39791930007617,0.874365995199047)
--(axis cs:1.39060806278862,0.865392137713265)
--(axis cs:1.37910881304378,0.856418280227482)
--(axis cs:1.36492222774461,0.8474444227417)
--(axis cs:1.34961109371531,0.838470565255918)
--(axis cs:1.33442344022215,0.829496707770135)
--(axis cs:1.3200840668463,0.820522850284353)
--(axis cs:1.3068195933374,0.811548992798571)
--(axis cs:1.2945601598818,0.802575135312788)
--(axis cs:1.28318090305578,0.793601277827006)
--(axis cs:1.27265387305384,0.784627420341224)
--(axis cs:1.26305727352768,0.775653562855441)
--(axis cs:1.25447773222626,0.766679705369659)
--(axis cs:1.24689426066867,0.757705847883877)
--(axis cs:1.24012977019409,0.748731990398094)
--(axis cs:1.23390506038968,0.739758132912312)
--(axis cs:1.22796216266567,0.730784275426529)
--(axis cs:1.22218315299137,0.721810417940747)
--(axis cs:1.21664498864996,0.712836560454965)
--(axis cs:1.21160237694952,0.703862702969182)
--(axis cs:1.2074267071571,0.6948888454834)
--(axis cs:1.20451791999964,0.685914987997618)
--(axis cs:1.20317555742442,0.676941130511835)
--(axis cs:1.20342281087493,0.667967273026053)
--(axis cs:1.20483969997168,0.658993415540271)
--(axis cs:1.20652225463112,0.650019558054488)
--(axis cs:1.20726418540475,0.641045700568706)
--(axis cs:1.20593769001492,0.632071843082924)
--(axis cs:1.20190717084053,0.623097985597141)
--(axis cs:1.19525887120286,0.614124128111359)
--(axis cs:1.18672526822109,0.605150270625577)
--(axis cs:1.17736564589437,0.596176413139794)
--(axis cs:1.16820255006735,0.587202555654012)
--(axis cs:1.1600100681065,0.57822869816823)
--(axis cs:1.15331547585123,0.569254840682447)
--(axis cs:1.14851675454715,0.560280983196665)
--(axis cs:1.14594987799334,0.551307125710882)
--(axis cs:1.14580250522838,0.5423332682251)
--(axis cs:1.14791018189932,0.533359410739318)
--(axis cs:1.15158342473958,0.524385553253535)
--(axis cs:1.15562021010214,0.515411695767753)
--(axis cs:1.15855613531512,0.506437838281971)
--(axis cs:1.15906526304503,0.497463980796188)
--(axis cs:1.15633990809514,0.488490123310406)
--(axis cs:1.15029420414745,0.479516265824624)
--(axis cs:1.14152937684932,0.470542408338841)
--(axis cs:1.13110152597521,0.461568550853059)
--(axis cs:1.12019238603004,0.452594693367277)
--(axis cs:1.10979119240539,0.443620835881494)
--(axis cs:1.10047264165592,0.434646978395712)
--(axis cs:1.09232152415349,0.42567312090993)
--(axis cs:1.0850138248746,0.416699263424147)
--(axis cs:1.07801836519297,0.407725405938365)
--(axis cs:1.07084348320781,0.398751548452583)
--(axis cs:1.06323802287105,0.3897776909668)
--(axis cs:1.05527611788781,0.380803833481018)
--(axis cs:1.04730441132538,0.371829975995235)
--(axis cs:1.03978684200913,0.362856118509453)
--(axis cs:1.03312131349441,0.353882261023671)
--(axis cs:1.02750851099399,0.344908403537888)
--(axis cs:1.02292403499609,0.335934546052106)
--(axis cs:1.01919337758617,0.326960688566324)
--(axis cs:1.01611880289168,0.317986831080541)
--(axis cs:1.01358426362261,0.309012973594759)
--(axis cs:1.01158330663078,0.300039116108977)
--(axis cs:1.010165202344,0.291065258623194)
--(axis cs:1.00934446731933,0.282091401137412)
--(axis cs:1.00903603327815,0.27311754365163)
--(axis cs:1.00905349275547,0.264143686165847)
--(axis cs:1.00916257222213,0.255169828680065)
--(axis cs:1.00915070746909,0.246195971194283)
--(axis cs:1.00887568020889,0.2372221137085)
--(axis cs:1.00828183408905,0.228248256222718)
--(axis cs:1.00739479100958,0.219274398736936)
--(axis cs:1.00630782850052,0.210300541251153)
--(axis cs:1.00516157035942,0.201326683765371)
--(axis cs:1.00411196111608,0.192352826279588)
--(axis cs:1.00328866396048,0.183378968793806)
--(axis cs:1.00275837368951,0.174405111308024)
--(axis cs:1.00251043894694,0.165431253822241)
--(axis cs:1.00247110033982,0.156457396336459)
--(axis cs:1.0025366292689,0.147483538850677)
--(axis cs:1.00260745054802,0.138509681364894)
--(axis cs:1.00260983448158,0.129535823879112)
--(axis cs:1.00250304787939,0.12056196639333)
--(axis cs:1.00227790566296,0.111588108907547)
--cycle;

\addplot [line width=0.75pt, darkslategray63]
table {%
0 0.548737351474722
0 1
};
\addplot [line width=2.25pt, darkslategray63]
table {%
0 0.789169387437537
0 0.950526999659693
};
\addplot [semithick, darkslategray63, mark=-, mark size=2.34375, mark options={solid,fill=white}]
table {%
0 0.883876128678675
};
\addplot [line width=0.75pt, darkslategray63]
table {%
1 0.251497554714283
1 1
};
\addplot [line width=2.25pt, darkslategray63]
table {%
1 0.626724902524653
1 0.894021739130435
};
\addplot [semithick, darkslategray63, mark=-, mark size=2.34375, mark options={solid,fill=white}]
table {%
1 0.792452830188679
};

\nextgroupplot[
scaled y ticks=manual:{}{\pgfmathparse{#1}},
tick align=outside,
tick pos=left,
x grid style={darkgray176},
xlabel={Loss},
xmin=-0.5, xmax=1.5,
xtick style={color=black},
xtick={0,1},
xticklabels={MTR,vMTR},
y grid style={darkgray176},
ymin=0, ymax=1,
ytick style={color=black},
yticklabels={}
]
\path [draw=darkslategray63, fill=steelblue49115161, line width=0.5pt]
(axis cs:0.00105901252762996,0.298947198254087)
--(axis cs:-0.00105901252762996,0.298947198254087)
--(axis cs:-0.00136194189573213,0.306023471376753)
--(axis cs:-0.0016490408516991,0.313099744499419)
--(axis cs:-0.00185709424696232,0.320176017622085)
--(axis cs:-0.00192494589522512,0.327252290744752)
--(axis cs:-0.00184725670330982,0.334328563867418)
--(axis cs:-0.00170953583473073,0.341404836990084)
--(axis cs:-0.00166399171779798,0.34848111011275)
--(axis cs:-0.00185373543374379,0.355557383235416)
--(axis cs:-0.00233528905725681,0.362633656358083)
--(axis cs:-0.00304830066938265,0.369709929480749)
--(axis cs:-0.00384801823488042,0.376786202603415)
--(axis cs:-0.00458053303864377,0.383862475726081)
--(axis cs:-0.00515792095705869,0.390938748848747)
--(axis cs:-0.00558725646537763,0.398015021971413)
--(axis cs:-0.00593428237760753,0.40509129509408)
--(axis cs:-0.00625139804194139,0.412167568216746)
--(axis cs:-0.00652969812333325,0.419243841339412)
--(axis cs:-0.00671065113717383,0.426320114462078)
--(axis cs:-0.0067361719344389,0.433396387584744)
--(axis cs:-0.00658736269161905,0.440472660707411)
--(axis cs:-0.00628802668044406,0.447548933830077)
--(axis cs:-0.005890542196274,0.454625206952743)
--(axis cs:-0.005470014639935,0.461701480075409)
--(axis cs:-0.00512701614900056,0.468777753198075)
--(axis cs:-0.00497652548050916,0.475854026320742)
--(axis cs:-0.00510743939125283,0.482930299443408)
--(axis cs:-0.00552798723962614,0.490006572566074)
--(axis cs:-0.00614035148178096,0.49708284568874)
--(axis cs:-0.00678052827151015,0.504159118811406)
--(axis cs:-0.00730783839799228,0.511235391934072)
--(axis cs:-0.00767544344364392,0.518311665056739)
--(axis cs:-0.00792450132599846,0.525387938179405)
--(axis cs:-0.00812193999178194,0.532464211302071)
--(axis cs:-0.00831834318501095,0.539540484424737)
--(axis cs:-0.00856617557354663,0.546616757547403)
--(axis cs:-0.00895319984439976,0.55369303067007)
--(axis cs:-0.00957999727109401,0.560769303792736)
--(axis cs:-0.0104794757740697,0.567845576915402)
--(axis cs:-0.0115580666005525,0.574921850038068)
--(axis cs:-0.0126318867515776,0.581998123160734)
--(axis cs:-0.0135365550248681,0.5890743962834)
--(axis cs:-0.0142109938423126,0.596150669406067)
--(axis cs:-0.0146821827284138,0.603226942528733)
--(axis cs:-0.0149804990261253,0.610303215651399)
--(axis cs:-0.0150854666564447,0.617379488774065)
--(axis cs:-0.0149728038329746,0.624455761896731)
--(axis cs:-0.0147379814232159,0.631532035019398)
--(axis cs:-0.0146928025711392,0.638608308142064)
--(axis cs:-0.0153303942715528,0.64568458126473)
--(axis cs:-0.0171344387793133,0.652760854387396)
--(axis cs:-0.0203249660974596,0.659837127510062)
--(axis cs:-0.0247098160332486,0.666913400632729)
--(axis cs:-0.029773962127507,0.673989673755395)
--(axis cs:-0.0349806298401377,0.681065946878061)
--(axis cs:-0.0400869390513704,0.688142220000727)
--(axis cs:-0.0452495957596739,0.695218493123393)
--(axis cs:-0.0508570968867168,0.702294766246059)
--(axis cs:-0.0572390886670573,0.709371039368726)
--(axis cs:-0.0644826087901764,0.716447312491392)
--(axis cs:-0.072465150864984,0.723523585614058)
--(axis cs:-0.081013102751969,0.730599858736724)
--(axis cs:-0.0900053764706974,0.73767613185939)
--(axis cs:-0.099344024807355,0.744752404982057)
--(axis cs:-0.108870193669908,0.751828678104723)
--(axis cs:-0.118326467300672,0.758904951227389)
--(axis cs:-0.127367154764931,0.765981224350055)
--(axis cs:-0.135576515712974,0.773057497472721)
--(axis cs:-0.142526988125646,0.780133770595387)
--(axis cs:-0.147933757133059,0.787210043718054)
--(axis cs:-0.151835912879254,0.79428631684072)
--(axis cs:-0.154619449822874,0.801362589963386)
--(axis cs:-0.156798248855074,0.808438863086052)
--(axis cs:-0.15871077242069,0.815515136208718)
--(axis cs:-0.160393607683188,0.822591409331385)
--(axis cs:-0.161739133895618,0.829667682454051)
--(axis cs:-0.162804729331111,0.836743955576717)
--(axis cs:-0.164043061967799,0.843820228699383)
--(axis cs:-0.166305797904074,0.850896501822049)
--(axis cs:-0.170604517230671,0.857972774944716)
--(axis cs:-0.177693893055014,0.865049048067382)
--(axis cs:-0.187617248182986,0.872125321190048)
--(axis cs:-0.199463395961063,0.879201594312714)
--(axis cs:-0.211586936587562,0.88627786743538)
--(axis cs:-0.222286695913725,0.893354140558047)
--(axis cs:-0.230576784856219,0.900430413680713)
--(axis cs:-0.236598743933717,0.907506686803379)
--(axis cs:-0.241512025368802,0.914582959926045)
--(axis cs:-0.247042151998016,0.921659233048711)
--(axis cs:-0.254966364236179,0.928735506171377)
--(axis cs:-0.266738635923085,0.935811779294044)
--(axis cs:-0.283352371621645,0.94288805241671)
--(axis cs:-0.305324938539802,0.949964325539376)
--(axis cs:-0.332318956193089,0.957040598662042)
--(axis cs:-0.361888186457147,0.964116871784708)
--(axis cs:-0.387815646235901,0.971193144907375)
--(axis cs:-0.4,0.978269418030041)
--(axis cs:-0.387825888261687,0.985345691152707)
--(axis cs:-0.346177342695686,0.992421964275373)
--(axis cs:-0.279896939101408,0.999498237398039)
--(axis cs:0.279896939101408,0.999498237398039)
--(axis cs:0.279896939101408,0.999498237398039)
--(axis cs:0.346177342695686,0.992421964275373)
--(axis cs:0.387825888261687,0.985345691152707)
--(axis cs:0.4,0.978269418030041)
--(axis cs:0.387815646235901,0.971193144907375)
--(axis cs:0.361888186457147,0.964116871784708)
--(axis cs:0.332318956193089,0.957040598662042)
--(axis cs:0.305324938539802,0.949964325539376)
--(axis cs:0.283352371621645,0.94288805241671)
--(axis cs:0.266738635923085,0.935811779294044)
--(axis cs:0.254966364236179,0.928735506171377)
--(axis cs:0.247042151998016,0.921659233048711)
--(axis cs:0.241512025368802,0.914582959926045)
--(axis cs:0.236598743933717,0.907506686803379)
--(axis cs:0.230576784856219,0.900430413680713)
--(axis cs:0.222286695913725,0.893354140558047)
--(axis cs:0.211586936587562,0.88627786743538)
--(axis cs:0.199463395961063,0.879201594312714)
--(axis cs:0.187617248182986,0.872125321190048)
--(axis cs:0.177693893055014,0.865049048067382)
--(axis cs:0.170604517230671,0.857972774944716)
--(axis cs:0.166305797904074,0.850896501822049)
--(axis cs:0.164043061967799,0.843820228699383)
--(axis cs:0.162804729331111,0.836743955576717)
--(axis cs:0.161739133895618,0.829667682454051)
--(axis cs:0.160393607683188,0.822591409331385)
--(axis cs:0.15871077242069,0.815515136208718)
--(axis cs:0.156798248855074,0.808438863086052)
--(axis cs:0.154619449822874,0.801362589963386)
--(axis cs:0.151835912879254,0.79428631684072)
--(axis cs:0.147933757133059,0.787210043718054)
--(axis cs:0.142526988125646,0.780133770595387)
--(axis cs:0.135576515712974,0.773057497472721)
--(axis cs:0.127367154764931,0.765981224350055)
--(axis cs:0.118326467300672,0.758904951227389)
--(axis cs:0.108870193669908,0.751828678104723)
--(axis cs:0.099344024807355,0.744752404982057)
--(axis cs:0.0900053764706974,0.73767613185939)
--(axis cs:0.081013102751969,0.730599858736724)
--(axis cs:0.072465150864984,0.723523585614058)
--(axis cs:0.0644826087901764,0.716447312491392)
--(axis cs:0.0572390886670573,0.709371039368726)
--(axis cs:0.0508570968867168,0.702294766246059)
--(axis cs:0.0452495957596739,0.695218493123393)
--(axis cs:0.0400869390513704,0.688142220000727)
--(axis cs:0.0349806298401377,0.681065946878061)
--(axis cs:0.029773962127507,0.673989673755395)
--(axis cs:0.0247098160332486,0.666913400632729)
--(axis cs:0.0203249660974596,0.659837127510062)
--(axis cs:0.0171344387793133,0.652760854387396)
--(axis cs:0.0153303942715528,0.64568458126473)
--(axis cs:0.0146928025711392,0.638608308142064)
--(axis cs:0.0147379814232159,0.631532035019398)
--(axis cs:0.0149728038329746,0.624455761896731)
--(axis cs:0.0150854666564447,0.617379488774065)
--(axis cs:0.0149804990261253,0.610303215651399)
--(axis cs:0.0146821827284138,0.603226942528733)
--(axis cs:0.0142109938423126,0.596150669406067)
--(axis cs:0.0135365550248681,0.5890743962834)
--(axis cs:0.0126318867515776,0.581998123160734)
--(axis cs:0.0115580666005525,0.574921850038068)
--(axis cs:0.0104794757740697,0.567845576915402)
--(axis cs:0.00957999727109401,0.560769303792736)
--(axis cs:0.00895319984439976,0.55369303067007)
--(axis cs:0.00856617557354663,0.546616757547403)
--(axis cs:0.00831834318501095,0.539540484424737)
--(axis cs:0.00812193999178194,0.532464211302071)
--(axis cs:0.00792450132599846,0.525387938179405)
--(axis cs:0.00767544344364392,0.518311665056739)
--(axis cs:0.00730783839799228,0.511235391934072)
--(axis cs:0.00678052827151015,0.504159118811406)
--(axis cs:0.00614035148178096,0.49708284568874)
--(axis cs:0.00552798723962614,0.490006572566074)
--(axis cs:0.00510743939125283,0.482930299443408)
--(axis cs:0.00497652548050916,0.475854026320742)
--(axis cs:0.00512701614900056,0.468777753198075)
--(axis cs:0.005470014639935,0.461701480075409)
--(axis cs:0.005890542196274,0.454625206952743)
--(axis cs:0.00628802668044406,0.447548933830077)
--(axis cs:0.00658736269161905,0.440472660707411)
--(axis cs:0.0067361719344389,0.433396387584744)
--(axis cs:0.00671065113717383,0.426320114462078)
--(axis cs:0.00652969812333325,0.419243841339412)
--(axis cs:0.00625139804194139,0.412167568216746)
--(axis cs:0.00593428237760753,0.40509129509408)
--(axis cs:0.00558725646537763,0.398015021971413)
--(axis cs:0.00515792095705869,0.390938748848747)
--(axis cs:0.00458053303864377,0.383862475726081)
--(axis cs:0.00384801823488042,0.376786202603415)
--(axis cs:0.00304830066938265,0.369709929480749)
--(axis cs:0.00233528905725681,0.362633656358083)
--(axis cs:0.00185373543374379,0.355557383235416)
--(axis cs:0.00166399171779798,0.34848111011275)
--(axis cs:0.00170953583473073,0.341404836990084)
--(axis cs:0.00184725670330982,0.334328563867418)
--(axis cs:0.00192494589522512,0.327252290744752)
--(axis cs:0.00185709424696232,0.320176017622085)
--(axis cs:0.0016490408516991,0.313099744499419)
--(axis cs:0.00136194189573213,0.306023471376753)
--(axis cs:0.00105901252762996,0.298947198254087)
--cycle;

\path [draw=darkslategray63, fill=peru22412844, line width=0.5pt]
(axis cs:1.00103159717138,0.298947198254087)
--(axis cs:0.99896840282862,0.298947198254087)
--(axis cs:0.998865039906172,0.30603224511877)
--(axis cs:0.998824419839376,0.313117291983452)
--(axis cs:0.998831864968681,0.320202338848135)
--(axis cs:0.998871577792341,0.327287385712817)
--(axis cs:0.998906670972706,0.3343724325775)
--(axis cs:0.998852209120632,0.341457479442183)
--(axis cs:0.998582002351702,0.348542526306865)
--(axis cs:0.997982626562221,0.355627573171548)
--(axis cs:0.997024745635803,0.362712620036231)
--(axis cs:0.995801984016925,0.369797666900913)
--(axis cs:0.994501266735109,0.376882713765596)
--(axis cs:0.99330789227004,0.383967760630278)
--(axis cs:0.992299689960571,0.391052807494961)
--(axis cs:0.991412619004131,0.398137854359644)
--(axis cs:0.990521287455152,0.405222901224326)
--(axis cs:0.989581881222502,0.412307948089009)
--(axis cs:0.988720089189667,0.419392994953692)
--(axis cs:0.988188372923278,0.426478041818374)
--(axis cs:0.988225587460726,0.433563088683057)
--(axis cs:0.988919140433581,0.440648135547739)
--(axis cs:0.990150816128543,0.447733182412422)
--(axis cs:0.991643425549262,0.454818229277105)
--(axis cs:0.993065577270881,0.461903276141787)
--(axis cs:0.994125091776258,0.46898832300647)
--(axis cs:0.994612363312856,0.476073369871153)
--(axis cs:0.994425221601998,0.483158416735835)
--(axis cs:0.99362838948651,0.490243463600518)
--(axis cs:0.992515323282084,0.4973285104652)
--(axis cs:0.991544903273434,0.504413557329883)
--(axis cs:0.991090636353978,0.511498604194566)
--(axis cs:0.99115664264857,0.518583651059248)
--(axis cs:0.991330035175247,0.525668697923931)
--(axis cs:0.991057136155721,0.532753744788613)
--(axis cs:0.990008777192718,0.539838791653296)
--(axis cs:0.988204495296592,0.546923838517979)
--(axis cs:0.98582860933599,0.554008885382661)
--(axis cs:0.983004763728128,0.561093932247344)
--(axis cs:0.979814177010266,0.568178979112027)
--(axis cs:0.97652190612475,0.575264025976709)
--(axis cs:0.973689898709999,0.582349072841392)
--(axis cs:0.971942957428628,0.589434119706075)
--(axis cs:0.971529227440722,0.596519166570757)
--(axis cs:0.972076456759019,0.60360421343544)
--(axis cs:0.97280020528251,0.610689260300122)
--(axis cs:0.973000353411012,0.617774307164805)
--(axis cs:0.972415484280848,0.624859354029488)
--(axis cs:0.971164012820421,0.63194440089417)
--(axis cs:0.969400735935752,0.639029447758853)
--(axis cs:0.967044249305437,0.646114494623536)
--(axis cs:0.963811322762716,0.653199541488218)
--(axis cs:0.959500604522885,0.660284588352901)
--(axis cs:0.954249194380071,0.667369635217584)
--(axis cs:0.948487343894097,0.674454682082266)
--(axis cs:0.942567725012107,0.681539728946949)
--(axis cs:0.936382761228819,0.688624775811631)
--(axis cs:0.929360134898997,0.695709822676314)
--(axis cs:0.920877089264567,0.702794869540997)
--(axis cs:0.910696924658415,0.709879916405679)
--(axis cs:0.899020427066004,0.716964963270362)
--(axis cs:0.886180756413661,0.724050010135044)
--(axis cs:0.872359448799983,0.731135056999727)
--(axis cs:0.857575639384772,0.73822010386441)
--(axis cs:0.841845244575896,0.745305150729092)
--(axis cs:0.825295997502216,0.752390197593775)
--(axis cs:0.808197954478252,0.759475244458458)
--(axis cs:0.790986956647383,0.76656029132314)
--(axis cs:0.774277779092584,0.773645338187823)
--(axis cs:0.758788414863957,0.780730385052505)
--(axis cs:0.745177711715246,0.787815431917188)
--(axis cs:0.733908879921456,0.794900478781871)
--(axis cs:0.725219869739256,0.801985525646553)
--(axis cs:0.719151091658093,0.809070572511236)
--(axis cs:0.715535304124937,0.816155619375919)
--(axis cs:0.713939440196059,0.823240666240601)
--(axis cs:0.713603113900846,0.830325713105284)
--(axis cs:0.713369826137237,0.837410759969967)
--(axis cs:0.711635147393475,0.844495806834649)
--(axis cs:0.70651897699311,0.851580853699332)
--(axis cs:0.69648865296017,0.858665900564014)
--(axis cs:0.68126517347763,0.865750947428697)
--(axis cs:0.662402573333889,0.87283599429338)
--(axis cs:0.643011927354664,0.879921041158062)
--(axis cs:0.626664099380621,0.887006088022745)
--(axis cs:0.615995236657026,0.894091134887428)
--(axis cs:0.611658225756482,0.90117618175211)
--(axis cs:0.612113220228207,0.908261228616793)
--(axis cs:0.614395272125828,0.915346275481475)
--(axis cs:0.615486923398556,0.922431322346158)
--(axis cs:0.61360590533511,0.929516369210841)
--(axis cs:0.60888927516703,0.936601416075523)
--(axis cs:0.603329130046682,0.943686462940206)
--(axis cs:0.6,0.950771509804888)
--(axis cs:0.601855807273405,0.957856556669571)
--(axis cs:0.610885474885387,0.964941603534254)
--(axis cs:0.628402247021924,0.972026650398936)
--(axis cs:0.655990502620009,0.979111697263619)
--(axis cs:0.695410710161686,0.986196744128302)
--(axis cs:0.746637394261319,0.993281790992984)
--(axis cs:0.80568657852855,1.00036683785767)
--(axis cs:1.19431342147145,1.00036683785767)
--(axis cs:1.19431342147145,1.00036683785767)
--(axis cs:1.25336260573868,0.993281790992984)
--(axis cs:1.30458928983831,0.986196744128302)
--(axis cs:1.34400949737999,0.979111697263619)
--(axis cs:1.37159775297808,0.972026650398936)
--(axis cs:1.38911452511461,0.964941603534254)
--(axis cs:1.39814419272659,0.957856556669571)
--(axis cs:1.4,0.950771509804888)
--(axis cs:1.39667086995332,0.943686462940206)
--(axis cs:1.39111072483297,0.936601416075523)
--(axis cs:1.38639409466489,0.929516369210841)
--(axis cs:1.38451307660144,0.922431322346158)
--(axis cs:1.38560472787417,0.915346275481475)
--(axis cs:1.38788677977179,0.908261228616793)
--(axis cs:1.38834177424352,0.90117618175211)
--(axis cs:1.38400476334297,0.894091134887428)
--(axis cs:1.37333590061938,0.887006088022745)
--(axis cs:1.35698807264534,0.879921041158062)
--(axis cs:1.33759742666611,0.87283599429338)
--(axis cs:1.31873482652237,0.865750947428697)
--(axis cs:1.30351134703983,0.858665900564014)
--(axis cs:1.29348102300689,0.851580853699332)
--(axis cs:1.28836485260653,0.844495806834649)
--(axis cs:1.28663017386276,0.837410759969967)
--(axis cs:1.28639688609915,0.830325713105284)
--(axis cs:1.28606055980394,0.823240666240601)
--(axis cs:1.28446469587506,0.816155619375919)
--(axis cs:1.28084890834191,0.809070572511236)
--(axis cs:1.27478013026074,0.801985525646553)
--(axis cs:1.26609112007854,0.794900478781871)
--(axis cs:1.25482228828475,0.787815431917188)
--(axis cs:1.24121158513604,0.780730385052505)
--(axis cs:1.22572222090742,0.773645338187823)
--(axis cs:1.20901304335262,0.76656029132314)
--(axis cs:1.19180204552175,0.759475244458458)
--(axis cs:1.17470400249778,0.752390197593775)
--(axis cs:1.1581547554241,0.745305150729092)
--(axis cs:1.14242436061523,0.73822010386441)
--(axis cs:1.12764055120002,0.731135056999727)
--(axis cs:1.11381924358634,0.724050010135044)
--(axis cs:1.100979572934,0.716964963270362)
--(axis cs:1.08930307534159,0.709879916405679)
--(axis cs:1.07912291073543,0.702794869540997)
--(axis cs:1.070639865101,0.695709822676314)
--(axis cs:1.06361723877118,0.688624775811631)
--(axis cs:1.05743227498789,0.681539728946949)
--(axis cs:1.0515126561059,0.674454682082266)
--(axis cs:1.04575080561993,0.667369635217584)
--(axis cs:1.04049939547712,0.660284588352901)
--(axis cs:1.03618867723728,0.653199541488218)
--(axis cs:1.03295575069456,0.646114494623536)
--(axis cs:1.03059926406425,0.639029447758853)
--(axis cs:1.02883598717958,0.63194440089417)
--(axis cs:1.02758451571915,0.624859354029488)
--(axis cs:1.02699964658899,0.617774307164805)
--(axis cs:1.02719979471749,0.610689260300122)
--(axis cs:1.02792354324098,0.60360421343544)
--(axis cs:1.02847077255928,0.596519166570757)
--(axis cs:1.02805704257137,0.589434119706075)
--(axis cs:1.02631010129,0.582349072841392)
--(axis cs:1.02347809387525,0.575264025976709)
--(axis cs:1.02018582298973,0.568178979112027)
--(axis cs:1.01699523627187,0.561093932247344)
--(axis cs:1.01417139066401,0.554008885382661)
--(axis cs:1.01179550470341,0.546923838517979)
--(axis cs:1.00999122280728,0.539838791653296)
--(axis cs:1.00894286384428,0.532753744788613)
--(axis cs:1.00866996482475,0.525668697923931)
--(axis cs:1.00884335735143,0.518583651059248)
--(axis cs:1.00890936364602,0.511498604194566)
--(axis cs:1.00845509672657,0.504413557329883)
--(axis cs:1.00748467671792,0.4973285104652)
--(axis cs:1.00637161051349,0.490243463600518)
--(axis cs:1.005574778398,0.483158416735835)
--(axis cs:1.00538763668714,0.476073369871153)
--(axis cs:1.00587490822374,0.46898832300647)
--(axis cs:1.00693442272912,0.461903276141787)
--(axis cs:1.00835657445074,0.454818229277105)
--(axis cs:1.00984918387146,0.447733182412422)
--(axis cs:1.01108085956642,0.440648135547739)
--(axis cs:1.01177441253927,0.433563088683057)
--(axis cs:1.01181162707672,0.426478041818374)
--(axis cs:1.01127991081033,0.419392994953692)
--(axis cs:1.0104181187775,0.412307948089009)
--(axis cs:1.00947871254485,0.405222901224326)
--(axis cs:1.00858738099587,0.398137854359644)
--(axis cs:1.00770031003943,0.391052807494961)
--(axis cs:1.00669210772996,0.383967760630278)
--(axis cs:1.00549873326489,0.376882713765596)
--(axis cs:1.00419801598307,0.369797666900913)
--(axis cs:1.0029752543642,0.362712620036231)
--(axis cs:1.00201737343778,0.355627573171548)
--(axis cs:1.0014179976483,0.348542526306865)
--(axis cs:1.00114779087937,0.341457479442183)
--(axis cs:1.00109332902729,0.3343724325775)
--(axis cs:1.00112842220766,0.327287385712817)
--(axis cs:1.00116813503132,0.320202338848135)
--(axis cs:1.00117558016062,0.313117291983452)
--(axis cs:1.00113496009383,0.30603224511877)
--(axis cs:1.00103159717138,0.298947198254087)
--cycle;

\addplot [line width=0.75pt, darkslategray63]
table {%
0 0.588609331567708
0 0.999498237398039
};
\addplot [line width=2.25pt, darkslategray63]
table {%
0 0.814156106440217
0 0.967328517828023
};
\addplot [semithick, darkslategray63, mark=-, mark size=2.34375, mark options={solid,fill=white}]
table {%
0 0.90554075800455
};
\addplot [line width=0.75pt, darkslategray63]
table {%
1 0.588609331567708
1 1.00036683785767
};
\addplot [line width=2.25pt, darkslategray63]
table {%
1 0.799655830262293
1 0.94130212217349
};
\addplot [semithick, darkslategray63, mark=-, mark size=2.34375, mark options={solid,fill=white}]
table {%
1 0.879008743541213
};
\end{groupplot}

\end{tikzpicture}